\documentclass[aps,prx,twocolumn,preprintnumbers,superscriptaddress,amsmath,amssymb,nofootinbib,nobalancelastpage]{revtex4-2}
\pdfoutput=1
\usepackage{graphicx, xcolor}
\graphicspath{{figures/}}
\usepackage{amsfonts,amsthm,microtype,mathrsfs,bbm}
\usepackage[colorlinks,allcolors=blue]{hyperref}
\usepackage{braket,mathtools,physics,enumerate}
\usepackage[capitalize]{cleveref}  
\usepackage{cases}
\usepackage{MnSymbol}
\usepackage{thm-restate}
\usepackage{upgreek}
\newtheorem{theorem}{Theorem}
\newtheorem*{theorem*}{Theorem}
\newtheorem{definition}{Definition}
\newtheorem{proposition}{Proposition}
\newtheorem{lemma}{Lemma}
\newtheorem*{lemma*}{Lemma}
\newtheorem{corollary}{Corollary}
\newtheorem{assumption}{Assumption}

\newtheorem*{manualtheorem}{Theorem 1'}

\theoremstyle{definition}

\renewcommand{\ket}[1]{|#1\rangle}
\renewcommand{\bra}[1]{\langle#1|}
\newcommand{\1}{\mathbbm{1}}
\def\-{\raisebox{0 pt}{-}}

\DeclareMathOperator{\diag}{diag}
\DeclareMathOperator{\poly}{poly}
\DeclareMathOperator{\polylog}{polylog}
\DeclareMathOperator{\image}{Im}

\newcommand{\prlsection}[1]{{\em {#1}.---~}}

\usepackage{tikz}
\usetikzlibrary{decorations.pathreplacing,calligraphy,decorations.markings}
\usetikzlibrary{fadings}
\usetikzlibrary{arrows.meta}

\definecolor{tensorcolor}{rgb}{0.65,0.77,0.95}
\definecolor{btensorcolor}{rgb}{0.85,0.72,0.89}
\definecolor{whitetensorcolor}{rgb}{0.93,0.93,0.93}

\newcommand{\gate}[2]{
    \begin{scope}[shift={(#1)}]
        \draw[ thick, fill=tensorcolor, rounded corners=1pt] (-#2,-0.25) rectangle (#2,0.25);
    \end{scope}
        }

\newcommand{\gatebg}[2]{
    \begin{scope}[shift={(#1)}]
        \draw[ thick, fill=tensorcolor, rounded corners=1pt] (-#2,-0.5) rectangle (#2,0.5);
    \end{scope}
        }

\newcommand{\MPSTensor}[5]{
	\begin{scope}[shift={(#1)}]
    \ifnum#5=0
		\draw[thick] (-#2,0) -- (#2,0);
		\draw[thick] (0,#2) -- (0,0);
    \fi
    \ifnum#5=-1
		\draw[thick] (0,0) -- (#2,0);
		\draw[thick] (0,#2) -- (0,0);
    \fi
    \ifnum#5=1
		\draw[thick] (-#2,0) -- (0,0);
		\draw[thick] (0,#2) -- (0,0);
    \fi
        \draw[ thick, fill=tensorcolor, rounded corners=1pt] (-#3,-#3) rectangle (#3,#3);
		\draw (0,0) node {\scriptsize #4};
	\end{scope}
}

\newcommand{\GTensor}[5]{
	\begin{scope}[shift={(#1)}]
    \ifnum#5=0
		\draw[thick] (-#2,0) -- (#2,0);
		\draw[thick] (0,#2) -- (0,-#2);
    \fi
    \ifnum#5=-1
		\draw[thick] (0,0) -- (#2,0);
		\draw[thick] (0,#2) -- (0,-#2);
    \fi
    \ifnum#5=1
		\draw[thick] (-#2,0) -- (0,0);
		\draw[thick] (0,#2) -- (0,-#2);
    \fi
        \draw[ thick, fill=tensorcolor, rounded corners=1pt] (-#3,-#3) rectangle (#3,#3);
		\draw (0,0) node {\scriptsize #4};
	\end{scope}
}

\newcommand{\GVecTensor}[5]{
	\begin{scope}[shift={(#1)}]
    \ifnum#5=0
		\draw[thick] (-#2,0) -- (#2,0);
    \fi
    \ifnum#5=-1
		\draw[thick] (0,0) -- (#2,0);
    \fi
    \ifnum#5=1
		\draw[thick] (-#2,0) -- (0,0);
    \fi
    	\draw[thick] (-0.35*#3,0) -- (-0.35*#3,#2);
        \draw[thick] (0.35*#3,0) -- (0.35*#3,#2);
        \draw[ thick, fill=tensorcolor, rounded corners=1pt] (-#3,-#3) rectangle (#3,#3);
		\draw (0,0) node {\scriptsize #4};
	\end{scope}
}

\newcommand{\GBraTensor}[5]{
	\begin{scope}[shift={(#1)}]
    \ifnum#5=0
		\draw[thick] (-#2,0) -- (#2,0);
    \fi
    \ifnum#5=-1
		\draw[thick] (0,0) -- (#2,0);
    \fi
    \ifnum#5=1
		\draw[thick] (-#2,0) -- (0,0);
    \fi
    	\draw[thick] (-0.35*#3,0) -- (-0.35*#3,-#2);
        \draw[thick] (0.35*#3,0) -- (0.35*#3,-#2);
        \draw[ thick, fill=tensorcolor, rounded corners=1pt] (-#3,-#3) rectangle (#3,#3);
		\draw (0,0) node {\scriptsize #4};
	\end{scope}
}

\newcommand{\GDTensor}[5]{
	\begin{scope}[shift={(#1)}]
    \ifnum#5=0
		\draw[thick] (-#2,0) -- (#2,0);
		\draw[thick] (0,#2) -- (0,-#2);
    \fi
    \ifnum#5=-1
		\draw[thick] (0,0) -- (#2,0);
		\draw[thick] (0,#2) -- (0,-#2);
    \fi
    \ifnum#5=1
		\draw[thick] (-#2,0) -- (0,0);
		\draw[thick] (0,#2) -- (0,-#2);
    \fi
        \draw[ thick, fill=tensorcolor, rounded corners=1pt] (-#3,-#3) rectangle (#3,#3);
    \def\dx{#3/3};
	\draw (0,0) node {\scriptsize #4};
	\end{scope}
}

\newcommand{\ATensor}[3]{
    \GTensor{#1}{1}{.5}{#2}{#3};
}

\newcommand{\DoubleATensor}[2]{
	\begin{scope}[shift={(#1)}]
        \ATensor{(0,.8)}{}{#2};
        \ATensor{(0,-.8)}{}{#2};
	\end{scope}
}

\newcommand{\ETensor}[2]{
	\begin{scope}[shift={(#1)}]
        \GDTensor{(0,.8)}{1}{.5}{}{#2};
        \GTensor{(0,-.8)}{1}{.5}{}{#2};
    \def\dx{.75};
	\draw [thick] (0,-1.8) to  [bend left=90] (-\dx,-1.8);
	\draw [thick] (0,1.8) to  [bend right=90] (-\dx,1.8);
	\draw [thick] (-\dx,1.8) to  (-\dx,-1.8);
	\end{scope}
}

\newcommand{\bTensor}[2]{
	\begin{scope}[shift={(#1)}]
	    \draw [thick] (-1,0) to  (1,0);
		\filldraw[color=black, fill=btensorcolor, thick] (0,0) circle (\stradius);
	\draw (0,0) node {#2};
	\end{scope}
}

\newcommand{\BTensor}[4]{
	\begin{scope}[shift={(#1)}]
        \ifnum#2=-1
            \bTensor{(0,0.9)}{#4};
            \bTensor{(0,-0.9)}{#3};
            \draw [thick] (-1,0.9) to  [bend left=45] (-1.3,1.2);
            \draw [thick] (-1.3,1.2) to  [bend left=45] (-1,1.5);
            \draw [thick] (-1, 1.5) to (-0.6, 1.5);
            \draw [thick] (-1,-0.9) to  [bend right=45] (-1.3,-1.2);
            \draw [thick] (-1.3,-1.2) to  [bend right=45] (-1,-1.5);
            \draw [thick] (-1, -1.5) to (-0.6, -1.5);
        \fi
        \ifnum#2=1
            \draw [thick] (-0.8,0.9) to  [bend right=45] (-0.5,1.2);
            \draw [thick] (-0.5,1.2) to  [bend right=45] (-0.8,1.5);
            \draw [thick] (-0.8, 1.5) to (-1.0, 1.5);
                        \draw [thick] (-0.8,-0.9) to  [bend left=45] (-0.5,-1.2);
            \draw [thick] (-0.5,-1.2) to  [bend left=45] (-0.8,-1.5);
            \draw [thick] (-0.8, -1.5) to (-1.0, -1.5);

        \fi
	\end{scope}
}

\newcommand{\SingleDots}[2]{
	\begin{scope}[shift={(#1)}]
      \draw [thick,  dash pattern=on 1pt off 1.4pt] (-#2*0.6,0) to (#2*0.6,0);
	\end{scope}
}

\newcommand{\DoubleDots}[2]{
	\begin{scope}[shift={(#1)}]
      \SingleDots{0,0.8}{#2};
      \SingleDots{0,-0.8}{#2};
	\end{scope}
}

\newcommand{\DoubleLongLine}[1]{
	\begin{scope}[shift={(#1)}]
        \def\dx{.75};
	    \draw [thick] (0,-1.8) to  [bend left=90] (0-\dx,-1.8);
	    \draw [thick] (0,3.2) to  [bend right=90] (0-\dx,3.2);
	    \draw [thick] (-\dx,3.2) to  (-\dx,-1.8);
	    \draw [thick] (0,3.2) to  (0,2.8);
    \end{scope}
}

\newcommand{\IdentityTensor}[3]{
	\begin{scope}[shift={(#1)}]
      \draw [thick] (0,-1) to (0,1);
    \ifnum#3=1
        \draw [thick] (0,-1) to (0,1);
	    \filldraw[color=black, fill=whitetensorcolor, thick] (0,0) circle (\stradius);
	    \draw (0,0) node {#2};
    \fi
	\end{scope}
}

\newcommand{\DoubleIdentityTensor}[3]{
	\begin{scope}[shift={(#1)}]
        \draw [thick] (0,-1.8) to (0,1.8);
    \ifnum#3=1
        \draw [thick] (0,-1.8) to (0,1.8);
	    \filldraw[color=black, fill=whitetensorcolor, thick] (0,0) circle (\stradius);
	    \draw (0,0) node {#2};
    \fi
	\end{scope}
}

\newcommand{\SideIdentityTensor}[4]{
	\begin{scope}[shift={(#1)}]
    \ifnum#4=-1
	   \draw [thick] (\doubledx-1,0.8) to  [bend right=90] (\doubledx-1,-0.8);
    \fi
    \ifnum#4=-2
	   \draw [thick] (\doubledx-1,0.8) to  [bend right=90] (\doubledx-1,-0.8);
      \draw [thick] (\doubledx-1,0.8) -- (\doubledx-0.5,0.8);
      \draw [thick] (\doubledx-1,-0.8) -- (\doubledx-0.5,-0.8);
    \fi
    \ifnum#4=-3
	   \draw [thick] (\doubledx-1,0.8) to  [bend right=90] (\doubledx-1,-0.8);
      \draw [thick] (\doubledx-1,0.8) -- (\doubledx-0.5,0.8);
      \draw [thick] (\doubledx-1,-0.8) -- (\doubledx-0.5,-0.8);
	\filldraw[color=black, fill=whitetensorcolor, thick] (\doubledx-1.4,0) circle (#3);
	\draw (\doubledx-1.4,0) node {#2};
    \fi
    \ifnum#4=1
	   \draw [thick] (-\doubledx+1,0.8) to  [bend left=90] (-\doubledx+1,-0.8);
    \fi
    \ifnum#4=2
	   \draw [thick] (-\doubledx+1,0.8) to  [bend left=90] (-\doubledx+1,-0.8);
      \draw [thick] (-\doubledx+1,0.8) -- (-\doubledx+0.5,0.8);
      \draw [thick] (-\doubledx+1,-0.8) -- (-\doubledx+0.5,-0.8);
    \fi
    \ifnum#4=3
	   \draw [thick] (-\doubledx+1,0.8) to  [bend left=90] (-\doubledx+1,-0.8);
      \draw [thick] (-\doubledx+1,0.8) -- (-\doubledx+0.5,0.8);
      \draw [thick] (-\doubledx+1,-0.8) -- (-\doubledx+0.5,-0.8);
	\filldraw[color=black, fill=whitetensorcolor, thick] (-\doubledx+1.4,0) circle (#3);
	\draw (-\doubledx+1.4,0) node {#2};
    \fi
\end{scope}
}

\newcommand{\SingleTrLeft}[1]{
	\begin{scope}[shift={(#1)}]
      \draw [thick] (0,0) to (\doubledx-0.8,0);
	   \draw [thick] (0,0) to  [bend left=90] (0,0.8);
	   \draw [thick, dotted] (0,0.8) to  (1,0.8);
	\end{scope}
}

\newcommand{\SingleTrRight}[1]{
	\begin{scope}[shift={(#1)}, xscale=-1]
	   \SingleTrLeft{(0,0)};
	\end{scope}
}

\newcommand{\myarrow}[2]{
	\begin{scope}[shift={(#1)}]
    \ifnum#2=1
\draw[-{Stealth[length=1mm, width=2.3mm]}] (0,0.0) -- (0,0.03);
    \fi
    \ifnum#2=2
\draw[-{Stealth[length=1mm, width=2.3mm]}] (0,0) -- (0,-0.03);
    \fi
    \ifnum#2=3
\draw[-{Stealth[length=1mm, width=2.3mm]}] (0,0) -- (0.03,0);
    \fi
    \ifnum#2=4
\draw[-{Stealth[length=1mm, width=2.3mm]}] (0,0) -- (-0.03,0);
    \fi
\end{scope}
}

\newcommand{\mysmallarrow}[2]{%
  \draw[gray!60, thick, -{Stealth[length=1.5mm, width=1.2mm]}] #1 -- #2;
}

\newcommand{\SideIdentityTensorRT}[4]{
	\begin{scope}[shift={(#1)}]
    \ifnum#4=-1
	   \draw [very thick] (\doubledx-1,0.8) to  [bend right=90] (\doubledx-1,-0.8);
    \fi
    \ifnum#4=-2
	   \draw [very thick] (\doubledx-1,0.8) to  [bend right=90] (\doubledx-1,-0.8);
      \draw [very thick] (\doubledx-1,0.8) -- (\doubledx-0.5,0.8);
      \draw [very thick] (\doubledx-1,-0.8) -- (\doubledx-0.5,-0.8);
    \fi
    \ifnum#4=-3
	   \draw [very thick] (\doubledx-1,0.9) to  [bend right=90] (\doubledx-1,-0.9);
      \draw [very thick] (\doubledx-1,0.9) -- (\doubledx-0.5,0.9);
      \draw [very thick] (\doubledx-1,-0.9) -- (\doubledx-0.5,-0.9);
	\filldraw[color=black, fill=whitetensorcolor, thick] (\doubledx-1.4,0) circle (#3);
	\draw (\doubledx-1.4,0) node {#2};
    \fi
    \ifnum#4=1
	   \draw [very thick] (-\doubledx+1,0.8) to  [bend left=90] (-\doubledx+1,-0.8);
    \fi
    \ifnum#4=2
	   \draw [very thick] (-\doubledx+1,0.8) to  [bend left=90] (-\doubledx+1,-0.8);
      \draw [very thick] (-\doubledx+1,0.8) -- (-\doubledx+0.5,0.8);
      \draw [very thick] (-\doubledx+1,-0.8) -- (-\doubledx+0.5,-0.8);
    \fi
    \ifnum#4=3
	   \draw [very thick] (-\doubledx+1,0.9) to  [bend left=90] (-\doubledx+1,-0.9);
      \draw [very thick] (-\doubledx+1,0.9) -- (-\doubledx+0.5,0.9);
      \draw [very thick] (-\doubledx+1,-0.9) -- (-\doubledx+0.5,-0.9);
	\filldraw[color=black, fill=whitetensorcolor, thick] (-\doubledx+1.4,0) circle (#3);
	\draw (-\doubledx+1.4,0) node {#2};
    \fi
\end{scope}
}

\newcommand\doubledx{1.6}
\newcommand\singledx{1.8}
\newcommand\identitydx{1}
\newcommand\stradius{0.5}

\newcommand\subsetsim{\mathrel{%
  \ooalign{\raise0.2ex\hbox{$\subset$}\cr\hidewidth\raise-0.8ex\hbox{\scalebox{0.9}{$\sim$}}\hidewidth\cr}}}

\begin{document}

\title{Quantum Circuits for Matrix-Product Unitaries}

\author{Georgios Styliaris}
\thanks{Authors with equal contribution.}
\affiliation{Max Planck Institute of Quantum Optics, Hans-Kopfermann-Str. 1, Garching 85748, Germany}
\affiliation{Munich Center for Quantum Science and Technology (MCQST), Schellingstr. 4, 80799 M{\"{u}}nchen, Germany}

\author{Rahul Trivedi}
\thanks{Authors with equal contribution.}
\affiliation{Max Planck Institute of Quantum Optics, Hans-Kopfermann-Str. 1, Garching 85748, Germany}
\affiliation{Munich Center for Quantum Science and Technology (MCQST), Schellingstr. 4, 80799 M{\"{u}}nchen, Germany}

\author{J.~Ignacio Cirac}
\affiliation{Max Planck Institute of Quantum Optics, Hans-Kopfermann-Str. 1, Garching 85748, Germany}
\affiliation{Munich Center for Quantum Science and Technology (MCQST), Schellingstr. 4, 80799 M{\"{u}}nchen, Germany}

\date{\today}

\begin{abstract}
Matrix-product unitaries (MPUs) are many-body unitary operators that, as a consequence of their tensor-network structure, preserve the entanglement area law in 1D systems. However, it is unknown how to implement an MPU as a quantum circuit since the individual tensors describing the MPU are not unitary. In this Letter, we show that a large class of MPUs can be implemented with a polynomial-depth quantum circuit. For an $N$-site MPU built from a repeated bulk tensor with open boundary, we explicitly construct a quantum circuit of polynomial depth $T = O(N^{\alpha})$ realizing the MPU, where the constant $\alpha$ depends only on the bulk and boundary tensor and not the system size $N$. We show that this class includes nontrivial unitaries that generate long-range entanglement and, in particular, contains a large class of unitaries constructed from representations of $C^*$-weak Hopf algebras. Furthermore, we also adapt our construction to nonuniform translationally-varying MPUs and show that they can be implemented by a circuit of depth $O(N^{\beta} \, \mathrm{poly}\, D)$ where $\beta \le  1 + \log_2 \sqrt{D}/ s_{\min}$, with $D$ being the bond dimension and $s_{\min}$ the smallest nonzero Schmidt value of the normalized Choi state corresponding to the MPU.
\end{abstract}

\maketitle

\prlsection{Introduction} The entanglement area law is a hallmark of many-body quantum systems with local interactions~\cite{eisert2010colloquium}. It captures the observation that in many physical quantum states, only degrees of freedom near the boundary of a region contribute significantly to its entanglement with the rest~\cite{verstraete2006matrix,hastings2007area,wolf2008area,brandao2015exponential,anshu2022area}. On the level of quantum operations, it is thus a fundamental task to characterize unitaries that \emph{preserve} the area law. Perhaps the most natural class with this property are unitaries admitting a tensor-network representation, which inherently ensures the preservation of the area law~\cite{cirac2021matrix}.

In one spatial dimension, the corresponding tensor-network class is that of matrix-product unitaries (MPUs)~\cite{cirac2021matrix}. In addition to describing discrete time evolution, MPUs capture families of exactly solvable quantum circuits~\cite{wang2025hopf}, describe unitary symmetries of 1D systems~\cite{garre2025fractional,franco2025symmetry}, and have been useful to characterize Floquet phases in 2D via their boundary dynamics~\cite{po2016chiral}. A well-understood subclass of MPUs is one-dimensional quantum cellular automata (QCA)~\cite{farrelly2020review}, which coincide with MPUs formed by a single repeated tensor~\cite{cirac2017matrix1,sahinoglu2018matrix,shukla2025simple}. This subfamily, which includes finite-depth quantum circuits, is characterized by a strict causal cone and thus preserves the phase of matter of the input state. This is, however, not a universal property: certain MPUs, such as sequential quantum circuits~\cite{chen2024sequential}, fermionic MPUs~\cite{piroli2021fermionic}, or MPUs with open boundary~\cite{styliaris2025matrix}, can change the underlying phase as they can generate long-range entanglement.

Although the tensor-network form of MPUs has the advantage of automatically enforcing the preservation of the area law, the individual tensors do not inherently correspond to unitary or other physical processes. This poses a challenge for translating the tensor-network description into implementable quantum circuits. While it is well known that every matrix-product state (MPS) can be prepared by a sequential (linear-depth) circuit acting on a product state~\cite{schon2005sequential}, no general framework exists for efficiently implementing MPUs as circuits beyond the limited case of QCA~\cite{cirac2017matrix1,sahinoglu2018matrix}.  The present Letter addresses this gap (Fig.~\ref{fig_gen}).

\begin{figure}[t]
\centering
\includegraphics[scale=0.15]{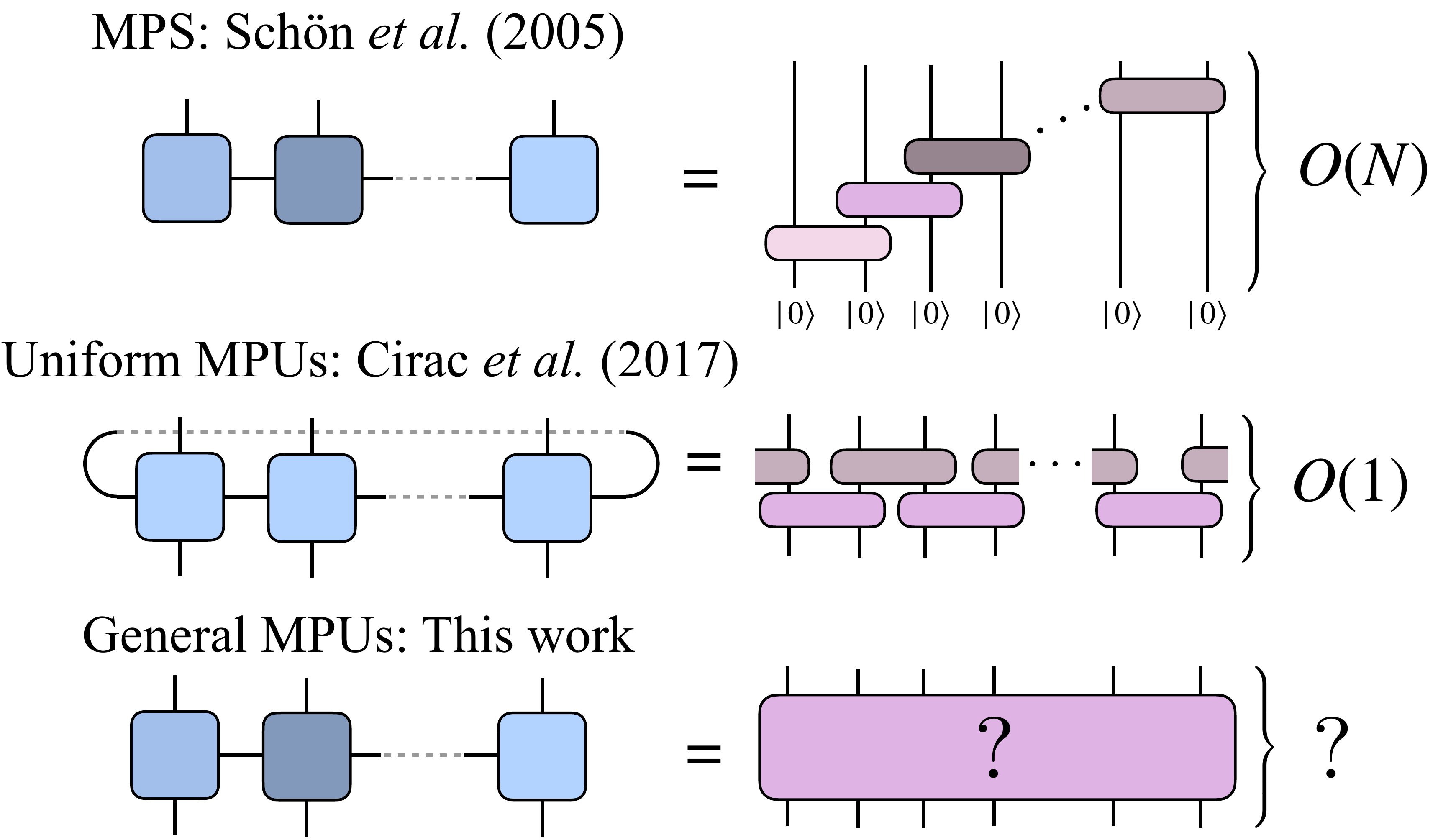}
    \caption{MPS over $N$ sites admit a sequential circuit decomposition with depth $O(N \poly D)$, where $D$ is the bond dimension~\cite{schon2005sequential}. How can MPUs be implemented as quantum circuits, and what is the corresponding depth? The answer is only known for the case of MPUs with uniform bulk and periodic boundary, which correspond to QCA~\cite{cirac2017matrix1}. Here, we introduce a circuit decomposition of $\poly(N)$ depth for arbitrary MPUs, when they are well conditioned. We show this holds (i) for MPUs with uniform bulk and open boundary, and (ii) for arbitrary MPUs when $\sqrt{D}/s_{\min} = O(1)$ ($s_{\min}:$ smallest nonzero singular value of the MPU Choi state).}
    \label{fig_gen}
\end{figure}

The problem of implementing an MPU as a quantum circuit is fundamentally more challenging than MPS preparation, which is far better understood~\cite{schon2005sequential,ge2016rapid,malz2024preparation,smith2024constant,sahay2025classifying,stephen2024preparing,zhang2024characterizing}. Unlike MPSs, only a subclass of MPUs can be implemented by sequential circuits~\cite{styliaris2025matrix}. This may at first seem surprising, given that sequential circuits are sufficient to map any given initial MPS to any other target MPS. However, certain MPUs, such as the multi-controlled NOT gate, can provably not be implemented by sequential circuits since they noncausally propagate correlations~\cite{styliaris2025matrix}.

An arguably obvious strategy to implement an MPU by exploiting its tensor network structure would be to start from its Choi state~\cite{nielsen2002quantum}, which is automatically an MPS and thus efficiently preparable.
However, while such a Choi MPS can be physically prepared with a linear-depth circuit~\cite{schon2005sequential}, this does not necessarily imply that the associated unitary can be efficiently applied to an unknown input state. Given $2N$ qubits initialized in the Choi state of the $N$-qubit unitary, the only general-purpose strategy to apply the unitary on an input state is via $N$ joint measurements (e.g., Bell measurements) between the input state and the Choi state~\cite{gottesman1999demonstrating,cirac2001entangling,piroli2021quantum} followed by a post-selection on the $N$ measurement outcomes. However, such a strategy will typically incur an exponential cost due to the post-selection, rendering it inefficient.

Alternatively, one could consider formulating this question within the oracle model. Since the MPU structure guarantees that individual matrix elements of the unitary can be efficiently classically computed~\cite{cirac2021matrix}, a quantum oracle accessing the MPU elements can be implemented in polynomial time. However, there is no general-purpose strategy to, given such an oracle for the unitary elements, apply the unitary onto an input state. Using state-preparation strategies given oracles for computing the vector elements of the state~\cite{aaronson2016complexity} can, again, only allow us to prepare the Choi state of such a unitary efficiently. Consequently, to obtain an efficient circuit implementation of an MPU, we need to identify and exploit structure beyond the fact that it has an MPS Choi state or efficiently computable matrix elements.

In this Letter, we present a systematic and efficient quantum circuit decomposition of MPUs starting from their tensor-network representation. The resulting circuit, under the assumption that the bond dimension cannot be further reduced, has depth $O(\poly N)$ for $N$-site MPUs consisting of a repeated bulk tensor and open boundary. This class goes beyond QCA -- moreover, we provide a systematic construction of non-trivial MPUs arising from $C^*$-weak Hopf algebras~\cite{bohm1996coassociative,molnar2022matrix,ruiz2024matrix} which are in this class. Furthermore, we consider non translation-invariant MPUs. Our algorithm gives a circuit decomposition with polynomial depth, provided a spectral condition is satisfied.

\prlsection{MPUs with uniform bulk}MPUs are defined via rank-4 tensors, graphically represented as
\begin{align}
    A^{ij}_{mn} =
        \begin{array}{c}
        \begin{tikzpicture}[scale=0.5,baseline={([yshift=-0.65ex] current bounding box.center)}]
            \ATensor{0,0}{{\small$A$}}{0}
		\draw (-1.4,0) node {$m$};
		\draw (1.4,0) node {$n$};
		\draw (0,1.4) node {$i$};
		\draw (0,-1.4) node {$j$};
        \end{tikzpicture}
        \end{array}
    = \Bigg(
        \begin{array}{c}
        \begin{tikzpicture}[scale=0.5,baseline={([yshift=-0.65ex] current bounding box.center)}]
            \ATensor{0,0}{\small{$A^\dagger$}}{0};
		\draw (-1.4,0) node {$m$};
		\draw (1.4,0) node {$n$};
		\draw (0,1.4) node {$j$};
		\draw (0,-1.4) node {$i$};
        \end{tikzpicture}
        \end{array}
        \Bigg) ^*
    \; \in \mathbb C \;,
\end{align}
where the vertical legs $i,j = 1,\dots,d$ label the physical space and the horizontal ones $m,n = 1,\dots, D$ the auxiliary space, with $D$ the bond dimension. We will use arrows to indicate the input/output space of operators, e.g.,
\[      
        \begin{array}{c}
        \begin{tikzpicture}[scale=0.5,baseline={([yshift=-0.65ex] current bounding box.center)}]
            \ATensor{0,0}{{\small$A$}}{0}
            \myarrow{0,-.85}{2};
            \myarrow{0,.85}{1};
            \myarrow{-.65,0}{3};
            \myarrow{.65,0}{4};
        \end{tikzpicture}
        \end{array}
        =
        \sum A^{ij}_{mn} \ket{ij}\bra{mn} \;.
\]

We first consider MPUs with uniform bulk and open boundary.
\begin{definition}[Uniform-bulk MPU]
    The tensor $A$, together with the boundary vectors $l$ and $r$, define a \emph{uniform-bulk MPU} if the $N$-site matrix-product operator (MPO)
\begin{align}
 U_N =
        \begin{array}{c}
        \begin{tikzpicture}[scale=0.5]
		      \foreach \x in {2,...,2}{
                \GTensor{(-\singledx*\x,0)}{1}{.5}{\small $A$}{0}
            \myarrow{-\singledx*\x,-.65}{1};
            \myarrow{-\singledx*\x,.85}{1};
        }
		      \foreach \x in {4,...,4}{
                \GTensor{(-\singledx*\x,0)}{1}{.5}{\small $A$}{0}
            \myarrow{-\singledx*\x,-.65}{1};
            \myarrow{-\singledx*\x,.85}{1};
        }
		      \foreach \x in {3,...,3}{
                \SingleDots{-\singledx*\x,0}{\singledx}
                \draw (-\singledx*\x,-.5) node {\small $N$};
        }
		\foreach \x in {5,...,5}{
		      \filldraw[color=black, fill=whitetensorcolor, thick] (-\singledx*\x+0.3,0) circle (0.5);
            \draw (-\singledx*\x+0.3,0) node {\small$l$};
        }
		\foreach \x in {1,...,1}{
		      \filldraw[color=black, fill=whitetensorcolor, thick] (-\singledx*\x-.3,0) circle (0.5);
            \draw (-\singledx*\x-.3,0) node {\small$r$};
        } 
        \end{tikzpicture}
                \end{array}
\end{align}
is unitary for all $N$. 
\end{definition}
\noindent This open-boundary definition also includes uniform-bulk MPUs with a boundary of the form
\begin{align} \label{eq:U_N_hom}
U_N = 
    \begin{array}{c}
        \begin{tikzpicture}[scale=.5,baseline={([yshift=-0.75ex] current bounding box.center)}]
		      \foreach \x in {0,...,0}{
                \SingleTrRight{(0,0)}
        }
		      \foreach \x in {1,...,1}{
                \GTensor{(-\singledx*\x,0)}{1}{.5}{\small $A$}{0}
                \myarrow{-\singledx*\x,-.65}{1};
                \myarrow{-\singledx*\x,.85}{1};
        }
		      \foreach \x in {2,...,2}{
                \SingleDots{-\singledx*\x,0}{\singledx}
                \draw (-\singledx*\x,-.5) node {\small $N$};
        }
		      \foreach \x in {3,...,3}{
                \GTensor{(-\singledx*\x,0)}{1}{.5}{\small $A$}{0}
                \myarrow{-\singledx*\x,-.65}{1};
                \myarrow{-\singledx*\x,.85}{1};
        }
		      \foreach \x in {4,...,4}{
                \bTensor{-\singledx*\x,0}{$b$}
        }
		      \foreach \x in {5,...,5}{
                \SingleTrLeft{(-\singledx*\x+.3,0)}
        }
        \end{tikzpicture}
        \end{array}
    \;,
\end{align}
studied in Ref.~\cite{styliaris2025matrix}, by suitably enlarging the bond dimension.

The class of uniform-bulk MPUs includes all (translation-invariant) 1D QCA, that is, unitaries possessing a strict light cone. This is because QCA can be represented as uniform-bulk MPUs with periodic boundary~\cite{cirac2017matrix1,sahinoglu2018matrix} ($b = \1$ above) and thus they also fall in the open boundary setting considered here by increasing the bond dimension from $D$ to $D^2$.  Beyond QCA, uniform-bulk MPUs encompass multi-qudit control operations with a single target qudit~\cite{styliaris2025matrix}, such as the multi-control $Z$-gate $U_{C\dots CZ} = \1 - 2 (\ket{0} \bra{0})^{\otimes N}$. Unlike QCA, however, such control unitaries can create long-range correlations through a single application. This directly implies that a uniform-bulk representation of long-range entangling unitaries necessitates open boundary conditions~\cite{styliaris2025matrix}.

A further example of uniform-bulk MPUs, capable of creating long-range correlations, is provided by \emph{MPO-projector unitaries} $U = \sum_{j=1}^\ell e^{i \phi_j} P_j$. Here $P_j$ form a complete set of orthogonal MPO projectors, with a uniform bulk, and $\ell$ is an $N-$independent constant. Such projectors have a nontrivial structure and arise naturally in the context of topologically ordered states and anyons~\cite{bultinck2017anyons}. They also contain long-range entangling unitaries. In our setting, this class is useful because it provides nontrivial MPUs specified directly at the tensor-network level rather than through a circuit description. It can be viewed as a broad generalization of the $U_{C\dots CZ}$ gate, in that the resulting unitaries act by applying distinct phase factors to the subspaces defined by the projectors. In Supplemental Material~\cite{sm}, we present a systematic construction of unitaries arising from the input of a $C^*$-weak Hopf algebra and an associated representation, based on the techniques developed in Refs.~\cite{molnar2022matrix,ruiz2024matrix}. This MPU class, which includes and generalizes the MPO-projector unitaries, lacks a known quantum circuit implementation.

\prlsection{Main result}We seek an efficient quantum circuit decomposition scheme for MPUs, that is, with the number of gates scaling at most polynomially with system size. We begin with the case of uniform-bulk MPUs, which already capture nontrivial examples. We then extend our construction to general, nonuniform MPUs. Throughout, we allow arbitrary two-qudit gates, taking the physical dimension $d$ to be a constant. For our main result, we need the following assumption:
\begin{assumption}\label{assumption}
It holds that
\begin{align}
 \rank \left(
        \begin{array}{c}
        \begin{tikzpicture}[scale=0.5]
		      \foreach \x in {5,...,5}{
                \GTensor{(-\singledx*\x,0)}{1}{.5}{\small $A$}{0}
            \myarrow{-\singledx*\x,-.85}{2};
            \myarrow{-\singledx*\x,.85}{1};
            \myarrow{-\singledx*\x+0.7,0}{4};
        }
		\foreach \x in {6,...,6}{
		      \filldraw[color=black, fill=whitetensorcolor, thick] (-\singledx*\x+0.4,0) circle (0.5);
            \draw (-\singledx*\x+0.4,0) node {\small$l$};
        }
        \end{tikzpicture}
                \end{array}
    \right) = 
 \rank \left(
        \begin{array}{c}
        \begin{tikzpicture}[scale=0.5]
		\foreach \x in {1,...,1}{
                \GTensor{(-\singledx*\x-.5,0)}{1}{.5}{\small $A$}{0}
            \myarrow{-\singledx*\x-.5,-.85}{2};
            \myarrow{-\singledx*\x-.5,.85}{1};
            \myarrow{-\singledx*\x-1.15,0}{3};
        }
		\foreach \x in {0,...,0}{
		      \filldraw[color=black, fill=whitetensorcolor, thick] (-\singledx*\x-0.9,0) circle (0.5);
            \draw (-\singledx*\x-0.9,0) node {\small$r$};
        } 
        \end{tikzpicture}
                \end{array}
        \right)
    = D
        \;.
\end{align}
\end{assumption}
\noindent Physically, the assumption implies that the bond dimension cannot be further compressed, i.e., the MPU has only non-zero operator Schmidt values at every bipartition. The condition is automatically satisfied if $A$, viewed as an MPS tensor, is injective. More generally, it remains valid if $A$ in \cref{eq:U_N_hom} is in block-diagonal canonical form~\cite{perez_garcia2007matrix}, with linearly independent blocks, and the boundary operator $b$ is full rank~\cite{sm}. This set of assumptions is usually taken as a starting point~\cite{bultinck2017anyons, molnar2022matrix}. This class includes, beyond the trivial case of QCA, the MPO-projector unitaries introduced previously. Our main result is:

\begin{theorem} \label{prop:main_hom}
    Uniform-bulk MPUs satisfying Assumption~\ref{assumption} can be implemented with a quantum circuit of $O(\poly N)$ depth and $O(N)$ auxiliary qudits.
\end{theorem}

We note that, in general, a polynomially scaling circuit depth is necessary to implement the MPUs covered by the Theorem using geometrically local gates and local auxiliary qudits. This is because the family includes unitaries capable of generating long-range correlations, such as the multi-control $Z$-gate, which require at least linear depth.

\prlsection{Tree implementation of MPUs} The key idea behind the algorithm in \cref{prop:main_hom} is to first implement small disjoint local unitaries and then to successively merge them in a tree scheme (see Fig.~\ref{fig_merging}). However, individual MPU tensors do not, in general, correspond to unitary operators. Nevertheless, for a uniform-bulk MPU, any segment of $n$ consecutive bulk tensors can be turned into an \emph{isometry}. More specifically, we define
    \begin{align} \label{eq:def_isometries}
        V_{n}^{[L,R]} \coloneqq 
        \begin{array}{c}
        \begin{tikzpicture}[scale=0.5,baseline={([yshift=-6ex] current bounding box.center)}]
		      \foreach \x in {1,...,1}{
                \GTensor{(-\singledx*\x-.5,0)}{1}{.5}{\small $A$}{1}
            \myarrow{-\singledx*\x-.5,-.65}{1};
            \myarrow{-\singledx*\x-.5,.85}{1};
            \myarrow{-\singledx*\x+1,1.6}{1};
            \draw [thick,rounded corners] (-\singledx*\x,0) -- (-\singledx*\x+1,0) -- (-\singledx*\x+1,1.8);
		      \filldraw[color=black, fill=whitetensorcolor, thick] (-\singledx*\x+1,.8) circle (0.5);
            \draw (-\singledx*\x+.95,0.8) node {\small $R$};
        }
		      \foreach \x in {3,...,3}{
            \draw [thick,rounded corners] (-\singledx*\x,0) -- (-\singledx*\x-1.5,0) -- (-\singledx*\x-1.5,1.8);
                \GTensor{(-\singledx*\x,0)}{1}{.5}{\small $A$}{-1}
            \myarrow{-\singledx*\x,.85}{1};
            \myarrow{-\singledx*\x,-.65}{1};
            \myarrow{-\singledx*\x-1.5,1.6}{1};
		      \filldraw[color=black, fill=whitetensorcolor, thick] (-\singledx*\x-1.5,.8) circle (0.5);
            \draw (-\singledx*\x-1.55,0.8) node {\small $L$};
        }
		      \foreach \x in {2,...,2}{
                \SingleDots{-\doubledx*\x-.5,0}{\doubledx*.8};
                \draw (-\doubledx*\x-.6,-.5) node {\small $n$};
        }
        \end{tikzpicture}
                \end{array}
                \;,
\end{align}
where $L,R$ can be operators or vectors and, by definition, $U_N = V^{[l,r]}_{N}$. In the End Matter, we show explicitly construct tensors $L,R \in \mathcal M \left( \mathbb C^{D} \right)$, independent of $n$, such that $V_{n}^{[L,R]}$ is an isometry for all $n$, i.e.,  $V_{n}^{[L,R]\dagger} V^{[L,R]}_{n} = \1$. Furthermore, the isometry property is retained at the boundary, i.e., $V^{[l,R]}_{n}$, $V^{[L,r]}_{n}$ are also isometries. Moreover, using \cref{assumption} introduced earlier, we can guarantee that $L,R$ can be chosen to be \emph{full-rank} -- the importance of this will be clear shortly. 

\begin{figure}[t]
    \centering
        \begin{gather*}
		\begin{tikzpicture}[scale=.23,thick,baseline={([yshift=-6ex]current bounding box.center)}]
%
%
%
%
\begin{scope}[yscale=1.2]
            \begin{scope}[shift={(0, 0)}]
                \draw[thick,rounded corners]  (0,0) -- (0.9,0) -- (0.9,0.7);
                \draw[thick] (0,0.7) -- (0,-0.7);
                \gate{0,0}{0.5};
            \end{scope}
			\foreach \x in {1,2,...,3}{
            \begin{scope}[shift={(2.4*\x, 0)}]
                \draw[thick,rounded corners] (-0.9,0.7) -- (-0.9,0) -- (0.9,0) -- (0.9,0.7);
                \draw[thick] (0,0.7) -- (0,-0.7);
                \gate{0,0}{0.5};
            \end{scope}
            }
			\foreach \x in {0,1,...,2}{
            \begin{scope}[shift={(2.4*\x+9.6, 0)}]
                \draw[thick,rounded corners] (-0.9,0.7) -- (-0.9,0) -- (0.9,0) -- (0.9,0.7);
                \draw[thick] (0,0.7) -- (0,-0.7);
                \gate{0,0}{0.5};
            \end{scope}
            }
            \begin{scope}[shift={(16.8, 0)}]
                \draw[thick,rounded corners] (-0.9,0.7) -- (-0.9,0) -- (0,0);
                \draw[thick] (0,0.7) -- (0,-0.7);
                \gate{0,0}{0.5};
            \end{scope}
            \begin{scope}[shift={(0, 3)}]
                \draw[thick,rounded corners] (2.5,0) -- (3.3,0) -- (3.3,0.7);
                \draw[thick] (0,0.7) -- (0,-0.7);
                \draw[thick] (2.4,0.7) -- (2.4,-0.7);
                \gate{1.2,0}{1.6};
            \end{scope}
			\foreach \x in {1,2,...,2}{
            \begin{scope}[shift={(4.8*\x, 3)}]
                \draw[thick,rounded corners] (-0.9,0.7) -- (-0.9,0) -- (3.3,0) -- (3.3,0.7);
                \draw[thick] (0,0.7) -- (0,-0.7);
                \draw[thick] (2.4,0.7) -- (2.4,-0.7);
                \gate{1.2,0}{1.6};
            \end{scope}
            }
			\foreach \x in {3}{
            \begin{scope}[shift={(4.8*\x, 3)}]
                \draw[thick,rounded corners] (-0.9,0.7) -- (-0.9,0) -- (2.4,0);
                \draw[thick] (0,0.7) -- (0,-0.7);
                \draw[thick] (2.4,0.7) -- (2.4,-0.7);
                \gate{1.2,0}{1.6};
            \end{scope}
            }
            \begin{scope}[shift={(0, 6)}]
                \draw[thick,rounded corners]  (0,0) -- (8.1,0) -- (8.1,0.7);
                \draw[thick] (0,0.7) -- (0,-0.7);
                \draw[thick] (2.4,0.7) -- (2.4,-0.7);
                \draw[thick] (4.8,0.7) -- (4.8,-0.7);
                \draw[thick] (7.2,0.7) -- (7.2,-0.7);
                \gate{3.6,0}{4};
            \end{scope}
            \begin{scope}[shift={(9.6, 6)}]
                \draw[thick,rounded corners] (-0.9,0.7) -- (-0.9,0) -- (7.5,0);
                \draw[thick] (0,0.7) -- (0,-0.7);
                \draw[thick] (2.4,0.7) -- (2.4,-0.7);
                \draw[thick] (4.8,0.7) -- (4.8,-0.7);
                \draw[thick] (7.2,0.7) -- (7.2,-0.7);
                \gate{3.6,0}{4};
            \end{scope}
            \begin{scope}[shift={(0, 9)}]
                \draw[thick] (0,0.7) -- (0,-0.7);
                \draw[thick] (2.4,0.7) -- (2.4,-0.7);
                \draw[thick] (4.8,0.7) -- (4.8,-0.7);
                \draw[thick] (7.2,0.7) -- (7.2,-0.7);
            \end{scope}
            \begin{scope}[shift={(9.6, 9)}]
                \draw[thick] (0,0.7) -- (0,-0.7);
                \draw[thick] (2.4,0.7) -- (2.4,-0.7);
                \draw[thick] (4.8,0.7) -- (4.8,-0.7);
                \draw[thick] (7.2,0.7) -- (7.2,-0.7);
                \gate{-1.2,0}{8.8};
            \end{scope}
\draw (-3,0) node {\footnotesize$1$};
\draw (-3,3) node {\footnotesize$2$};
\draw (-3,6) node {\footnotesize $3$};
\draw[thick, dash pattern=on 1pt off 1pt] (-3,6+1.1) -- (-3,9-1.1);
\draw (-3,9) node {\footnotesize $\lceil \log_2 N \rceil$};
	\foreach \x in {0,1,...,3}{
        \mysmallarrow{(0+\x*4.8,1)}{(1.2+\x*4.8-0.4,2)};
        \mysmallarrow{(2.4+\x*4.8,1)}{(1.2+\x*4.8+0.4,2)};
    }
	\foreach \x in {0,1,...,1}{
        \mysmallarrow{(1.2+\x*9.6,4)}{(3.6+\x*9.6-1,5)};
        \mysmallarrow{(6+\x*9.6,4)}{(3.6+\x*9.6+1,5)};
    }
	\foreach \x in {0,1,...,1}{
        \mysmallarrow{(3.6,7)}{(8.4-2,8)};
        \mysmallarrow{(13.2,7)}{(8.4+2,8)};
    }
\end{scope}
\begin{scope}[shift={(22.3, 1)}]
\draw (3,-2) node {\scriptsize$\ket{0}$};
\draw[thick] (-0.8,-1.3) -- (-0.8,5.8);
\draw[thick] (0.8,-1.3) -- (0.8,5.8);
\draw[thick] (6-0.8,-1.3) -- (6-0.8,5.8);
\draw[thick] (6+0.8,-1.3) -- (6+0.8,5.8);
\draw[thick] (3,-1.3) -- (3,5.8);
\draw[thick,rounded corners] (0,0) -- (2.4,0) -- (2.4,5.8);
\draw[thick,rounded corners] (0,0) -- (-2.4,0) -- (-2.4,5.8);
\draw[thick,rounded corners] (6,0) -- (3.6,0) -- (3.6,5.8);
    \gatebg{0,0}{2};
    \gatebg{6,0}{2};
\draw[fill=btensorcolor,rounded corners=1pt] (2,0.8) rectangle (4,2.8);
\draw (3,1.8) node {\small $U$};
\draw[ thick, fill=btensorcolor, rounded corners=1pt] (-2.9,3.5) rectangle (8,5);
\draw (3,4.25) node {\small $G^\ell$};
\draw (3,7) node {$\mathrel{\scalebox{2}{$=$}}$};
\begin{scope}[shift={(0,9.5)}]
\draw[thick] (-0.8,-1.3) -- (-0.8,3);
\draw[thick] (0.8,-1.3) -- (0.8,3);
\draw[thick] (6-0.8,-1.3) -- (6-0.8,3);
\draw[thick] (6+0.8,-1.3) -- (6+0.8,3);
\draw[thick] (2.4,2) -- (2.4,3);
\draw[thick] (3.6,2) -- (3.6,3);
\draw[thick] (3,2) -- (3,3);
\draw[thick,rounded corners] (0,0) -- (-2.4,0) -- (-2.4,3);
\draw (3.1,1.2) node {\scriptsize $\ket{000}$};
\gatebg{3,0}{5};
\end{scope}
\end{scope}
\draw[draw=gray!60,dashed] (17.5,8.5) -- (19,14);
\draw[draw=gray!60,dashed] (17.5,2.5) -- (19,-2);
\draw[draw=gray,thick,dash pattern=on 1pt off 1pt,rounded corners=1pt] (8.4,8.5) rectangle (17.5,2.5);
\draw[draw=gray,thick,dash pattern=on 1pt off 1pt,rounded corners=1pt] (19,14) rectangle (31,-2);
        \end{tikzpicture}
        \end{gather*}
    \caption{The MPU implementation algorithm proceeds by first realizing small local isometries and then recursively merging them in a tree-like structure, here indicated by arrows. After $\lceil \log_2 N \rceil$ layers of merging, the global MPU is obtained \emph{(left)}. Each individual merging step is performed deterministically using an amplitude amplification technique \emph{(right)}.}
    \label{fig_merging}
\end{figure}
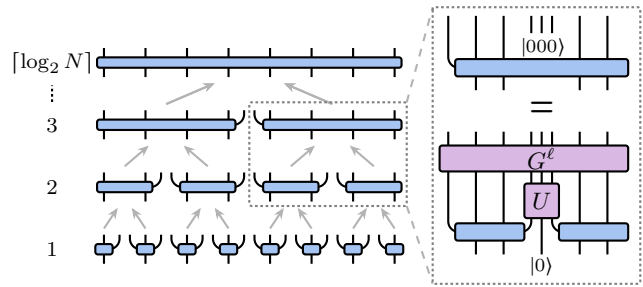

Since the isometry property of $V^{[L,R]}_n$ holds for arbitrary $n$, it naturally motivates a recursive circuit decomposition for uniform-bulk MPUs. In particular, neighboring isometries are successively merged following a tree-like structure (\cref{fig_merging}). To illustrate the idea, suppose, for the moment, that we have an implementation of $V^{[L,R]}_n$ and $V^{[L,R]}_m$ in terms of 2-qudit gates (recall that an isometry corresponds to a unitary with a partially fixed input). Then, the larger isometry $V^{[L,R]}_{n+m}$ can be realized by merging the two via the action of the operator
\begin{align}
    M \coloneqq \ket{00} \bra{ \1 } (R^{-1} \otimes L^{-1}) = 
        \begin{tikzpicture}[scale=0.55,baseline={([yshift=0.0ex] current bounding box.center)}]
		      \foreach \x in {0,...,0}{
            \myarrow{0,-.75}{1};
            \myarrow{1.5,-.75}{1};
            \draw [thick,rounded corners] (0,-1) -- (0,.9) -- (1.5,.9) -- (1.5,-1);
		      \filldraw[color=black, fill=whitetensorcolor, thick] (0,0) circle (0.6);
            \draw (0,0) node {\small$R^{-1}$};
		      \filldraw[color=black, fill=whitetensorcolor, thick] (1.5,0) circle (0.6);
            \draw (1.5,0) node {\small$L^{-1}$};
            \draw [thick] (0,2.4) -- (0,1.9);
            \draw (0,1.5) node {\small$\ket{0}$};
            \draw [thick] (1.5,2.4) -- (1.5,1.9);
            \draw (1.5,1.5) node {\small$\ket{0}$};
            \myarrow{0,2.2}{1};
            \myarrow{1.5,2.2}{1};
        }
        \end{tikzpicture}
        \;.
\end{align}
More explicitly, acting with $M$ on the intermediate bond legs gives
    \begin{multline} \label{eq:isometries_merging}
        M (V^{[L,R]}_{n}\otimes V^{[L,R]}_{m}) = \ket{00} \otimes V^{[L,R]}_{n+m}  = \\
        \\
                \begin{array}{c}
        \begin{tikzpicture}[scale=0.5]
		      \foreach \x in {1,...,1}{
                \GTensor{(-\singledx*\x-.5,0)}{1}{.5}{\small $A$}{1}
            \myarrow{-\singledx*\x-.5,-.65}{1};
            \myarrow{-\singledx*\x-.5,.85}{1};
            \draw [thick,rounded corners] (-\singledx*\x,0) -- (-\singledx*\x+1,0) -- (-\singledx*\x+1,2);
		      \filldraw[color=black, fill=whitetensorcolor, thick] (-\singledx*\x+1,.9) circle (0.6);
            \draw (-\singledx*\x+1,0.9) node {\small $R$};
        }
		      \foreach \x in {3,...,3}{
            \draw [thick,rounded corners] (-\singledx*\x,0) -- (-\singledx*\x-1.5,0) -- (-\singledx*\x-1.5,2);
                \GTensor{(-\singledx*\x,0)}{1}{.5}{\small $A$}{-1}
            \myarrow{-\singledx*\x,.85}{1};
            \myarrow{-\singledx*\x,-.65}{1};
            \myarrow{-\singledx*\x-1.5,1.85}{1};
		      \filldraw[color=black, fill=whitetensorcolor, thick] (-\singledx*\x-1.5,.9) circle (0.6);
            \draw (-\singledx*\x-1.5,0.9) node {\small $L$};
        }
		      \foreach \x in {2,...,2}{
                \SingleDots{-\doubledx*\x-.5,0}{\doubledx*.8}
                \draw (-\doubledx*\x-0.65,-.5) node {\small $n$};
        }
    \begin{scope}[shift={(7.6,0)}]
		      \foreach \x in {1,...,1}{
                \GTensor{(-\singledx*\x-.5,0)}{1}{.5}{\small $A$}{1}
            \myarrow{-\singledx*\x-.5,-.65}{1};
            \myarrow{-\singledx*\x-.5,.85}{1};
            \myarrow{-\singledx*\x+1,1.85}{1};
            \draw [thick,rounded corners] (-\singledx*\x,0) -- (-\singledx*\x+1,0) -- (-\singledx*\x+1,2);
		      \filldraw[color=black, fill=whitetensorcolor, thick] (-\singledx*\x+1,.9) circle (0.6);
            \draw (-\singledx*\x+1,0.9) node {\small $R$};
        }
		      \foreach \x in {3,...,3}{
            \draw [thick,rounded corners] (-\singledx*\x,0) -- (-\singledx*\x-1.5,0) -- (-\singledx*\x-1.5,2);
                \GTensor{(-\singledx*\x,0)}{1}{.5}{\small $A$}{-1}
            \myarrow{-\singledx*\x,.85}{1};
            \myarrow{-\singledx*\x,-.65}{1};
		      \filldraw[color=black, fill=whitetensorcolor, thick] (-\singledx*\x-1.5,.9) circle (0.6);
            \draw (-\singledx*\x-1.5,0.9) node {\small $L$};
        }
		      \foreach \x in {2,...,2}{
                \SingleDots{-\doubledx*\x-.5,0}{\doubledx*.8}
                \draw (-\doubledx*\x-0.65,-.5) node {\small $m$};
        }
    \end{scope}
        \begin{scope}[shift={(-2.2,1)}]
            \draw [thick] (0,2.4) -- (0,1.9);
            \draw (0,1.5) node {\small$\ket{0}$};
            \myarrow{0,2.2}{1};
        \end{scope}
        \begin{scope}[shift={(0.6,1)}]        
            \draw [thick] (1.5,2.4) -- (1.5,1.9);
            \draw (1.5,1.5) node {\small$\ket{0}$};
            \myarrow{1.5,2.2}{1};
        \end{scope}
    \begin{scope}[shift={(-.8,2.5)}]
		      \foreach \x in {0,...,0}{
            \draw [thick,rounded corners] (0,-1) -- (0,.9) -- (1.5,.9) -- (1.5,-1);
		      \filldraw[color=black, fill=whitetensorcolor, thick] (0,0) circle (0.6);
            \draw (0,0) node {\small $R^{-1}$};
		      \filldraw[color=black, fill=whitetensorcolor, thick] (1.5,0) circle (0.6);
            \draw (1.5,0) node {\small $L^{-1}$};
        }
    \end{scope}
        \end{tikzpicture}
                \end{array}
                 \;,
\end{multline}
thereby merging the two isometries into a single one over a larger region -- recall that, importantly, $L,R$ are full rank.

These observations give rise to the following implementation strategy for a uniform-bulk MPU (\cref{fig_merging}):
\begin{enumerate}[(i)]
    \item Implement $V^{[l,R]}_1 \otimes \left(V^{[L,R]}_{1}\right)^{\otimes N-2} \otimes V^{[L,r]}_1$ in parallel. This requires $O(1)$ depth and $O(N)$ auxiliary qudits.
    \item Merge neighboring isometries via applying $M$ in a tree fashion, as in Fig.~\ref{fig_merging}, until the full MPU is obtained.
\end{enumerate}

While $M$ correctly implements the desired merging between two neighboring isometries, it is itself not unitary. Consequently, each application of $M$ must be interpreted as a generalized measurement, necessitating post-selection. However, this post-selection is generally infeasible as its success probability decays exponentially with the system size\footnote{Recall that the input state to the MPU $U_N$ is unknown. Thus, any failure during a merging step in general necessitates a complete restart of the procedure due to potential global correlations of the input state.}.

To bypass the post-selection problem, we will instead employ a \emph{unitary} subroutine that merges neighboring isometries. Specifically, we will show that if a unitary realizing the isometry $V^{[x,R]}_{n}\otimes V^{[L,y]}_{m}$ ($x\in \{l,L\}$, $y \in \{ r,R \}$) can be implemented in depth $T(V^{[x,R]}_{n}\otimes V^{[L,y]}_{m})$, then
\begin{align} \label{eq:complexity_single_merging}
        T(V^{[x,y]}_{n+m}) \le  2q_{\rm{unif}} T(V^{[x,R]}_{n}\otimes V^{[L,y]}_{m}) + O(n+m)
\end{align}
using $O(1)$ additional auxiliary qudits. The constant $q_{\rm{unif}}$, which we refer to as the \emph{MPU conditioning number}, is defined for the uniform case as
\begin{align}
    q_{\rm{unif}}(R,L) \coloneqq \sqrt{\tr\left[R^{-2}(L^{-2})^T\right]} \;.
\end{align}
We note that $q_{\rm{unif}}$ is independent of the system size $N$ since $R,L$ are system-size independent operators. The total cost of implementing the MPU, starting from $V_1^{[\cdot,\cdot]}$ can be bounded from \cref{eq:complexity_single_merging}, leading to a depth
\begin{align} \label{eq:MPU_compl_unif}
    T(U_N) = O(N^{1+\log_2 q_{\rm{unif}}}),
\end{align}
which is $O(\poly N)$, as claimed in \cref{prop:main_hom}. The degree of the polynomial governed by the MPU conditioning number $q_{{\rm unif}}$, which can be optimized by the choice of $L,R$ satisfying \cref{eq:def_isometries} (see End Matter).

\prlsection{Merging isometries without post-selection}Now we explain how the merging of neighboring isometries in \cref{eq:isometries_merging} can be performed unitarily. In general terms, the problem we need to solve is: Implement
\begin{align}
    MV = 
    \begin{array}{c}
         \begin{tikzpicture}[scale=0.5,baseline={([yshift=-.8ex] current bounding box.center)}]
		      \foreach \x in {1,2,...,4}{
              \draw[thick] (\x-2.5,-1) -- (\x-2.5,0);
            \myarrow{\x-2.5,-.7}{1};
              }
		      \foreach \x in {1,2,...,6}{
              \draw[thick] (\x-3.5,0) -- (\x-3.5,1.7);
              }
		      \foreach \x in {3,4,...,4}{            \myarrow{\x-3.5,1.55}{1};
              }
		      \foreach \x in {1,2,...,2}{            \myarrow{\x-3.5,1}{1};
              }
		      \foreach \x in {5,6,...,6}{            \myarrow{\x-3.5,1}{1};
              }
                     \draw[thick, fill=tensorcolor, rounded corners=2pt] (-3,-.5) rectangle (3,.3);
	    \draw (0,-.1) node {\small $V$};
                     \draw[thick, fill=whitetensorcolor, rounded corners=2pt] (-0.8,0.6) rectangle (0.8,1.2);
	    \draw (0,0.9) node {\small $M$};
         \end{tikzpicture}
         \; \in \text{isometry}
    \end{array}
\end{align}
if $V$ is an isometry but $M$ is only a linear operator (i.e., not necessarily unitary). For this, the first step is to implement $M = \sum_i c_i W_i$ as a linear combination of unitaries with the aid of an auxiliary system prepared in the state $B \ket{0}_A \coloneqq \frac{1}{C} \sum_i \sqrt{c_i} \ket{i}_A$ (we denote $C \coloneqq \sum_i c_i$). Concretely, this can be done by applying $U \coloneqq B^\dagger W^{\rm{ctrl}} B$, where $W^{\rm{ctrl}} \coloneqq \sum_{i} \left(W_i\right)_{S} \otimes \ket{i}_A\bra{i}$ is a control unitary and $B$ acts on the ancilla. Indeed, for any $\ket{\psi}_{S}$ in the image $\mathcal S \coloneqq \image V$, a direct calculation gives
\begin{align} \label{eq:aa_subspace}
    U \ket{\psi}_{S} \ket{0}_A = \frac{1}{C} \ket{\Phi} + \sqrt{1 - \frac{1}{C^2}} \ket{\Phi^\perp} \;,
\end{align}
where $\ket{\Phi} \coloneqq (M\ket{\psi}_{S}) \ket{0}_A$ such that $_A\langle 0 \ket{\Phi^\perp} = 0$. In conclusion, this shows how post-selecting the ancilla in the $\ket{0}_A$ state results in applying $M$.

At first glance, the above construction appears to offer no advantage, since it again relies on post-selection. However, \cref{eq:aa_subspace} strongly hints that an amplitude amplification approach could be used to boost the success probability, thus making the scheme deterministic. Unfortunately, standard amplitude amplification~\cite{brassard2002amplitude} is not applicable here, since it requires a reflection about the input state, which is unknown. The setup is instead reminiscent of \emph{oblivious amplitude amplification}~\cite{berry2014exponential}, where the reflection is independent of the input. However, this method cannot be applied directly, as it requires $V$ to be a unitary, whereas in our case it is merely an isometry.
To overcome this, we develop a generalization of the oblivious amplitude amplification scheme, applicable to the isometric setting. The method is not strictly oblivious, as the amplitude amplification reflections over the $\{ \ket{\Phi},\ket{\Phi^{\perp}} \}$ subspace [cf.~\cref{eq:aa_subspace}] depend on the image of $V$, but, importantly, not on the unknown input state. This modification suffices to keep the circuit implementation efficient. Specifically, we show that
    \begin{align}
        G^\ell U (\ket{\psi}_S\ket{0}_A) = (M\ket{\psi}_{S}) \ket{0}_A
    \end{align}
where
    \begin{align}
        G \coloneqq - U R_{\Psi} U^\dagger R_{\Phi}
    \end{align}
corresponds to a single amplitude amplification rotation. The total number of rotations is bounded by $\ell \le C$, where in our case $C \le q_{\rm{unif}}$ is governed by the MPU conditioning number~\cite{sm}. The reflections are given by
\begin{subequations}
    \begin{align}
        R_{\Phi} & \coloneqq \1_S \otimes 2 \ket{0}_A\bra{0} - \1, \\
        R_{\Psi} & \coloneqq \left( \1_S \otimes 2\ket{0}_A \bra{0} - \1 \right) \left( \1 - P_{\mathcal S^\perp} \otimes 2 \ket{0}_A\bra{0} \right),
    \end{align}
\end{subequations}
where $P_{\mathcal S^\perp}$ is the (orthogonal) projector onto the orthogonal complement of $\mathcal S$. When $\mathcal{S}$ is the entire space, $P_{\mathcal{S}^\perp} = 0$, and the standard oblivious amplitude amplification protocol is recovered. Finally, we remark that the projector $P_{\mathcal{S}^\perp}$ can be implemented given an implementation of the isometry $V$ --- for instance, if $V$ is implemented as a unitary with some auxiliary input qubits set to $\ket{0}$, $P_{\mathcal{S}^\perp}$ is equivalent to undoing this unitary, then applying a phase-flip unitary controlled on the auxiliary qubits, and then applying the unitary.

In summary, a single merging step $V^{[x,R]}_{n} \otimes V^{[L,y]}_{m} \mapsto V^{[x,y]}_{n+m}$ can be performed deterministically by applying the unitary $G^\ell U$ (see \cref{fig_merging}). Each application of the amplification unitary $G$ requires two evaluations of $V^{[x,R]}_{n} \otimes V^{[L,y]}_{m}$, and a total of $\ell \le q_{\rm unif}$ repetitions is sufficient. The remaining circuit components contribute a cost of $O(1)$. A detailed counting is provided in~\cite{sm} and leads to the final bound in \cref{eq:complexity_single_merging}.

\prlsection{Nonuniform MPUs}Our implementation algorithm can be straightforwardly adapted to general, nonuniform MPUs. These are specified by site-dependent tensors $A_{1}, \dots, A_{N}$ such that the resulting MPO is unitary. For this case, $D$ denotes the maximum of the local bond dimensions $D_k$. We take for convenience uniform physical dimension $d$.

The starting point is again a decomposition of the form \cref{eq:def_isometries}. However, due to the lack of translational invariance, the corresponding isometries are now site-dependent:
    \begin{align} \label{eq:isometries_inhomo}
        V_{jk}^{[L_j,R_k]} \coloneqq 
        \begin{array}{c}
        \begin{tikzpicture}[scale=0.5]
		      \foreach \x in {1,...,1}{
                \GTensor{(-\singledx*\x-.5,0)}{1}{.5}{\small $A_k$}{1}
            \myarrow{-\singledx*\x-.5,-.65}{1};
            \myarrow{-\singledx*\x-.5,.85}{1};
            \myarrow{-\singledx*\x+1,1.85}{1};
            \draw [thick,rounded corners] (-\singledx*\x,0) -- (-\singledx*\x+1,0) -- (-\singledx*\x+1,2);
		      \filldraw[color=black, fill=whitetensorcolor, thick] (-\singledx*\x+1,.9) circle (0.6);
            \draw (-\singledx*\x+1,0.9) node {\small$R_{k}$};
        }
		      \foreach \x in {3,...,3}{
            \draw [thick,rounded corners] (-\singledx*\x,0) -- (-\singledx*\x-1.5,0) -- (-\singledx*\x-1.5,2);
                \GTensor{(-\singledx*\x,0)}{1}{.5}{\small $A_j$}{-1}
            \myarrow{-\singledx*\x,.85}{1};
            \myarrow{-\singledx*\x,-.65}{1};
            \myarrow{-\singledx*\x-1.5,1.85}{1};
		      \filldraw[color=black, fill=whitetensorcolor, thick] (-\singledx*\x-1.5,.9) circle (0.6);
            \draw (-\singledx*\x-1.5,0.9) node {\small$L_{j}$};
        }
		      \foreach \x in {2,...,2}{
                \SingleDots{-\doubledx*\x-.5,0}{\doubledx*.8}
        }
        \end{tikzpicture}
                \end{array}
        \;.
\end{align}
\noindent By exploiting the freedom in the representation of the MPU tensors~\cite{perez_garcia2007matrix}, one can bring the MPU into a canonical form, which corresponds to the usual MPS canonical form of its Choi state~\cite{styliaris2025matrix}. In this representation, we can find full-rank operators $L_j,R_k$ that give rise to valid isometries, as in the uniform case, even without the need for \cref{assumption}, but with the important difference that they are now site-dependent (see End Matter).

This immediately enables the straightforward application of the tree circuit construction, which relies only on these properties. However, care needs to be taken when calculating the cost of a single merging, since now we need to replace $q_{\rm{unif}}$ in \cref{eq:complexity_single_merging} with its site-dependent counterpart
\begin{align} \label{eq:def_q_k}
    q_k &\coloneqq \inf_{\substack{R_{k},\,L_{k+1}}} \sqrt{\tr \left[ R^{-2}_{k} (L^{-2}_{k+1})^{T} \right]} \;,
\end{align}
where the infimum is taken over valid choices that yield an isometry in \cref{eq:isometries_inhomo}. Note that all $q_k$ remain invariant with respect to gauge transformations of the form $(A_{k})^{ij} \mapsto X_{k-1}^{-1} (A_{k})^{ij} X_k$ for invertible $X_k$ and are thus a property of the MPU and not its specific representation. We define the MPU conditioning number for the general case as
\begin{align}
    q &\coloneqq \max_{k} q_k \;.
\end{align}
Our main result for general, nonuniform MPUs is:
\begin{theorem} \label{prop:main}
    An MPU over $N$ sites can be implemented with a quantum circuit of depth $O(N^{1+\log_2 q} \poly D)$ and $O(N \log D$) auxiliary qudits.
\end{theorem}

The complexity is now governed by the scaling of the MPU conditioning number $q$. A practical upper bound on $q$ can be obtained from the operator Schmidt values of the MPU. Specifically, consider its (properly normalized) Choi state $\ket{U_N}$. Unitarity ensures that $\langle U_N | U_N \rangle = \sum_i s_{k,i}^2 = d^{-N}\Tr(U_N^\dagger U_N) = 1$ where $s_{k,i}$ are the Schmidt values at bipartition $[1,\dots,k:k+1,\dots,N]$. Define the smallest nonzero Schmidt value as
\begin{align}
    s_{\min} \coloneqq \min_{k,i} \{ s_{k,i} : s_{k,i}\ne 0 \}
\end{align}
In the End Matter, we construct an explicit choice for $L_k,R_k$, for which the MPU conditioning number can be related to the MPU Schmidt values, yielding
\begin{align}
    q \le \sqrt{D}/s_{\min} \;.
\end{align}
This gives a general and efficient criterion to verify polynomial circuit complexity for an arbitrary MPU.

\prlsection{Outlook}We have addressed the question of how to implement MPUs as quantum circuits, providing an explicit construction of polynomial depth for MPUs built from a repeated bulk tensor with open boundary. Our method extends to general MPUs under a mild spectral condition, expressed in terms of the bond dimension and smallest nonzero Schmidt value of the Choi state. More broadly, our results show how the representation of unitaries as tensor networks, a form that constrains their entanglement structure, implies their efficient physical realization. It remains open to optimize the circuit depth of our algorithm using measurements and (global) feedforward, which could result in speedups. For projected entangled-pair unitaries (PEPUs) in higher spatial dimensions, our algorithm remains applicable provided the boundary state of local regions factorizes -- an analog of the uncorrelated $L,R$ condition that holds automatically in 1D. In this setting, additional care is required, as the corresponding conditioning parameter $q$ may increase, as the merging step in higher dimensions involves regions with growing boundary size. It remains an interesting direction for future work to identify and characterize subclasses of two-dimensional PEPUs that can be efficiently implemented using such methods. 

\prlsection{Acknowledgments}We are grateful to L.~Fidkowski, J.~Haah, M.~Florido-Llin\'{a}s, Y.~Liu, A.~Moln\'{a}r, D.~P\'{e}rez-Garc\'{i}a, A.~Ruiz-de-Alarc\'{o}n, and N.~Schuch for insightful discussions. J.I.C. acknowledges funding from the Federal Ministry of Education and Research Germany (BMBF) via the project Almanaq. Work at the Max Planck Institute of Quantum Optics is part of the Munich Quantum Valley, which is supported by the Bavarian state government with funds from the Hightech Agenda Bayern Plus.

\bibliography{my_refs}

\section*{End Matter}

\prlsection{Local isometries: Uniform bulk}Here we show that the operators of \cref{eq:def_isometries} are indeed isometries satisfying the properties mentioned in the main text. Consider the $N=3$ site MPU and $\sigma$, $\tau$ arbitrary single-site density operators. Then from unitarity $U^\dagger U = \1$ it is immediate that
\begin{align}
    \Tr_{1,3} \left[ (\sigma \otimes \1 \otimes \tau) U^\dagger U\right] = \1 \;.
\end{align}
In graphical representation,
\begin{align} \label{eq:LR_unitarity}
        \begin{array}{c}
        \begin{tikzpicture}[scale=0.5]
		      \foreach \x in {2,...,2}{
            \SideIdentityTensor{-\x*\doubledx,0}{\small $R^2$}{\stradius}{3};
        }
		      \foreach \x in {3,...,3}{
                \DoubleATensor{(-\doubledx*\x,0)}{-1};
                \draw (-\doubledx*\x,-0.8) node {\small $A$};
                \draw (-\doubledx*\x,0.8) node {\small $A^\dagger$};
            \myarrow{-\doubledx*\x,-0.65-0.8}{1}
            \myarrow{-\doubledx*\x,0.85+0.8}{1}
        }
		      \foreach \x in {4,...,4}{
            \SideIdentityTensor{-\x*\doubledx,0}{\small $L^2$}{\stradius}{-3};
        }
        \end{tikzpicture}
                \end{array}
        = 
        \begin{array}{c}
        \begin{tikzpicture}[scale=0.5]
		      \foreach \x in {1,...,1}{
                \DoubleIdentityTensor{(-\identitydx*\x,0)}{}{0};
            \myarrow{-\identitydx*\x,0}{1}
        }
        \end{tikzpicture}
        \end{array}
        \;,
    \end{align}
where
\begin{align}
    L^2 \coloneqq 
    \begin{tikzpicture}[scale=0.5,xscale=-1,baseline={([yshift=-2.8ex] current bounding box.center)}]
		\foreach \x in {-1,...,-1}{
            \draw[thick] (-\doubledx*\x,-0.8) -- (0,-0.8);
		      \filldraw[color=black, fill=whitetensorcolor, thick] (-\doubledx*\x,-0.8) circle (0.5);
            \draw (-\doubledx*\x,-0.8) node {\small$l$};
            \draw[thick] (-\doubledx*\x,0.8) -- (0,0.8);
		      \filldraw[color=black, fill=whitetensorcolor, thick] (-\doubledx*\x,0.8) circle (0.5);
            \draw (-\doubledx*\x,0.8) node {\small$l^*$};
        }
		      \foreach \x in {0,...,0}{
        \draw[thick, fill=whitetensorcolor, rounded corners=2pt] (-.3*\doubledx,2.8) rectangle (0.3*\doubledx,1.8);
	    \draw (\x*\doubledx,2.3) node {\small $\sigma$};
                \DoubleLongLine{-\doubledx*\x,0};
                \DoubleATensor{(-\doubledx*\x,0)}{1};
                \draw (-\doubledx*\x,-0.8) node {\small $A$};
                \draw (-\doubledx*\x,0.8) node {\small $A^\dagger$};
            \draw[thick] (-\doubledx*\x-1,0.8) -- (-\doubledx*\x-1.5,0.8);
            \myarrow{(-\doubledx*\x-1.25,0.8)}{4};
            \draw[thick] (-\doubledx*\x-1,-0.8) -- (-\doubledx*\x-1.5,-0.8);
            \myarrow{-\doubledx*\x-1,-0.8}{3};
        }
    \end{tikzpicture}
 \;, \quad
    R^2 \coloneqq 
    \begin{tikzpicture}[scale=0.5,baseline={([yshift=-2.8ex] current bounding box.center)}]
		\foreach \x in {-1,...,-1}{
            \draw[thick] (-\doubledx*\x,-0.8) -- (0,-0.8);
		      \filldraw[color=black, fill=whitetensorcolor, thick] (-\doubledx*\x,-0.8) circle (0.5);
            \draw (-\doubledx*\x,-0.8) node {\small$r$};
            \draw[thick] (-\doubledx*\x,0.8) -- (0,0.8);
		      \filldraw[color=black, fill=whitetensorcolor, thick] (-\doubledx*\x,0.8) circle (0.5);
            \draw (-\doubledx*\x,0.8) node {\small$r^*$};
        }
		      \foreach \x in {0,...,0}{
        \draw[thick, fill=whitetensorcolor, rounded corners=2pt] (-.3*\doubledx,2.8) rectangle (0.3*\doubledx,1.8);
	    \draw (\x*\doubledx,2.3) node {\small $\tau$};
                \DoubleLongLine{-\doubledx*\x,0};
                \DoubleATensor{(-\doubledx*\x,0)}{1};
                \draw (-\doubledx*\x,-0.8) node {\small $A$};
                \draw (-\doubledx*\x,0.8) node {\small $A^\dagger$};
            \draw[thick] (-\doubledx*\x-1,0.8) -- (-\doubledx*\x-1.5,0.8);
            \myarrow{(-\doubledx*\x-1.25,0.8)}{4};
            \draw[thick] (-\doubledx*\x-1,-0.8) -- (-\doubledx*\x-1.5,-0.8);
            \myarrow{-\doubledx*\x-1,-0.8}{3};
        }
    \end{tikzpicture}
    \;.
\end{align}
Note that both $L^2,R^2$ are positive-semidefinite; thus, their square roots $L,R$ are well-defined. \Cref{eq:LR_unitarity} thus establishes $V^{[L,R]\dagger}_1 V^{[L,R]}_1 = \1$. Repeating the same argument for $N = n+2$ and tracing out the first and last sites establishes that  $V^{[L,R]\dagger}_n V^{[L,R]}_n = \1$ while clearly $L,R$ are independent of $N$ and $n$. The same argument applies for tracing only the first or last site, thus yielding that $V^{[l,R]\dagger}_n,V^{[L,r]\dagger}_n$ are also isometries. The resulting $L,R$ depend on the choice of $\sigma,\tau$ and their support, which could be an arbitrary number of sites instead of a single site considered previously. This freedom can be used to optimize the constant $q_{\rm{unif}}$ that governs the degree of the polynomial in the uniform-bulk MPU circuit [\cref{eq:MPU_compl_unif}].

It remains only to show that there exists a full-rank choice of $L,R$. This is equivalent to \cref{assumption}. To see this, first notice that $\rank [L(\sigma)] \le \rank[L(\1)]$ for all states $\sigma$. This is because, for any matrix $X$ and positive-semidefinite $\sigma$, $\rank(X^\dagger \sigma X) \le \rank(X^\dagger X)$ with equality if $\sigma$ is full rank. Thus, if a full-rank $L$ exists, it is obtained for $\sigma \propto \1$. \Cref{assumption} states $\rank[L(\1)]=D$, i.e., it is full rank. A similar argument applies for $R$.

\prlsection{Local isometries: General MPUs}We now consider the nonuniform case [\cref{eq:isometries_inhomo}], with a similar derivation. Consider $\sigma$, $\tau$ arbitrary density operators supported over sites $1,\dots,j-1$ and $k+1,\dots,N$, respectively. From unitarity,
\begin{align}
        \begin{array}{c}
        \begin{tikzpicture}[scale=0.5]
		      \foreach \x in {0,...,0}{
            \SideIdentityTensor{\x*\doubledx,0}{\small $R^2_k$}{\stradius+.1}{3};
        }
		      \foreach \x in {1,...,1}{
                \DoubleATensor{(-\doubledx*\x,0)}{0};
                \draw (-\doubledx*\x,-0.8) node {\small $A_{k}$};
                \draw (-\doubledx*\x,0.8) node {\small $A^\dagger_{k}$};
            \myarrow{-\doubledx*\x,-0.65-0.8}{1}
            \myarrow{-\doubledx*\x,0.85+0.8}{1}
        }
		      \foreach \x in {2,...,2}{
                \DoubleDots{-\doubledx*\x,0}{\doubledx/2};
        }
		      \foreach \x in {3,...,3}{
                \DoubleATensor{(-\doubledx*\x,0)}{-1};
                \draw (-\doubledx*\x,-0.8) node {\small $A_{j}$};
                \draw (-\doubledx*\x,0.8) node {\small $A^\dagger_{j}$};
            \myarrow{-\doubledx*\x,-0.65-0.8}{1}
            \myarrow{-\doubledx*\x,0.85+0.8}{1}
        }
		      \foreach \x in {4,...,4}{
            \SideIdentityTensor{-\x*\doubledx,0}{\small $L^2_j$}{\stradius+0.1}{-3};
        }
        \end{tikzpicture}
                \end{array}
        = 
        \begin{array}{c}
        \begin{tikzpicture}[scale=0.5]
		      \foreach \x in {1,...,1}{
                \DoubleIdentityTensor{(-\identitydx*\x,0)}{}{0};
            \myarrow{-\identitydx*\x,0}{1}
        }
		      \foreach \x in {2,...,2}{
                \DoubleDots{-\identitydx*\x,0}{\identitydx/2};
        }
		      \foreach \x in {3,...,3}{
                \DoubleIdentityTensor{(-\identitydx*\x,0)}{}{0};
            \myarrow{-\identitydx*\x,0}{1}
        }
        \end{tikzpicture}
        \end{array}
        \;,
    \end{align}
where
\begin{align} 
    L^2_j \coloneqq 
    \begin{tikzpicture}[scale=0.4,xscale=-1,baseline={([yshift=-2.2ex] current bounding box.center)}]
        \draw[thick, fill=whitetensorcolor, rounded corners=2pt] (-2.3*\doubledx,2.8) rectangle (0.4*\doubledx,1.8);
	    \draw (-1*\doubledx,2.3) node {\small $\sigma$};
		      \foreach \x in {0,...,0}{
                \DoubleLongLine{-\doubledx*\x,0};
                \DoubleATensor{(-\doubledx*\x,0)}{1};
                \draw (-\doubledx*\x,-0.8) node {\scriptsize $A_{1}$};
                \draw (-\doubledx*\x,0.8) node {\scriptsize $A^\dagger_{1}$};
        }
		      \foreach \x in {1,...,1}{
                \DoubleDots{-\doubledx*\x,0}{\doubledx/2}
        }
		      \foreach \x in {2,...,2}{
                \DoubleLongLine{-\doubledx*\x,0}
                \DoubleATensor{(-\doubledx*\x,0)}{0}
                \draw (-\doubledx*\x+1,-1.6) node {\scriptsize $A_{j-1}$};
                \draw (-\doubledx*\x+1,0.2) node {\scriptsize $A^\dagger_{j-1}$};
            \draw[thick] (-\doubledx*\x-1,0.8) -- (-\doubledx*\x-1.5,0.8);
            \myarrow{(-\doubledx*\x-1.25,0.8)}{4};
            \draw[thick] (-\doubledx*\x-1,-0.8) -- (-\doubledx*\x-1.5,-0.8);
            \myarrow{-\doubledx*\x-1,-0.8}{3};
        }
    \end{tikzpicture}
 \;, \;
    R^2_k \coloneqq 
    \begin{tikzpicture}[scale=0.4,baseline={([yshift=-2.2ex] current bounding box.center)}]
        \draw[thick, fill=whitetensorcolor, rounded corners=2pt] (-2.3*\doubledx,2.8) rectangle (0.4*\doubledx,1.8);
	    \draw (-1*\doubledx,2.3) node {\small $\tau$};
		      \foreach \x in {0,...,0}{
                \DoubleLongLine{-\doubledx*\x,0};
                \DoubleATensor{(-\doubledx*\x,0)}{1};
                \draw (-\doubledx*\x,-0.8) node {\scriptsize $A_{\!N}$};
                \draw (-\doubledx*\x,0.8) node {\scriptsize $A^\dagger_{\!N}$};
        }
		      \foreach \x in {1,...,1}{
                \DoubleDots{-\doubledx*\x,0}{\doubledx/2}
        }
		      \foreach \x in {2,...,2}{
                \DoubleLongLine{-\doubledx*\x,0}
                \DoubleATensor{(-\doubledx*\x,0)}{0}
                \draw (-\doubledx*\x+1,-1.6) node {\scriptsize $A_{k+1}$};
                \draw (-\doubledx*\x+1,0) node {\scriptsize $A^\dagger_{k+1}$};
            \draw[thick] (-\doubledx*\x-1,0.8) -- (-\doubledx*\x-1.5,0.8);
            \myarrow{(-\doubledx*\x-1.25,0.8)}{4};
            \draw[thick] (-\doubledx*\x-1,-0.8) -- (-\doubledx*\x-1.5,-0.8);
            \myarrow{-\doubledx*\x-1,-0.8}{3};
        }
    \end{tikzpicture}
    \;.
\end{align}
Note that both operators are again positive-definite by construction and thus their square roots are well-defined, implying $V^\dagger_{jk} V_{jk} = \1$ as claimed. Since $L_j$ and $R_k$ depend on the choices of $\sigma$ and $\tau$, one can optimize over this choice, as in \cref{eq:def_q_k}.

It remains to show that there is a choice with $L_k,R_k$ full rank. This is true by bringing, without loss of generality, the MPU into a canonical form, which corresponds to the usual MPS canonical form of its Choi state~\cite{styliaris2025matrix}. In this representation, by choosing $\sigma$ and $\tau$ maximally mixed states, we get~\cite{styliaris2025matrix}
\begin{align} \label{eq:choice_schmidt}
    L_k = \1, \quad  R_k = \diag(s_{k,1},\dots,s_{k,D_k}) \quad \forall k \;,
\end{align}
where the $s_{k,i}$ are \emph{strictly positive} and they exactly correspond to the nonzero Schmidt values of the Choi state $\ket{U_N}$ for the bipartition $[1,\dots,k:k+1,\dots,N]$.
%

\appendix

\setcounter{equation}{0}
\setcounter{figure}{0}
\setcounter{table}{0}
\makeatletter
\renewcommand{\thefigure}{S\arabic{figure}}

\section*{Supplemental Material}

\section{Proof of \cref{prop:main_hom} and \cref{prop:main}}

We now present the proof of the Theorems of the main text. For clarity, we begin with the more general nonuniform case and later adapt the argument to the uniform setting. This reverses the order in which the corresponding theorems appear in the main text.

Before proceeding, we introduce a few technical ingredients required for the proof. For completeness, we repeat some definitions, results, and proofs already discussed in the main text.

\subsection{MPU Tensors and Local Isometries}
We start with the definition of a (nonuniform) MPU.

\begin{definition}[MPU] \label{def:MPU_inhomo}
    A sequence of tensors $A_{1}, \dots, A_{N}$ defines an MPU if
\begin{align}
 U = 
        \begin{array}{c}
        \begin{tikzpicture}[scale=0.5]
		      \foreach \x in {1,...,1}{
                \GTensor{(-\singledx*\x-.5,0)}{1}{.5}{\small $A_N$}{1}
            \myarrow{-\singledx*\x-.5,-.65}{1};
            \myarrow{-\singledx*\x-.5,.85}{1};
        }
		      \foreach \x in {3,...,3}{
                \GTensor{(-\singledx*\x,0)}{1}{.5}{\small $A_2$}{0}
            \myarrow{-\singledx*\x,-.65}{1};
            \myarrow{-\singledx*\x,.85}{1};
        }
		      \foreach \x in {2,...,2}{
                \SingleDots{-\doubledx*\x-.5,0}{\doubledx*.8}
        }
		      \foreach \x in {4,...,4}{
                \GTensor{(-\singledx*\x,0)}{1}{.5}{\small $A_1$}{-1}
            \myarrow{-\singledx*\x,-.65}{1};
            \myarrow{-\singledx*\x,.85}{1};
        }
        \end{tikzpicture}
                \end{array}
\end{align}
is unitary, i.e.,
\begin{align}
U^\dagger U = 
        \begin{array}{c}
        \begin{tikzpicture}[scale=0.5]
		      \foreach \x in {1,...,1}{
                \DoubleATensor{(-\doubledx*\x,0)}{1}
                \draw (-\doubledx*\x,-0.8) node {\small $A_N$};
                \draw (-\doubledx*\x,0.8) node {\small $A^\dagger_{N}$};
            \myarrow{-\doubledx*\x,-0.65-0.8}{1}
            \myarrow{-\doubledx*\x,0.85+0.8}{1}
        }
		      \foreach \x in {3,...,3}{
                \DoubleATensor{(-\doubledx*\x,0)}{0}
                \draw (-\doubledx*\x,-0.8) node {\small $A_2$};
                \draw (-\doubledx*\x,0.8) node {\small $A^\dagger_{2}$};
            \myarrow{-\doubledx*\x,-0.65-0.8}{1}
            \myarrow{-\doubledx*\x,0.85+0.8}{1}
        }
		      \foreach \x in {2,...,2}{
                \DoubleDots{-\doubledx*\x,0}{\doubledx/2}
        }
		      \foreach \x in {4,...,4}{
                \DoubleATensor{(-\doubledx*\x,0)}{-1}
                \draw (-\doubledx*\x,-0.8) node {\small $A_1$};
                \draw (-\doubledx*\x,0.8) node {\small $A^\dagger_{1}$};
            \myarrow{-\doubledx*\x,-0.65-0.8}{1}
            \myarrow{-\doubledx*\x,0.85+0.8}{1}
        }
        \end{tikzpicture}
                \end{array}
         = 
                 \begin{array}{c}
        \begin{tikzpicture}[scale=0.5]
		      \foreach \x in {1,...,1}{
                \DoubleIdentityTensor{(-\identitydx*\x,0)}{}{0};
                \myarrow{-\identitydx*\x,0}{1}
        }
		      \foreach \x in {2,...,2}{
                \DoubleDots{-\identitydx*\x,0}{\identitydx/2};
        }
		      \foreach \x in {3,...,4}{
                \DoubleIdentityTensor{(-\identitydx*\x,0)}{}{0};
                \myarrow{-\identitydx*\x,0}{1}
        }
        \end{tikzpicture}
        \end{array}
        \;.
\end{align}
\end{definition}

As discussed in the main text and End Matter, the following isometry condition can be understood as the counterpart of unitarity on the level of individual tensors.

\begin{lemma}[Local isometries] \label{lem:isometries}
Let $A_1,\dots, A_N$ be an MPU and define the sets of positive-definite operators
\begin{align} \label{eq:def_LR}
    \mathcal L_j \coloneqq \Bigg\{
    \begin{tikzpicture}[scale=0.4,xscale=-1,baseline={([yshift=-2.2ex] current bounding box.center)}]
        \draw[thick, fill=whitetensorcolor, rounded corners=2pt] (-2.3*\doubledx,2.8) rectangle (0.4*\doubledx,1.8);
	    \draw (-1*\doubledx,2.3) node {\small $\sigma$};
		      \foreach \x in {0,...,0}{
                \DoubleLongLine{-\doubledx*\x,0};
                \DoubleATensor{(-\doubledx*\x,0)}{1};
                \draw (-\doubledx*\x,-0.8) node {\scriptsize $A_{1}$};
                \draw (-\doubledx*\x,0.8) node {\scriptsize $A^\dagger_{1}$};
        }
		      \foreach \x in {1,...,1}{
                \DoubleDots{-\doubledx*\x,0}{\doubledx/2}
        }
		      \foreach \x in {2,...,2}{
                \DoubleLongLine{-\doubledx*\x,0}
                \DoubleATensor{(-\doubledx*\x,0)}{0}
                \draw (-\doubledx*\x+1,-1.6) node {\scriptsize $A_{j-1}$};
                \draw (-\doubledx*\x+1,0.2) node {\scriptsize $A^\dagger_{j-1}$};
            \draw[thick] (-\doubledx*\x-1,0.8) -- (-\doubledx*\x-1.5,0.8);
            \myarrow{(-\doubledx*\x-1.25,0.8)}{4};
            \draw[thick] (-\doubledx*\x-1,-0.8) -- (-\doubledx*\x-1.5,-0.8);
            \myarrow{-\doubledx*\x-1,-0.8}{3};
        }
    \end{tikzpicture}
\succ 0 \Bigg\}_{\displaystyle \sigma} \;, \;
    \mathcal R_k \coloneqq \Bigg\{
    \begin{tikzpicture}[scale=0.4,baseline={([yshift=-2.2ex] current bounding box.center)}]
        \draw[thick, fill=whitetensorcolor, rounded corners=2pt] (-2.3*\doubledx,2.8) rectangle (0.4*\doubledx,1.8);
	    \draw (-1*\doubledx,2.3) node {\small $\tau$};
		      \foreach \x in {0,...,0}{
                \DoubleLongLine{-\doubledx*\x,0};
                \DoubleATensor{(-\doubledx*\x,0)}{1};
                \draw (-\doubledx*\x,-0.8) node {\scriptsize $A_{\!N}$};
                \draw (-\doubledx*\x,0.8) node {\scriptsize $A^\dagger_{\!N}$};
        }
		      \foreach \x in {1,...,1}{
                \DoubleDots{-\doubledx*\x,0}{\doubledx/2}
        }
		      \foreach \x in {2,...,2}{
                \DoubleLongLine{-\doubledx*\x,0}
                \DoubleATensor{(-\doubledx*\x,0)}{0}
                \draw (-\doubledx*\x+1,-1.6) node {\scriptsize $A_{k+1}$};
                \draw (-\doubledx*\x+1,0) node {\scriptsize $A^\dagger_{k+1}$};
            \draw[thick] (-\doubledx*\x-1,0.8) -- (-\doubledx*\x-1.5,0.8);
            \myarrow{(-\doubledx*\x-1.25,0.8)}{4};
            \draw[thick] (-\doubledx*\x-1,-0.8) -- (-\doubledx*\x-1.5,-0.8);
            \myarrow{-\doubledx*\x-1,-0.8}{3};
        }
    \end{tikzpicture}
    \succ 0 \Bigg\}_{\displaystyle \tau}
\end{align}
where $\sigma, \tau$ are density operators. Then for all $j \le k$, $L^2_j \in \mathcal L_j$, $R^2_k \in \mathcal R_k$,
    \begin{align}
        V_{jk} \coloneqq 
        \begin{array}{c}
        \begin{tikzpicture}[scale=0.5]
		      \foreach \x in {1,...,1}{
                \GTensor{(-\singledx*\x-.5,0)}{1}{.5}{\small $A_k$}{1}
            \myarrow{-\singledx*\x-.5,-.65}{1};
            \myarrow{-\singledx*\x-.5,.85}{1};
            \myarrow{-\singledx*\x+1,1.85}{1};
            \draw [thick,rounded corners] (-\singledx*\x,0) -- (-\singledx*\x+1,0) -- (-\singledx*\x+1,2);
		      \filldraw[color=black, fill=whitetensorcolor, thick] (-\singledx*\x+1,.9) circle (0.6);
            \draw (-\singledx*\x+1,0.9) node {\small$R_{k}$};
        }
		      \foreach \x in {3,...,3}{
            \draw [thick,rounded corners] (-\singledx*\x,0) -- (-\singledx*\x-1.5,0) -- (-\singledx*\x-1.5,2);
                \GTensor{(-\singledx*\x,0)}{1}{.5}{\small $A_j$}{-1}
            \myarrow{-\singledx*\x,.85}{1};
            \myarrow{-\singledx*\x,-.65}{1};
            \myarrow{-\singledx*\x-1.5,1.85}{1};
		      \filldraw[color=black, fill=whitetensorcolor, thick] (-\singledx*\x-1.5,.9) circle (0.6);
            \draw (-\singledx*\x-1.5,0.9) node {\small$L_{j}$};
        }
		      \foreach \x in {2,...,2}{
                \SingleDots{-\doubledx*\x-.5,0}{\doubledx*.8}
        }
        \end{tikzpicture}
                \end{array}
        \;.
\end{align}
is an isometry, i.e., $V^\dagger_{jk} V_{jk} = \1$.
\end{lemma}
\begin{proof}
The proof was given in the End Matter, but we repeat it here for completeness. Let $\sigma$, $\tau$ be arbitrary density operators supported over sites $1,\dots,j-1$ and $k+1,\dots,N$, respectively. Then from unitarity $U^\dagger U = \1$ it is immediate that
\begin{align}
    \Tr_{1,\dots,j-1, k+1,\dots, N} \left[ (\sigma \otimes \1 \otimes \tau) U^\dagger U\right] = \1 \;.
\end{align}
In graphical representation,
\begin{align}
        \begin{array}{c}
        \begin{tikzpicture}[scale=0.5]
		      \foreach \x in {0,...,0}{
            \SideIdentityTensor{\x*\doubledx,0}{\small $R^2_k$}{\stradius}{3};
        }
		      \foreach \x in {1,...,1}{
                \DoubleATensor{(-\doubledx*\x,0)}{0};
                \draw (-\doubledx*\x,-0.8) node {\small $A_{k}$};
                \draw (-\doubledx*\x,0.8) node {\small $A^\dagger_{k}$};
            \myarrow{-\doubledx*\x,-0.65-0.8}{1}
            \myarrow{-\doubledx*\x,0.85+0.8}{1}
        }
		      \foreach \x in {2,...,2}{
                \DoubleDots{-\doubledx*\x,0}{\doubledx/2};
        }
		      \foreach \x in {3,...,3}{
                \DoubleATensor{(-\doubledx*\x,0)}{-1};
                \draw (-\doubledx*\x,-0.8) node {\small $A_{j}$};
                \draw (-\doubledx*\x,0.8) node {\small $A^\dagger_{j}$};
            \myarrow{-\doubledx*\x,-0.65-0.8}{1}
            \myarrow{-\doubledx*\x,0.85+0.8}{1}
        }
		      \foreach \x in {4,...,4}{
            \SideIdentityTensor{-\x*\doubledx,0}{\small $L^2_j$}{\stradius}{-3};
        }
        \end{tikzpicture}
                \end{array}
        = 
        \begin{array}{c}
        \begin{tikzpicture}[scale=0.5]
		      \foreach \x in {1,...,1}{
                \DoubleIdentityTensor{(-\identitydx*\x,0)}{}{0};
            \myarrow{-\identitydx*\x,0}{1}
        }
		      \foreach \x in {2,...,2}{
                \DoubleDots{-\identitydx*\x,0}{\identitydx/2};
        }
		      \foreach \x in {3,...,3}{
                \DoubleIdentityTensor{(-\identitydx*\x,0)}{}{0};
            \myarrow{-\identitydx*\x,0}{1}
        }
        \end{tikzpicture}
        \end{array}
        \;,\addtocounter{assumption}{-1} 
    \end{align}
where $L_j, R_k$ are as in \cref{eq:def_LR}. Since $\sigma,\tau \succeq 0$, also $L_j,R_k \succeq 0$. Their square roots are thus well-defined, implying $V^\dagger_{jk} V_{jk} = \1$.
\end{proof}

Every MPU can be brought into a (nonuniform) canonical form by vectorizing and imposing the MPS (left) canonical form~\cite{perez_garcia2007matrix}, which results in the following~\cite{styliaris2025matrix}.

\begin{definition}[MPU canonical form]
An MPU is in canonical form if all its tensors satisfy
\begin{align} \label{eq:gauge_nonuniform}
        \frac{1}{d}
            \begin{array}{c}
            \begin{tikzpicture}[scale=0.5]
    		      \foreach \x in {1,...,1}{
                  \ETensor{-\doubledx*\x,0}{0};
                  \draw (-\doubledx*\x,-.8) node {\small $A_k$};
                  \draw (-\doubledx*\x,.8) node {\small $A^\dagger_k$};
                 }
    		      \foreach \x in {2,...,2}{
                  \SideIdentityTensor{-\doubledx*\x,0}{}{}{-1};
                 }
            \end{tikzpicture}
            \end{array}
        =
            \begin{array}{c}
            \begin{tikzpicture}[scale=.5]
                \SideIdentityTensor{0,0}{}{}{-2};
            \end{tikzpicture}
            \end{array}
        \;.
\end{align}
\noindent For the boundary tensor $A_1$ ($A_N$), the left (right) auxiliary space is interpreted as trivial.
\end{definition}
Since the MPS canonical form is obtained by successive singular-value decompositions, the resulting bond dimension at each bipartition is the minimal possible, i.e., all singular values are strictly positive. Thus, taking the MPU tensors in canonical form and setting $\sigma,\tau$ maximally mixed, we get, after a unitary gauge transformation,
\begin{align} \label{eq:app:choice_schmidt}
    L_k = \1, \quad  R_k = \diag(s_{k,1},\dots,s_{k,D_k}) \quad \forall k \;,
\end{align}
as claimed in \cref{eq:choice_schmidt}. Since also $\1 \in \mathcal L_k$, this shows that $\mathcal L_k, \mathcal R_k$ are always non-empty.

\subsection{Unitary decomposition of an operator}

We will need this simple lemma; we include its proof for completeness.

\begin{lemma} \label{lem:lcu}
    For every matrix $M \in \mathcal M\left( \mathbb C^{d} \right)$ there exists a decomposition
    \begin{align}
        M = \sum_{i=1}^{H} c_i W_i
    \end{align}
    with $W_i^\dagger W_i = \1$, $c_i>0$, and $H \le d^2$ such that
    \begin{align}
        \sum_i c_i = \| M \|_1 \;.
    \end{align}
\end{lemma}
\begin{proof}
 From SVD, we first decompose
 \begin{align}
     M = W S V^\dagger = \sum_i \ket{u_i} \bra{v_i} s_i \;.
 \end{align}
 Each $\ket{u_i} \bra{v_i}$ is in the convex hull of unitary matrices. This follows simply by writing
 \begin{align}
     \ket{u_1} \bra{v_1} = W \diag(1,0,0,\dots) V^\dagger
 \end{align}
 and then expressing the diagonal matrix as a convex combination of $\diag (1,\pm 1,\pm 1,\dots )$. Although this seemingly requires at most $2^{d-1}$ extremal points, using Caratheodory's theorem, it can be reduced to at most $d$. Thus we can write $\ket{u_i} \bra{v_i} = \sum_{i=1}^{r_j} P_{ij} W_{ij}$ for some unitaries $W_{ij}$, $P_{ij} > 0$ with $\sum_j P_{ij} = 1$, and $H_i \le d$. This implies
 \begin{align}
     M = \sum_{i=1}^{\rank M}\sum_{j=1}^{H_i} s_i P_{ij} W_{ij}
 \end{align}
and therefore
\begin{align}
    \sum_{ij} s_i P_{ij} = \sum_i s_i = \| M \|_1 \;.
\end{align}
\end{proof}

\subsection{Amplitude amplification}

We will need a generalization of the oblivious amplitude amplification algorithm~\cite{berry2014exponential}. The reason is that, in Lemma 3.6 of Ref.~\cite{berry2014exponential}, $V$ therein is unitary, while in our setting it will be only an isometry. We thus show the following technical lemma, extending a result of Ref.~\cite{berry2014exponential}.

\begin{lemma}[Subspace amplitude amplification] \label{lem:amp_amp}
    Let $U$ be a unitary on $\mathcal H_S \otimes \mathcal H_A$, $M$ an operator acting on $\mathcal H_S$, and $\mathcal S$ a subspace of $\mathcal H_S$. Denote $P_{\mathcal S^\perp}$ the orthogonal projector onto the orthogonal complement of $\mathcal S$. Suppose that:
    \begin{enumerate}[(i)]
        \item $\bra{\psi'} M^\dagger M \ket{\psi} = \bra{\psi'} \psi \rangle$ for all $\ket{\psi},\ket{\psi'} \in \mathcal S$.
        \item For all $\ket{\psi} \in \mathcal S$,
        \begin{align}
            U (\ket{\psi}_S\ket{0}_A) = \sin\theta \ket{\Phi} + \cos\theta \ket{\Phi^\perp}
        \end{align}
        with $\sin\theta\ne 0$, $\cos\theta \ne 0$. Here $\ket{\Phi} \coloneqq (M\ket{\psi}_S) \ket{0}_A$ satisfying $\bra{\Phi^\perp} \Phi \rangle = 0$ and $_A\bra{0} \Phi^\perp \rangle = 0$.
    \end{enumerate}
    Denoting $\ket{\Psi} \coloneqq \ket{\psi}_S\ket{0}_A$, define its orthogonal vector via
    \begin{align}
        U \ket{\Psi^\perp} \coloneqq \cos\theta \ket{\Phi} - \sin \theta \ket{\Phi^\perp}
    \end{align}
    which depends on $\ket{\psi}$. Then:
    \begin{enumerate}[(a)]
        \item For all $\ket{\psi}, \ket{\psi'} \in \mathcal S$,
        \begin{align} \label{eq:orth_psi_perp}
        _S\bra{\psi'}_A\bra{0}  \Psi^\perp \rangle  = 0 \;.
        \end{align}
        \item For all $\ket{\psi} \in \mathcal S$, the unitary
    \begin{align}
        R_{\Phi} & \coloneqq \1_S \otimes 2 \ket{0}_A\bra{0} - \1
    \end{align}
    reflects along $\ket{\Phi}$ on the $\{\ket{\Phi}, \ket{\Phi^\perp}\}$ subspace, while
    \begin{align}
        R_{\Psi} & \coloneqq \left( \1_S \otimes 2\ket{0}_A \bra{0} - \1 \right) \left( \1 - P_{\mathcal S^\perp} \otimes 2 \ket{0}_A\bra{0} \right)
    \end{align}
    reflects along $\ket{\Psi}$ on the $\{\ket{\Psi}, \ket{\Psi^\perp}\}$ subspace.
    \end{enumerate}
\end{lemma}
\begin{proof}
    \textbf{(a)} To simplify the notation, when not ambiguous, we drop the subscripts $S,A$. We have
    \begin{align*}
        \bra{\psi'}\bra{0}  \Psi^\perp \rangle &= \bra{\psi'}\bra{0}  U^\dagger U \ket{\Psi^\perp} \\
        & = \left( \sin \theta (\bra{\psi'} M^\dagger) \bra{0} + \cos\theta \bra{\Phi'^\perp} \right) \\
        &\qquad \qquad \times \left( \cos \theta (M \ket{\psi})\ket{0} - \sin\theta \ket{\Phi^\perp} \right) \\
        & = \sin\theta\cos\theta \left(\bra{\psi'} \psi \rangle - \bra{\Phi'^\perp} \Phi^\perp \rangle \right) \;,
    \end{align*}
    where $\ket{\Phi'^\perp}$ corresponds to $\ket{\psi'}$ and in the last step we used $(i)$ from the hypothesis and also $_A\langle0\ket{\Phi^\perp} = 0$, $_A\langle0\ket{\Phi'^\perp} = 0$. Finally, from the unitarity of $U$ we have
    \begin{align*}
        \bra{\psi'} \psi \rangle &= \bra{\Psi'} \Psi \rangle = \bra{\Psi'} U^\dagger U \ket{\Psi} \\
        & = \sin^2 \theta \bra{\psi'} \psi \rangle + \cos^2 \theta \bra{\Phi'^\perp} \Phi^\perp \rangle \;,
    \end{align*}
    where, in the last step, we again used $(i)$. For $\cos \theta \ne 0$, $\bra{\psi'} \psi \rangle = \bra{\Phi'^\perp} \Phi^\perp \rangle$ and the result follows.
    \\
    
    \textbf{(b)} Directly by its definition and using the assumption $_A\bra{0} \Phi^\perp \rangle = 0$,
    \begin{subequations} \label{subeqs:phi_reflection}
    \begin{align}
        R_{\Phi} \ket{\Phi} & = \ket{\Phi} \;, \\
        R_{\Phi} \ket{\Phi^\perp} & = -\ket{\Phi^\perp} \;. 
    \end{align}
    \end{subequations}
    Analogously, using $_B\bra{0} \Psi^\perp \rangle = 0$ and \cref{eq:orth_psi_perp},
    \begin{subequations}\label{subeqs:psi_reflection}
    \begin{align}
        R_{\Psi} \ket{\Psi} &= \ket{\Psi} \;, \\
        R_{\Psi} \ket{\Psi^\perp} & = - \ket{\Psi^\perp} \;.
    \end{align}
    \end{subequations}
\end{proof}

As a result, one can perform amplitude amplification by utilizing the reflection operators $R_{\Phi^\perp}$, $R_{\Psi}$:

\begin{corollary} \label{cor:amp_amp}
    In the setup of \cref{lem:amp_amp}, for any $\ell \in \mathbb Z$ the unitary
    \begin{align}
        G \coloneqq - U R_{\Psi} U^\dagger R_{\Phi}
    \end{align}
    transforms
    \begin{align}
        G^\ell U (\ket{\psi}_S\ket{0}_A) &= G^\ell\left( \sin\theta \ket{\Phi} + \cos\theta \ket{\Phi^\perp} \right) \nonumber \\
         &= \sin[(2\ell +1)\theta] \ket{\Phi} + \cos[(2\ell +1)\theta] \ket{\Phi^\perp} \;.
    \end{align}
\end{corollary}
\noindent For the special case $\mathcal S = \mathcal H_S$, clearly $P_{\mathcal S^\perp} = 0$ and thus $R_{\Psi} = R_{\Phi}$, recovering oblivious amplitude amplification~\cite{berry2014exponential}.
\begin{proof}
    The result follows by standard amplitude amplification using the results of part (b) in \cref{lem:amp_amp}. We repeat here the analysis for completeness.

    From the definitions of \cref{lem:amp_amp} we have
\begin{subequations}
\begin{align}
    U \ket{\Psi} &= \sin \theta \ket{\Phi} + \cos \theta \ket{\Phi^\perp} \;,\\
    U \ket{\Psi^\perp} &= \cos \theta \ket{\Phi} - \sin \theta \ket{\Phi^\perp} \;,\\
    U^\dagger \ket{\Phi} &= \sin \theta \ket{\Psi} + \cos \theta \ket{\Psi^\perp} \;, \\
    U^\dagger\ket{\Phi^\perp} &= \cos \theta \ket{\Psi} - \sin \theta \ket{\Psi^\perp} \;.
\end{align}
\end{subequations}
Using Eqs.~\eqref{subeqs:phi_reflection} and \eqref{subeqs:psi_reflection} together with the above, a direct calculation gives
\begin{align}
    G U \ket{\Psi} = \sin (3 \theta) \ket{\Phi} + \cos (3 \theta) \ket{\Phi^\perp} \;.
\end{align}
\end{proof}

\subsection{Proof of Theorem~\ref{prop:main}}

Let us restate the Theorem. Recall that
\begin{align}
    q &\coloneqq \min_{k} q_k  \le \sqrt{D}/s_{\min} \;, \quad \text{where} \\
    q_k &\coloneqq \inf_{\substack{R^2_{k} \in \mathcal R_k \\ L^2_{k+1} \in \mathcal L_{k+1}}}  \sqrt{\tr \left[ R^{-2}_{k} (L^{-2}_{k+1})^{T} \right]}
\end{align}
%

\newcounter{savetheorem}
\setcounter{savetheorem}{\value{theorem}}

\setcounter{theorem}{1}

\begin{theorem} \label{prop:app:main}
    An MPU over $N$ sites can be implemented with a quantum circuit of depth $O(N^{1+\log_2 q} \poly D)$ and $O(N \log D$) auxiliary qudits.
\end{theorem}

\setcounter{theorem}{\value{savetheorem}}

The key ingredient for our circuit decomposition is an algorithm to merge neighboring isometries -- that is, given the ability to realize $V_{j,k-1} \otimes V_{kl}$, to implement $V_{jl}$. This is established by combining the local isometries of \cref{lem:isometries} with our generalization of oblivious amplitude amplification, \cref{lem:amp_amp}:

\begin{lemma}[Merging lemma] \label{lem:merging}
    If a unitary realizing the isometry $V_{jk}\otimes V_{k+1,l}$ [cf. \cref{eq:def_isometries}] can be implemented in depth $T(V_{jk}\otimes V_{k+1,l})$ then $V_{jl}$ can be implemented in depth
\begin{multline} 
        T(V_{jl}) \le  q_k \big[ 2 T(V_{jk}\otimes V_{k+1,l})  \\ + O[(l-j)\polylog D]  + O(\poly D) \big]
\end{multline}
using $O(\log D)$ additional auxiliary qudits. 
\end{lemma}

\begin{proof}
    By the MPU canonical form of \cref{lem:isometries}, to merge $V_{jk}$ and $V_{k+1,l}$ we need to apply
\begin{align}
    M_k \coloneqq \ket{00} \bra{ \1 } (R^{-1}_{k} \otimes L^{-1}_{k+1}) = 
        \begin{tikzpicture}[scale=0.55,baseline={([yshift=0.0ex] current bounding box.center)}]
		      \foreach \x in {0,...,0}{
            \myarrow{0,-.75}{1};
            \myarrow{1.5,-.75}{1};
            \draw [thick,rounded corners] (0,-1) -- (0,.9) -- (1.5,.9) -- (1.5,-1);
		      \filldraw[color=black, fill=whitetensorcolor, thick] (0,0) circle (0.6);
            \draw (0,0) node {\small$R^{-1}_{k}$};
		      \filldraw[color=black, fill=whitetensorcolor, thick] (1.5,0) circle (0.6);
            \draw (1.5,0) node {\small$L^{-1}_{k\!+\!1}$};
            \draw [thick] (0,2.4) -- (0,1.9);
            \draw (0,1.5) node {\small$\ket{0}$};
            \draw [thick] (1.5,2.4) -- (1.5,1.9);
            \draw (1.5,1.5) node {\small$\ket{0}$};
            \myarrow{0,2.2}{1};
            \myarrow{1.5,2.2}{1};
        }
        \end{tikzpicture}
\end{align}
acting on the two auxiliary $D$-dimensional systems to be joined together:
    \begin{multline}
        M_k (V_{jk}\otimes V_{k+1,l}) = \ket{00} \otimes V_{jl}  = \\
        \\
                \begin{array}{c}
        \begin{tikzpicture}[scale=0.5]
		      \foreach \x in {1,...,1}{
                \GTensor{(-\singledx*\x-.5,0)}{1}{.5}{\small $A_k$}{1}
            \myarrow{-\singledx*\x-.5,-.65}{1};
            \myarrow{-\singledx*\x-.5,.85}{1};
            \draw [thick,rounded corners] (-\singledx*\x,0) -- (-\singledx*\x+1,0) -- (-\singledx*\x+1,2);
		      \filldraw[color=black, fill=whitetensorcolor, thick] (-\singledx*\x+1,.9) circle (0.6);
            \draw (-\singledx*\x+1,0.9) node {\small$R_{k}$};
        }
		      \foreach \x in {3,...,3}{
            \draw [thick,rounded corners] (-\singledx*\x,0) -- (-\singledx*\x-1.5,0) -- (-\singledx*\x-1.5,2);
                \GTensor{(-\singledx*\x,0)}{1}{.5}{\small $A_j$}{-1}
            \myarrow{-\singledx*\x,.85}{1};
            \myarrow{-\singledx*\x,-.65}{1};
            \myarrow{-\singledx*\x-1.5,1.85}{1};
		      \filldraw[color=black, fill=whitetensorcolor, thick] (-\singledx*\x-1.5,.9) circle (0.6);
            \draw (-\singledx*\x-1.5,0.9) node {\small$L_{j}$};
        }
		      \foreach \x in {2,...,2}{
                \SingleDots{-\doubledx*\x-.5,0}{\doubledx*.8}
        }
    \begin{scope}[shift={(7.6,0)}]
		      \foreach \x in {1,...,1}{
                \GTensor{(-\singledx*\x-.5,0)}{1}{.5}{\small $A_l$}{1}
            \myarrow{-\singledx*\x-.5,-.65}{1};
            \myarrow{-\singledx*\x-.5,.85}{1};
            \myarrow{-\singledx*\x+1,1.85}{1};
            \draw [thick,rounded corners] (-\singledx*\x,0) -- (-\singledx*\x+1,0) -- (-\singledx*\x+1,2);
		      \filldraw[color=black, fill=whitetensorcolor, thick] (-\singledx*\x+1,.9) circle (0.6);
            \draw (-\singledx*\x+1,0.9) node {\small$R_{k}$};
        }
		      \foreach \x in {3,...,3}{
            \draw [thick,rounded corners] (-\singledx*\x,0) -- (-\singledx*\x-1.5,0) -- (-\singledx*\x-1.5,2);
                \GTensor{(-\singledx*\x,0)}{1}{.5}{\scriptsize $A_{\!k\!+\!1}$}{-1}
            \myarrow{-\singledx*\x,.85}{1};
            \myarrow{-\singledx*\x,-.65}{1};
		      \filldraw[color=black, fill=whitetensorcolor, thick] (-\singledx*\x-1.5,.9) circle (0.6);
            \draw (-\singledx*\x-1.5,0.9) node {\small$L_{\!k\!+\!1}$};
        }
		      \foreach \x in {2,...,2}{
                \SingleDots{-\doubledx*\x-.5,0}{\doubledx*.8}
        }
    \end{scope}
        \begin{scope}[shift={(-2.2,1)}]
            \draw [thick] (0,2.4) -- (0,1.9);
            \draw (0,1.5) node {\small$\ket{0}$};
            \myarrow{0,2.2}{1};
        \end{scope}
        \begin{scope}[shift={(0.6,1)}]        
            \draw [thick] (1.5,2.4) -- (1.5,1.9);
            \draw (1.5,1.5) node {\small$\ket{0}$};
            \myarrow{1.5,2.2}{1};
        \end{scope}
    \begin{scope}[shift={(-.8,2.5)}]
		      \foreach \x in {0,...,0}{
            \draw [thick,rounded corners] (0,-1) -- (0,.9) -- (1.5,.9) -- (1.5,-1);
		      \filldraw[color=black, fill=whitetensorcolor, thick] (0,0) circle (0.6);
            \draw (0,0) node {\small$R^{-1}_{k}$};
		      \filldraw[color=black, fill=whitetensorcolor, thick] (1.5,0) circle (0.6);
            \draw (1.5,0) node {\small$L^{-1}_{k\!+\!1}$};
        }
    \end{scope}
        \end{tikzpicture}
                \end{array}
                 \;,
\end{multline}
We denote $S_1$ the space over which $M_k$ acts ($\dim S_1 = D^2$) and $S_2$ the remaining output space ($\dim S_2 = d^{l-j+1} D^2$) with $S \coloneqq S_1 S_2$ the joint ``system'' space. We first use \cref{lem:lcu} to expand
\begin{align}
    M_k = \sum_{i = 1}^H c_i W_i \;,
\end{align}
where $c_i > 0$, the sum runs over $H \le D^4$ unitaries and
\begin{align}
    C \coloneqq \sum_i c_i \le \| M_k \|_1 = q_k \;.
\end{align}

We will now use a linear combination of unitaries to turn the action of $M_k$ into a physical process. For this, we need an $R$-dimensional ``auxiliary'' space $A$ and the unitary $B$ acting on that space as
\begin{align}
 B \ket{0}_A \coloneqq \frac{1}{C} \sum_i \sqrt{c_i} \ket{i} \;.
\end{align}
Define the control unitary
\begin{align}
    W^{\rm{ctrl}} \coloneqq \sum_{i} \left(W_i \right)_{S_1} \otimes \ket{i}_A\bra{i}
\end{align}
and the unitary
\begin{align}
    U \coloneqq B^\dagger W^{\rm{ctrl}} B
\end{align}
acting over $S_1 A$. Consider any $\ket{\psi}_{S} \in \mathcal S$ in the image $ \mathcal S \coloneqq \image \left( V_{j,k-1} \otimes V_{kl} \right) \subseteq \mathcal H_{S}$ and denote
\begin{align}
 \ket{\Psi} = \ket{\psi}_{S} \ket{0}_A \;.   
\end{align}
A direct calculation gives
\begin{align}
    U \ket{\Psi} = \frac{1}{C} \ket{\Phi} + \sqrt{1 - \frac{1}{C^2}} \ket{\Phi^\perp} \;,
\end{align}
where
\begin{align}
    \ket{\Phi} \coloneqq (M_k\ket{\psi}_{S}) \ket{0}_A
\end{align}
such that
\begin{align}
    _A\langle 0 \ket{\Phi^\perp} = 0 \;.
\end{align}
Moreover $M_k$ satisfies $\bra{\psi'} M_k^\dagger M_k \ket{\psi} = \bra{\psi'} \psi \rangle$ for all $\ket{\psi},\ket{\psi'} \in \mathcal S$ due of \cref{lem:isometries}. The above shows that we can use the subspace amplitude amplification Lemma~\ref{lem:amp_amp} with $\sin \theta = 1/C$.

Let us now bound the depth $T(G)$ of implementing a single amplitude amplification operator from \cref{cor:amp_amp}. We have
\begin{align}
    T(U) & = O(\poly D)
\end{align}
because $U$ acts on the $S_1A$ register, while
\begin{align}
        T(R_{\Phi}) &= O(\polylog D)
\end{align}
since $R_{\Phi}$ amounts to a multi-control phase gate over the auxiliary $A$ register~\cite{nielsen2002quantum}. The same unitary is a component of $R_{\Psi}$ but here we additionally need to implement reflection $\left( \1 - P_{\mathcal S^\perp} \otimes 2 \ket{0}_A\bra{0} \right)$. For the latter we utilize the unitary $\widetilde{V}$ which implements the isometry $V_{jk}\otimes V_{k+1,l}$, i.e.,
\begin{align} \label{eq:dillation_isometry}
    \widetilde V \ket{0}^{\otimes 2(l-j+1)} \coloneqq V_{jk}\otimes V_{k+1,l} \;,
\end{align}
where each $\ket{0}$ above denotes a $D$-dimensional qudit. The desired reflection then becomes
\begin{align}
    \left( \1 - P_{\mathcal S^\perp} \otimes 2 \ket{0}_A\bra{0} \right) = \widetilde V F \widetilde V^\dagger 
\end{align}
where $F$ is a multi-control multi-target phase gate, with control over the auxiliary $A$ register and target a $D^{2(l-j+1)}$-dimensional register (corresponding to $\ket{0}^{\otimes 2(l-j+1)}$ in \cref{eq:dillation_isometry}). We thus have
\begin{align}
    T(C) = O[(l-j) \polylog D]
\end{align}
and therefore
\begin{align}
    T(R_{\Psi}) = 2 T(V_{jk}\otimes V_{k+1,l}) + O[(l-j)\polylog D] \;.
\end{align}
Putting everything together, the depth of a single $G$ rotation is
\begin{multline}
       T(G) \le 2 T(V_{jk}\otimes V_{k+1,l})  \\ + O[(l-j)\polylog D]  + O(\poly D) \;.
\end{multline}

Next, we repeatedly apply $G$ in order to decouple the $A$ register, as in ordinary amplitude amplification. The required number of repetitions is $\ell \le C$. This follows by
\begin{align} \label{eq:repetitions_grover}
    (2 \ell +1) \theta = \pi/2 \;,
\end{align}
in case $C$ is such that there is an integer $\ell$ satisfying the above equation. Otherwise, we can use an additional auxiliary qubit and modify
\begin{align}
 M_k' = M_k \otimes   \left( \cos \phi \1 + i \sin \phi Z \right) \;.
\end{align}
This additional rotation over the qubit redefines $C' = C\left(\left| \cos \phi \right| + \left| \sin \phi \right| \right)$. Tuning $\phi$ can then be used to satisfy \cref{eq:repetitions_grover} for an integer $\ell$.
\end{proof}

Our algorithm has the following steps.

\begin{enumerate}[(i)]
    \item Start by implementing the isometries $\bigotimes_k V_{kk}$.
    \item Iterate the Merging Lemma~\ref{lem:merging} in a tree geometry, i.e., first implement $V_{12} \otimes V_{34} \otimes \dots$, then $V_{14} \otimes V_{58} \otimes \dots$, and so on, until the MPU $V_{1N}$ is complete.
\end{enumerate}

We now bound the total depth of the algorithm by forming a recursion. Step (i) can be implemented in depth
\begin{align}
    T_1 = O(\poly D) \;.
\end{align}
Step (ii) need to be iterated $L = \lceil \log N \rceil$ times ($\log$ is base 2); denote $T_j$ ($j=1,\dots,L$) the corresponding depth. From the Merging Lemma~\ref{lem:merging},
\begin{align}
    T_j \le  q[ 2T_{j-1} + O(2^{j} \polylog D) + O(\poly D )]
\end{align}
Summing the series, we get
\begin{align}
    T_L =  O[(2q)^L \poly D ] \;.
\end{align}
This is because we have a recursion of the form $T_{j} = T_{j-1} + Q_j$ and thus (assuming $2q>1$)
\begin{align}
    T_L &= (2q)^{L-1} T_1 + \sum_{j=2}^L Q_j \sum_{n=0}^{L-j} (2q)^{n}  \nonumber \\
    & \le (2q)^{L-1} T_1 + \sum_{j=2}^L (2q)^{L-j+1} Q_j \nonumber \\
    & \le O[ (2q)^L \poly D]
\end{align}
Thus the total depth is
\begin{align}
    T_{\log N} = O(N^{\log q + 1} \poly D) \;.
\end{align}
The required total number of auxiliary qudits is $O(N \log D)$.

\subsection{Proof of \cref{prop:main_hom}}

We will prove Theorem 1', which is a refinement of \cref{prop:main_hom}. In the uniform setting, $D$ is a constant, so we omit it from the scaling. We need the following Lemma.

\begin{lemma}[Local isometries: Uniform case] \label{lem:isometries_unif}
Let the tensor $A$ with boundary vectors $l,r$ define a uniform-bulk MPU. Define the sets of positive-definite operators
\begin{align}
    \mathcal L_m &\coloneqq \Bigg\{
    \begin{tikzpicture}[scale=0.5,xscale=-1,baseline={([yshift=-2.8ex] current bounding box.center)}]
        \draw[thick, fill=whitetensorcolor, rounded corners=2pt] (-2.3*\doubledx,2.8) rectangle (0.4*\doubledx,1.8);
	    \draw (-1*\doubledx,2.3) node {\small $\sigma$};
		\foreach \x in {-1,...,-1}{
            \draw[thick] (-\doubledx*\x,-0.8) -- (0,-0.8);
		      \filldraw[color=black, fill=whitetensorcolor, thick] (-\doubledx*\x,-0.8) circle (0.5);
            \draw (-\doubledx*\x,-0.8) node {\small$l$};
            \draw[thick] (-\doubledx*\x,0.8) -- (0,0.8);
		      \filldraw[color=black, fill=whitetensorcolor, thick] (-\doubledx*\x,0.8) circle (0.5);
            \draw (-\doubledx*\x,0.8) node {\small$l^*$};
        }
              \foreach \x in {0,...,0}{
                \DoubleLongLine{-\doubledx*\x,0};
                \DoubleATensor{(-\doubledx*\x,0)}{1};
                \draw (-\doubledx*\x,-0.8) node {\small $A$};
                \draw (-\doubledx*\x,0.8) node {\small $A^\dagger$};
        }
		      \foreach \x in {1,...,1}{
                \DoubleDots{-\doubledx*\x,0}{\doubledx/2}
            \draw (-\doubledx*\x,-1.2) node {\small $m$};
        }
		      \foreach \x in {2,...,2}{
                \DoubleLongLine{-\doubledx*\x,0}
                \DoubleATensor{(-\doubledx*\x,0)}{0}
                \draw (-\doubledx*\x,.8) node {\small $A^\dagger$};
                \draw (-\doubledx*\x,-0.8) node {\small $A$};
            \draw[thick] (-\doubledx*\x-1,0.8) -- (-\doubledx*\x-1.5,0.8);
            \myarrow{(-\doubledx*\x-1.25,0.8)}{4};
            \draw[thick] (-\doubledx*\x-1,-0.8) -- (-\doubledx*\x-1.5,-0.8);
            \myarrow{-\doubledx*\x-1,-0.8}{3};
        }
    \end{tikzpicture}
\succ 0 \Bigg\}_{\displaystyle \sigma} \;,
\end{align}
\begin{align}
    \mathcal R_m &\coloneqq \Bigg\{
    \begin{tikzpicture}[scale=0.5,baseline={([yshift=-2.8ex] current bounding box.center)}]
        \draw[thick, fill=whitetensorcolor, rounded corners=2pt] (-2.3*\doubledx,2.8) rectangle (0.4*\doubledx,1.8);
	    \draw (-1*\doubledx,2.3) node {\small $\tau$};
		\foreach \x in {-1,...,-1}{
            \draw[thick] (-\doubledx*\x,-0.8) -- (0,-0.8);
		      \filldraw[color=black, fill=whitetensorcolor, thick] (-\doubledx*\x,-0.8) circle (0.5);
            \draw (-\doubledx*\x,-0.8) node {\small$r$};
            \draw[thick] (-\doubledx*\x,0.8) -- (0,0.8);
		      \filldraw[color=black, fill=whitetensorcolor, thick] (-\doubledx*\x,0.8) circle (0.5);
            \draw (-\doubledx*\x,0.8) node {\small$r^*$};
        }
		      \foreach \x in {0,...,0}{
                \DoubleLongLine{-\doubledx*\x,0};
                \DoubleATensor{(-\doubledx*\x,0)}{1};
                \draw (-\doubledx*\x,-0.8) node {\small $A$};
                \draw (-\doubledx*\x,0.8) node {\small $A^\dagger$};
        }
		      \foreach \x in {1,...,1}{
                \DoubleDots{-\doubledx*\x,0}{\doubledx/2}
            \draw (-\doubledx*\x,-1.2) node {\small $m$};
        }
		      \foreach \x in {2,...,2}{
                \DoubleLongLine{-\doubledx*\x,0}
                \DoubleATensor{(-\doubledx*\x,0)}{0}
                \draw (-\doubledx*\x,-0.8) node {\small $A$};
                \draw (-\doubledx*\x,0.8) node {\small $A^\dagger$};
            \draw[thick] (-\doubledx*\x-1,0.8) -- (-\doubledx*\x-1.5,0.8);
            \myarrow{(-\doubledx*\x-1.25,0.8)}{4};
            \draw[thick] (-\doubledx*\x-1,-0.8) -- (-\doubledx*\x-1.5,-0.8);
            \myarrow{-\doubledx*\x-1,-0.8}{3};
        }
    \end{tikzpicture}
    \succ 0 \Bigg\}_{\displaystyle \tau}
\end{align}
where $\sigma, \tau$ are density operators of over $m$ qudits. Define $\mathcal L \coloneqq \bigcup_{m} \mathcal L_m$ and $\mathcal R \coloneqq \bigcup_m \mathcal R_m$. Then for all $n\ge1$, $L^2  \in \mathcal L$, $R^2 \in \mathcal R$,
    \begin{align}
        V^{[L,R]}_{n} \coloneqq 
        \begin{array}{c}
        \begin{tikzpicture}[scale=0.5]
		      \foreach \x in {1,...,1}{
                \GTensor{(-\singledx*\x-.5,0)}{1}{.5}{\small $A_k$}{1}
            \myarrow{-\singledx*\x-.5,-.65}{1};
            \myarrow{-\singledx*\x-.5,.85}{1};
            \myarrow{-\singledx*\x+1,1.85}{1};
            \draw [thick,rounded corners] (-\singledx*\x,0) -- (-\singledx*\x+1,0) -- (-\singledx*\x+1,2);
		      \filldraw[color=black, fill=whitetensorcolor, thick] (-\singledx*\x+1,.9) circle (0.6);
            \draw (-\singledx*\x+1,0.9) node {\small$R$};
        }
		      \foreach \x in {3,...,3}{
            \draw [thick,rounded corners] (-\singledx*\x,0) -- (-\singledx*\x-1.5,0) -- (-\singledx*\x-1.5,2);
                \GTensor{(-\singledx*\x,0)}{1}{.5}{\small $A$}{-1}
            \myarrow{-\singledx*\x,.85}{1};
            \myarrow{-\singledx*\x,-.65}{1};
            \myarrow{-\singledx*\x-1.5,1.85}{1};
		      \filldraw[color=black, fill=whitetensorcolor, thick] (-\singledx*\x-1.5,.9) circle (0.6);
            \draw (-\singledx*\x-1.5,0.9) node {\small$L$};
        }
		      \foreach \x in {2,...,2}{
                \SingleDots{-\doubledx*\x-.5,0}{\doubledx*.8}
            \draw (-\doubledx*\x-.6,-0.5) node {\small $n$};
        }
        \end{tikzpicture}
                \end{array}
        \;.
\end{align}
is an isometry, i.e., $V^{[L,R]\dagger}_{n} V^{[L,R]}_{n} = \1$. Similarly for
    \begin{align}
        V^{[l,R]}_{n} \coloneqq 
        \begin{array}{c}
        \begin{tikzpicture}[scale=0.5]
		      \foreach \x in {1,...,1}{
                \GTensor{(-\singledx*\x-.5,0)}{1}{.5}{\small $A_k$}{1}
            \myarrow{-\singledx*\x-.5,-.65}{1};
            \myarrow{-\singledx*\x-.5,.85}{1};
            \myarrow{-\singledx*\x+1,1.85}{1};
            \draw [thick,rounded corners] (-\singledx*\x,0) -- (-\singledx*\x+1,0) -- (-\singledx*\x+1,2);
		      \filldraw[color=black, fill=whitetensorcolor, thick] (-\singledx*\x+1,.9) circle (0.6);
            \draw (-\singledx*\x+1,0.9) node {\small$R$};
        }
		      \foreach \x in {3,...,3}{
            \draw [thick,rounded corners] (-\singledx*\x,0) -- (-\singledx*\x-1.5,0) -- (-\singledx*\x-1.5,1);
                \GTensor{(-\singledx*\x,0)}{1}{.5}{\small $A$}{-1}
            \myarrow{-\singledx*\x,.85}{1};
            \myarrow{-\singledx*\x,-.65}{1};
		      \filldraw[color=black, fill=whitetensorcolor, thick] (-\singledx*\x-1.5,.9) circle (0.6);
            \draw (-\singledx*\x-1.5,0.9) node {\small$l$};
        }
		      \foreach \x in {2,...,2}{
                \SingleDots{-\doubledx*\x-.5,0}{\doubledx*.8}
            \draw (-\doubledx*\x-.6,-0.5) node {\small $n$};
        }
        \end{tikzpicture}
                \end{array}
        \;.
\end{align}
and
    \begin{align}
        V^{[L,r]}_{n} \coloneqq 
        \begin{array}{c}
        \begin{tikzpicture}[scale=0.5]
		      \foreach \x in {1,...,1}{
                \GTensor{(-\singledx*\x-.5,0)}{1}{.5}{\small $A_k$}{1}
            \myarrow{-\singledx*\x-.5,-.65}{1};
            \myarrow{-\singledx*\x-.5,.85}{1};
            \draw [thick,rounded corners] (-\singledx*\x,0) -- (-\singledx*\x+1,0) -- (-\singledx*\x+1,1);
		      \filldraw[color=black, fill=whitetensorcolor, thick] (-\singledx*\x+1,.9) circle (0.6);
            \draw (-\singledx*\x+1,0.9) node {\small$r$};
        }
		      \foreach \x in {3,...,3}{
            \draw [thick,rounded corners] (-\singledx*\x,0) -- (-\singledx*\x-1.5,0) -- (-\singledx*\x-1.5,2);
                \GTensor{(-\singledx*\x,0)}{1}{.5}{\small $A$}{-1}
            \myarrow{-\singledx*\x,.85}{1};
            \myarrow{-\singledx*\x,-.65}{1};
            \myarrow{-\singledx*\x-1.5,1.85}{1};
		      \filldraw[color=black, fill=whitetensorcolor, thick] (-\singledx*\x-1.5,.9) circle (0.6);
            \draw (-\singledx*\x-1.5,0.9) node {\small$L$};
        }
		      \foreach \x in {2,...,2}{
                \SingleDots{-\doubledx*\x-.5,0}{\doubledx*.8}
            \draw (-\doubledx*\x-.6,-0.5) node {\small $n$};
        }
        \end{tikzpicture}
                \end{array}
        \;.
\end{align}
\end{lemma}
\begin{proof}
The proof is similar to \cref{lem:isometries}. Consider the $N=3$ site MPU and $\sigma$, $\tau$ arbitrary single-site density operators. Then from unitarity $U^\dagger U = \1$ it is immediate that
\begin{align}
    \Tr_{1,3} \left[ (\sigma \otimes \1 \otimes \tau) U^\dagger U\right] = \1 \;.
\end{align}
In graphical representation,
\begin{align} \label{eq:app:LR_unitarity}
        \begin{array}{c}
        \begin{tikzpicture}[scale=0.5]
		      \foreach \x in {2,...,2}{
            \SideIdentityTensor{-\x*\doubledx,0}{\small $R^2$}{\stradius}{3};
        }
		      \foreach \x in {3,...,3}{
                \DoubleATensor{(-\doubledx*\x,0)}{-1};
                \draw (-\doubledx*\x,-0.8) node {\small $A$};
                \draw (-\doubledx*\x,0.8) node {\small $A^\dagger$};
            \myarrow{-\doubledx*\x,-0.65-0.8}{1}
            \myarrow{-\doubledx*\x,0.85+0.8}{1}
        }
		      \foreach \x in {4,...,4}{
            \SideIdentityTensor{-\x*\doubledx,0}{\small $L^2$}{\stradius}{-3};
        }
        \end{tikzpicture}
                \end{array}
        = 
        \begin{array}{c}
        \begin{tikzpicture}[scale=0.5]
		      \foreach \x in {1,...,1}{
                \DoubleIdentityTensor{(-\identitydx*\x,0)}{}{0};
            \myarrow{-\identitydx*\x,0}{1}
        }
        \end{tikzpicture}
        \end{array}
        \;,
    \end{align}
where
\begin{align}
    L^2 \coloneqq 
    \begin{tikzpicture}[scale=0.5,xscale=-1,baseline={([yshift=-2.8ex] current bounding box.center)}]
		\foreach \x in {-1,...,-1}{
            \draw[thick] (-\doubledx*\x,-0.8) -- (0,-0.8);
		      \filldraw[color=black, fill=whitetensorcolor, thick] (-\doubledx*\x,-0.8) circle (0.5);
            \draw (-\doubledx*\x,-0.8) node {\small$l$};
            \draw[thick] (-\doubledx*\x,0.8) -- (0,0.8);
		      \filldraw[color=black, fill=whitetensorcolor, thick] (-\doubledx*\x,0.8) circle (0.5);
            \draw (-\doubledx*\x,0.8) node {\small$l^*$};
        }
		      \foreach \x in {0,...,0}{
        \draw[thick, fill=whitetensorcolor, rounded corners=2pt] (-.3*\doubledx,2.8) rectangle (0.3*\doubledx,1.8);
	    \draw (\x*\doubledx,2.3) node {\small $\sigma$};
                \DoubleLongLine{-\doubledx*\x,0};
                \DoubleATensor{(-\doubledx*\x,0)}{1};
                \draw (-\doubledx*\x,-0.8) node {\small $A$};
                \draw (-\doubledx*\x,0.8) node {\small $A^\dagger$};
            \draw[thick] (-\doubledx*\x-1,0.8) -- (-\doubledx*\x-1.5,0.8);
            \myarrow{(-\doubledx*\x-1.25,0.8)}{4};
            \draw[thick] (-\doubledx*\x-1,-0.8) -- (-\doubledx*\x-1.5,-0.8);
            \myarrow{-\doubledx*\x-1,-0.8}{3};
        }
    \end{tikzpicture}
 \;, \quad
    R^2 \coloneqq 
    \begin{tikzpicture}[scale=0.5,baseline={([yshift=-2.8ex] current bounding box.center)}]
		\foreach \x in {-1,...,-1}{
            \draw[thick] (-\doubledx*\x,-0.8) -- (0,-0.8);
		      \filldraw[color=black, fill=whitetensorcolor, thick] (-\doubledx*\x,-0.8) circle (0.5);
            \draw (-\doubledx*\x,-0.8) node {\small$r$};
            \draw[thick] (-\doubledx*\x,0.8) -- (0,0.8);
		      \filldraw[color=black, fill=whitetensorcolor, thick] (-\doubledx*\x,0.8) circle (0.5);
            \draw (-\doubledx*\x,0.8) node {\small$r^*$};
        }
		      \foreach \x in {0,...,0}{
        \draw[thick, fill=whitetensorcolor, rounded corners=2pt] (-.3*\doubledx,2.8) rectangle (0.3*\doubledx,1.8);
	    \draw (\x*\doubledx,2.3) node {\small $\tau$};
                \DoubleLongLine{-\doubledx*\x,0};
                \DoubleATensor{(-\doubledx*\x,0)}{1};
                \draw (-\doubledx*\x,-0.8) node {\small $A$};
                \draw (-\doubledx*\x,0.8) node {\small $A^\dagger$};
            \draw[thick] (-\doubledx*\x-1,0.8) -- (-\doubledx*\x-1.5,0.8);
            \myarrow{(-\doubledx*\x-1.25,0.8)}{4};
            \draw[thick] (-\doubledx*\x-1,-0.8) -- (-\doubledx*\x-1.5,-0.8);
            \myarrow{-\doubledx*\x-1,-0.8}{3};
        }
    \end{tikzpicture}
    \;.
\end{align}
Note that both $L^2,R^2$ are positive-semidefinite; thus, their square roots $L,R$ are well-defined. \Cref{eq:app:LR_unitarity} thus establishes $V^{[L,R]\dagger}_1 V^{[L,R]}_1 = \1$ for all $L \in \mathcal L_1$, $R \in \mathcal R_1$. Repeating the same argument over $n+m_1+m_2$ sites and tracing out $m_1 \ge 0$ sites on the left and $m_2 \ge 0$ on the right gives $V^{[L,R]\dagger}_n V^{[L,R]}_n = \1$ for all $L \in \mathcal L^2_{m_1}$, $R^2 \in \mathcal R_{m_2}$. Taking $m_1 = 0$ or $m_2 = 0$ respectively yields that $V^{[l,R]}_n$, $V^{[L,r]}_n$ are isometries.
\end{proof}

\begin{manualtheorem} \label{prop:app_hom}
    Suppose $U_N$ is a uniform-bulk MPU satisfying \cref{assumption}. Then $U_N$ can be implemented with a quantum circuit of depth $O(N^{1+\log_2 q_{\rm{unif}}})$ and $O(N$) auxiliary qudits, where
\begin{align}
    q_{\rm{unif}} \coloneqq \inf_{\substack{R^2 \in \mathcal R \\ L^2 \in \mathcal L}}   \sqrt{\tr \left[ R^{-2} (L^{-2})^{T} \right]} = O(1) \;.
\end{align}
\end{manualtheorem}

\noindent The infimum, which does not explicitly appear in the main text, is not necessary; any choice $R^2 \in \mathcal R$,  $L^2 \in \mathcal L$ is valid.

\begin{proof}
Our strategy is to leverage the uniform structure of the tensor network in order to get rid of the system-size dependence in the MPU conditioning number. This is essentially achieved in \cref{lem:isometries_unif}, since any choice of $L,R$ is independent of $N$. This implies the following simple modification of our algorithm.
\begin{enumerate}[(i)]
    \item Implement $V^{[l,R]}_1 \otimes \left(V^{[L,R]}_{1}\right)^{\otimes N-2} \otimes V^{[L,r]}_1$ in parallel. This requires $O(1)$ depth and $O(N)$ auxiliary qudits.
    \item Merge neighboring isometries via applying $M$ in a tree fashion, as in Fig.~\ref{fig_merging}, until the full MPU is obtained.
\end{enumerate}
The key point is that the merging for the bulk tensors in steps in (ii) is achieved by a \textit{uniform} operator $M$ (i.e., not depending on the position or the total system size),
\begin{align} \label{app:eq:merging}
    M \coloneqq \ket{00} \bra{ \1 } (R \otimes L) = 
        \begin{tikzpicture}[scale=0.5,baseline={([yshift=0.0ex] current bounding box.center)}]
		      \foreach \x in {0,...,0}{
            \myarrow{0,-.75}{1};
            \myarrow{1.5,-.75}{1};
            \draw [thick,rounded corners] (0,-1) -- (0,.9) -- (1.5,.9) -- (1.5,-1);
		      \filldraw[color=black, fill=whitetensorcolor, thick] (0,0) circle (0.6);
            \draw (0,0) node {\small$R$};
		      \filldraw[color=black, fill=whitetensorcolor, thick] (1.5,0) circle (0.6);
            \draw (1.5,0) node {\small$L$};
            \draw [thick] (0,2.4) -- (0,1.9);
            \draw (0,1.5) node {\small$\ket{0}$};
            \draw [thick] (1.5,2.4) -- (1.5,1.9);
            \draw (1.5,1.5) node {\small$\ket{0}$};
            \myarrow{0,2.2}{1};
            \myarrow{1.5,2.2}{1};
        }
        \end{tikzpicture}
        \;.
\end{align}
As such, the Merging Lemma~\ref{lem:merging} can be applied for the bulk tensors by replacing
\begin{align}
    q_k \mapsto q_{\rm unif} \;.
\end{align}

To conclude the proof, we need to show that the sets $\mathcal R, \mathcal L$ contain full-rank operators\footnote{For the nonuniform case, this was guaranteed by the (nonuniform) canonical form, which led to \cref{eq:choice_schmidt}. However, here it is crucial to maintain a uniform representation (i.e., the bulk tensors remain all the same), so we cannot assume this canonical form.} (i.e., they are non-empty) -- otherwise, it might be that the merging step results in an unwanted projector in the bond space, altering the resulting tensor network after merging. However, our \cref{assumption} guarantees that indeed the above sets are non-empty, e.g., by taking $\sigma$ and $\tau$ maximally mixed density operators.

To see this, first notice that $\rank [L(\sigma)] \le \rank[L(\1)]$ for all states $\sigma$ (similarly for $R$). This is since, for any matrix $X$ and positive-semidefinite $\sigma$, $\rank(X^\dagger \sigma X) \le \rank(X^\dagger X)$ with equality if $\sigma$ is full rank. Thus, if a full-rank $L$ exists, it is obtained for $\sigma \propto \1$. \Cref{assumption} states $\rank[L(\1)]=\rank[R(\1)]=D$.

\end{proof}

\section{MPUs satisfying \cref{assumption}}

In the main text, for the uniform-bulk MPUs, we assumed:
\begin{assumption}\label{assumption_sm}
\renewcommand\theassumption{1} 
It holds that
\setcounter{assumption}{0}
\begin{align}
 \rank \left(
        \begin{array}{c}
        \begin{tikzpicture}[scale=0.5]
		      \foreach \x in {5,...,5}{
                \GTensor{(-\singledx*\x,0)}{1}{.5}{\small $A$}{0}
            \myarrow{-\singledx*\x,-.85}{2};
            \myarrow{-\singledx*\x,.85}{1};
            \myarrow{-\singledx*\x+0.7,0}{4};
        }
		\foreach \x in {6,...,6}{
		      \filldraw[color=black, fill=whitetensorcolor, thick] (-\singledx*\x+0.4,0) circle (0.5);
            \draw (-\singledx*\x+0.4,0) node {\small$l$};
        }
        \end{tikzpicture}
                \end{array}
    \right) = 
 \rank \left(
        \begin{array}{c}
        \begin{tikzpicture}[scale=0.5]
		\foreach \x in {1,...,1}{
                \GTensor{(-\singledx*\x-.5,0)}{1}{.5}{\small $A$}{0}
            \myarrow{-\singledx*\x-.5,-.85}{2};
            \myarrow{-\singledx*\x-.5,.85}{1};
            \myarrow{-\singledx*\x-1.15,0}{3};
        }
		\foreach \x in {0,...,0}{
		      \filldraw[color=black, fill=whitetensorcolor, thick] (-\singledx*\x-0.9,0) circle (0.5);
            \draw (-\singledx*\x-0.9,0) node {\small$r$};
        } 
        \end{tikzpicture}
                \end{array}
        \right)
    = D
        \;.
\end{align}
\end{assumption}
\noindent The condition is automatically satisfied if $A$, viewed as an MPS tensor, is injective. However, there is a much more general class for which the assumption holds true. For this, we need to introduce:

\begin{definition}[Block-diagonal uniform-bulk MPU] \label{def:app_homo_mpu}
    The bulk tensor $A$ with a boundary operator $b$ -- both independent of $N$ -- define an MPU if, for all system sizes $N \ge 1$, the MPO
\begin{align} \label{app:eq:U_N_hom}
U_N = 
    \begin{array}{c}
        \begin{tikzpicture}[scale=.5,baseline={([yshift=-0.75ex] current bounding box.center)}]
		      \foreach \x in {0,...,0}{
                \SingleTrRight{(0,0)}
        }
		      \foreach \x in {1,...,1}{
                \GTensor{(-\singledx*\x,0)}{1}{.5}{\small $A$}{0}
                \myarrow{-\singledx*\x,-.65}{1};
                \myarrow{-\singledx*\x,.85}{1};
        }
		      \foreach \x in {2,...,2}{
                \SingleDots{-\doubledx*\x-.5,0}{\doubledx*.8}
        }
		      \foreach \x in {3,...,3}{
                \GTensor{(-\singledx*\x,0)}{1}{.5}{\small $A$}{0}
                \myarrow{-\singledx*\x,-.65}{1};
                \myarrow{-\singledx*\x,.85}{1};
        }
		      \foreach \x in {4,...,4}{
                \bTensor{-\singledx*\x,0}{b}
        }
		      \foreach \x in {5,...,5}{
                \SingleTrLeft{(-\singledx*\x+.3,0)}
        }
        \end{tikzpicture}
        \end{array}
        \,.
\end{align}
is unitary.
\end{definition}
We now show that \cref{assumption_sm} is satisfied if $A$ above is in block-diagonal canonical form, with linearly independent blocks, and the corresponding blocks of the boundary tensor $b$ are all full rank. More specifically, suppose that
\begin{align}\label{eq:MPU_canonical}
A^{ij} = \bigoplus_{k = 1}^g A^{ij}_k \text{ and }b = \bigoplus_{k = 1}^g b_k,
\end{align}
where the tensors $A_k$, with $A_k^{ij} \in \mathcal{M} \left( \mathbb C^{D_k} \right)$, are independent elements in a basis of normal tensors 
\cite{cirac2017matrix2} and $b_k \in \mathcal{M} \left( \mathbb C^{D_k} \right)$ are invertible operators. We recall that since $A_k$ are independent elements in the basis of normal tensors, then after blocking we can construct another set of tensors $A^{-1}_1, A^{-1}_2 \dots A^{-1}_g$ such that \cite{perez_garcia2007matrix}
\begin{align}
\sum_{i, j} A_k^{ij} \otimes (A^{-1}_l)^{ij} = \delta_{k, l} \1_{D_k} \otimes \1_{D_k},
\end{align}
where $\1_{D_k}$ is the $D_k-$dimensional identity. Diagramatically, Eq.~\eqref{eq:basis_normal_tensor_inverse} can be expressed as
\begin{align}\label{eq:basis_normal_tensor_inverse}
\begin{tikzpicture}[scale=.5,baseline={([yshift=-0.75ex] current bounding box.center)}]
		      \foreach \x in {2,...,2}{
                \GVecTensor{(-\singledx*\x,0)}{1}{.5}{\small $A_k$}{0}
                \GBraTensor{(-\singledx*\x,1.5)}{1}{.5}{\small $A_l^{-1}$}{0}
        }
        \end{tikzpicture} = \delta_{k, l}
        \begin{tikzpicture}[scale=.5,baseline={([yshift=-0.75ex] current bounding box.center)}]
                \SideIdentityTensor{(0, 0)}{}{}{1}
                \SideIdentityTensor{(0.35, 0)}{}{}{-1}
                \draw[thick](-1, 0.8) -- (-0.6, 0.8);
                \draw[thick](-1, -0.8) -- (-0.6, -0.8);
                \draw[thick](0.95, -0.8) -- (1.35, -0.8);
                \draw[thick](0.95, 0.8) -- (1.35, 0.8);
        \end{tikzpicture}
        \;.
\end{align}

We first rewrite the MPU specified by Eq.~\eqref{eq:MPU_canonical} with open boundary conditions --- this MPU can equivalently be generated by the bulk tensor $\tilde{A}$
\begin{align}
    \tilde{A} = \bigoplus_{k = 1}^g \1_{D_k} \otimes A_k^{ij},
\end{align}
together with the left and right boundary vectors $\bra{l} = \bigoplus_{k = 1}^g \bra{l_k}$ and $\ket{r} = \bigoplus_{k =1 }^g \ket{r_k}$, respectively, where
\begin{align}
    \bra{l_k} = \sum_{\alpha, \beta = 1}^{D_k} (b_k)_{\alpha, \beta} \bra{\alpha, \beta}\text{ and } \ket{r_k} = \sum_{\alpha, \beta = 1}^{D_k} \ket{\alpha, \beta}.
\end{align}
In this open-boundary representation, the MPU has a bond-dimension $D = \sum_{k = 1}^b D_k^2$ with the bond-space being the vector space $\mathcal{H}_\text{bond} = \bigoplus_{k = 1}^g \mathbb{C}^{D_k}\otimes \mathbb{C}^{D_k} $. We now show that, for MPUs of this form, Assumption \ref{assumption_sm} is satisfied ---  to see this, let us assume the contrary, i.e., $\exists \ket{v} \in \mathcal{H}_\text{bond}$ such that
\begin{align}\label{eq:null_condition}
\bra{l} A^{ij} \ket{v} = 0 \ \forall i, j \;.
\end{align}
Since $\ket{v} \in \mathcal{H}_\text{bond}$, we can express it as $\ket{v} = \bigoplus_{k = 1}^g \ket{v_k}$ where $\ket{v_k}$ can be represented by a matrix $v_k \in \mathcal{M} \left( \mathbb C^{D_k} \right)$ via
\[
\ket{v_k} = \sum_{\alpha, \beta} (v_k)_{\alpha, \beta} \ket{\alpha, \beta}.
\]
Now, Eq.~\eqref{eq:null_condition} then implies that
\[
\sum_{k = 1}^g \bra{l_k} (\1_{D_k} \otimes A_k^{ij}) \ket{v_k} = 0.
\]
This can be diagramatically represented as
\[
\sum_{k = 1}^g \begin{array}{c}
        \begin{tikzpicture}[scale=.5,baseline={([yshift=-0.75ex] current bounding box.center)}]
		      \foreach \x in {0,...,0}{
                \SingleTrRight{(0,0)}
        }
        		      \foreach \x in {1,...,1}{
                \bTensor{-\singledx*\x,0}{$v_k$}
        }
		      \foreach \x in {2,...,2}{
                \GVecTensor{(-\singledx*\x,0)}{1}{.5}{\small $A_k$}{0}
                \draw (-\singledx*\x - 0.3*\singledx,.95) node {$i$};
                \draw (-\singledx*\x + 0.3*\singledx,.95) node {$j$};
        }
		      \foreach \x in {3,...,3}{
                \bTensor{-\singledx*\x,0}{$b_k$}
        }
		      \foreach \x in {4,...,4}{
                \SingleTrLeft{(-\singledx*\x+.3,0)}
        }
        \end{tikzpicture}
        \end{array} = 0,
\]
from which it follows that $\forall l \in \{1, 2 \dots g\}$,
\[
\sum_{k = 1}^g \begin{array}{c}
        \begin{tikzpicture}[scale=.5,baseline={([yshift=-0.75ex] current bounding box.center)}]
		      \foreach \x in {0,...,0}{
                \SingleTrRight{(0,0)}
        }
        		      \foreach \x in {1,...,1}{
                \bTensor{-\singledx*\x,0}{$v_k$}
        }
		      \foreach \x in {2,...,2}{
                \GVecTensor{(-\singledx*\x,0)}{1}{.5}{\small $A_k$}{0}
                \GBraTensor{(-\singledx*\x,1.5)}{1}{.5}{\small $A_l^{-1}$}{0}
        }
		      \foreach \x in {3,...,3}{
                \bTensor{-\singledx*\x,0}{$b_k$}
        }
		      \foreach \x in {4,...,4}{
                \SingleTrLeft{(-\singledx*\x+.3,0)}
        }
        \end{tikzpicture}
        \end{array} = 0 \;.
\]
Using Eq.~\eqref{eq:basis_normal_tensor_inverse}, we then obtain that $v_l b_l = 0$ --- since $b_l$ is full-rank and hence invertible (by assumption), this implies that $v_l=  0$ and therefore $\ket{v} = 0$. This analysis thus verifies that the condition in \cref{assumption_sm} involving the left boundary vector is satisfied. An identical analysis can be repeated for the condition in \cref{assumption_sm} involving the right boundary vector, thus establishing \cref{assumption_sm}.

\section{MPUs from $C^*$-Weak Hopf Algebras}

We will show here how MPUs can be constructed from $C^*$-weak Hopf algebras, and also that they can be prepared in $\text{poly}(N)$ time using the main result of this Letter.

\subsection{Review of $C^*$-Weak Hopf Algebras}\label{sec:weak_hopf_review}
We first recall the definition of a $C^*$-weak Hopf algebra, and some relevant properties.
\begin{definition}[$C^*$-Weak Hopf Algebra]\label{def:weak_hopf_Cstar}
A $C^*$-Weak Hopf Algebra $\mathcal{A}$ over complex numbers is a finite-dimensional $C^*$ algebra such that:
\begin{enumerate}
    \item[(a)] There is a linear map $\Delta : \mathcal{A} \to \mathcal{A} \otimes \mathcal{A}$, called the co-product, such that
    \[
        (\Delta \otimes \textnormal{id}) \circ  \Delta = (\textnormal{id} \otimes \Delta) \circ \Delta.
    \]
    \item[(b)] There is a linear map $\epsilon : \mathcal{A} \to \mathbb{C}$ called the counit such that
    \[
        (\epsilon \otimes \textnormal{id})\circ \Delta = (\textnormal{id} \otimes \epsilon)\circ \Delta = \textnormal{id}.
    \]
    \item[(c)] The co-product $\Delta$ is compatible with the algebra multiplication and involution, i.e., $\forall x, y \in \mathcal{A}$
    \[
    \Delta(xy) = \Delta(x) \Delta(y) \text{ and }\Delta(x^*) = \Delta(x)^*.
    \]
     The co-unit $\epsilon$ satisfies
    \[
    \epsilon(xyz) = \sum_{\alpha} \epsilon(xy_\alpha^{(1)}) \epsilon(y_{\alpha}^{(2)} z),
    \]
    where
    \[
    \sum_{\alpha} y_{\alpha}^{(1)} \otimes y_{\alpha}^{(2)} = \Delta(y),
    \]
    and the unit $1 \in \mathcal{A}$ satisfies
    \[
    ((\Delta \otimes \textnormal{id})\circ \Delta)(1) = (\Delta(1) \otimes 1)(1 \otimes \Delta(1)).
    \]
    \item[(d)] There is a linear bijection $S:\mathcal{A} \to \mathcal{A}$, called the antipode, such that $\forall x \in \mathcal{A}$
    \begin{align*}
    &\sum_{\alpha} S(x_{\alpha}^{(1)})x_\alpha^{(2)} = \sum_{\alpha}\epsilon(e_\alpha^{(1)} x)e_\alpha^{(2)}, \nonumber\\
    &\sum_{\alpha} x_{\alpha}^{(1)} S(x_\alpha^{(2)}) = \sum_\alpha \epsilon(x e_\alpha^{(2)}) e_\alpha^{(1)},
    \end{align*}
    where
    \begin{align*}
        &\Delta(x) = \sum_\alpha x_\alpha^{(1)}\otimes x_\alpha^{(2)}, \Delta(1) = \sum_\alpha e_\alpha^{(1)}\otimes e_\alpha^{(2)}.
    \end{align*}
\end{enumerate}
\end{definition}
Given a $C^*$-weak Hopf algebra $\mathcal{A}$, its dual $\mathcal{A}^* = \text{Hom}(\mathcal{A}, \mathbb{C}) = \{f:\mathcal{A} \to \mathbb{C} , f \text{ linear}\}$ can also be endowed the structure of a $C^*$-weak Hopf algebra. Given $f \in \text{Hom}(\mathcal{A}, \mathbb{C})$, $x \in \mathcal{A}$, we let $\langle f, x\rangle =f(x) \in \mathbb{C}$. 
\begin{lemma}[Ref.~\cite{bohm1999weak}]\label{lemma:dual_C*}
    Given a $C^*$-weak Hopf algebra $\mathcal{A}$, the dual $\mathcal{A}^* = \textnormal{Hom}(\mathcal{A}, \mathbb{C})$ with
    \begin{enumerate}
        \item[(a)] Multiplication between $f, g \in \mathcal{A}^*$ defined by
        \[
        \forall x \in \mathcal{A}: \langle fg , x\rangle = \langle f\otimes g, \Delta(x)\rangle,
        \]
        \item[(b)] Involution $f^*$ of $f \in \mathcal{A}^*$ defined by
        \[
        \forall x \in \mathcal{A}: \langle f^*, x\rangle = \overline{\langle f, S(x)^*\rangle},
        \]
        \item[(c)] Coproduct $\Delta_*:\mathcal{A}^* \to \mathcal{A}^* \otimes \mathcal{A}^*$ defined by
        \[
        \forall f \in \mathcal{A}^*, x, y \in \mathcal{A}: \langle \Delta_*(f), x\otimes y\rangle = \langle f, xy\rangle,
        \]
        \item[(d)] Antipode $S_*:\mathcal{A}^*\to \mathcal{A}^*$ defined by
        \[
        \forall f \in \mathcal{A}^*, x\in \mathcal{A}: \langle S_*(f), x\rangle = \langle f, S(x)\rangle,
        \]
        \item[(e)] Unit $\epsilon \in \mathcal{A}^*$ and co-unit $1 \in \mathcal{A}$,
    \end{enumerate}
    is a $C^*$-weak Hopf algebra.
\end{lemma}

Next, we will consider representations of $C^*$-weak Hopf algebras. Recall that a representation of a finite-dimensional $C^*$ algebra $\mathcal{A}$ is a linear map $\phi : \mathcal{A} \to \mathcal{M}(\mathbb{C}^{D})$ such that
\[
\phi(xy) = \phi(x)\phi(y), \; \phi(x^*) = \phi(x)^\dagger, \; \text{and } \phi(1) = 1.
\]
We will denote by $\text{Irr}(\mathcal{A})$ the equivalence classes of irreducible representations of $\mathcal{A}$. Every $C^*$ algebra is semisimple, i.e., $\mathcal{A} \cong \bigoplus_{a\in \text{Irr}(\mathcal{A})} \mathcal{M}(\mathbb{C}^{D_a})$ for some $D_a$. Therefore, any representation of $\mathcal{A}$, $\phi$, permits a decomposition of the form
\[
\phi(x) \cong \bigoplus_{a \in \text{Irr}(\mathcal{A})} \phi_a(x) \otimes I^{\otimes m_a},
\]
where $m_a$ is the multiplicity of $\phi_a$ in $\phi$, and $\phi_a$ is a representative of the irrep class $a$. Since, by lemma \ref{lemma:dual_C*}, the dual $\mathcal{A}^*$ is also a $C^*$ algebra, all representations $\psi:\mathcal{A}^* \to \mathcal{M}(\mathbb{C}^D)$ also permit a representation
\[
\psi(f)\cong \bigoplus_{a \in \text{Irr}(\mathcal{A}^*)} \psi_a(x) \otimes I^{\otimes m_a},
\]
where $m_a$ is the multiplicity of $\psi_a$ in $\psi$ and $\psi_a$ is a representation of the irrep class $a$. With every irrep class of the dual algebra $\mathcal{A}^*$, $a \in \text{Irr}(\mathcal{A}^*)$, we can associate $\tau_a \in \mathcal{A}$ which satisfies
\[
\langle f, \tau_a \rangle = \text{Tr}(\psi_a(f)) \ \forall \ f \in \mathcal{A}^*.
\]
Note that $\tau_a$ is independent of the exact choice of the representative $\psi_a$. The elements $\{\tau_a\}_{a \in \text{Irr}(\mathcal{A}^*)}$, called the \emph{cocentral elements} of the algebra $\mathcal{A}$, will play an important role in our construction of MPUs in the next subsection.

Given representations of $\mathcal{A}^*$, we can build new representation using the coproduct and the antipode for $\mathcal{A}^*$:
\begin{enumerate}
    \item[(1)] Suppose $\psi_1:\mathcal{A}^* \to \mathcal{M}(\mathbb{C}^{D_1})$, $\psi_2 :\mathcal{A}^* \to \mathcal{M}(\mathbb{C}^{D_2})$ are two representations of $\mathcal{A}^*$. Consider $\psi:\mathcal{A}^* \to \text{End}(\mathbb{C}^{D_1})\otimes \text{End}(\mathbb{C}^{D_2})$, $\psi(f) = (\psi_1 \otimes \psi_2) \circ \Delta_*(f)$ --- note however that $\psi(f)$ is not a valid representation since $\psi(\epsilon) \neq I$. However, note that $\psi(\epsilon)^2 = \psi(\epsilon)$ and $\psi(\epsilon) = \psi(\epsilon)^\dagger$, so $\psi(\epsilon)$ is an orthogonal projector --- furthermore, $\forall f\in \mathcal{A}^*$, $\psi(f) \psi(\epsilon) = \psi(\epsilon)\psi(f) = \psi(f)$. Thus, we can restrict $\psi$ to be a map from $\mathcal{A}^*$ to the image of $\psi(\epsilon)$ --- this representation will be denoted by $\psi_1 \boxtimes \psi_2$. This allows us to obtain ``fusion rules" since for $a, b \in \text{Irr}(\mathcal{A}^*)$,
    \[
    \psi_a \boxtimes \psi_b = \bigoplus_{c\in \text{Irr}(\mathcal{A}^*)} \psi_c \otimes I^{\otimes N_{a, b}^c}.
    \]
    From this, we obtain that
    \begin{align*}
        \text{Tr}((\psi_a \boxtimes \psi_b) (f))&=\sum_{c\in \text{Irr}(\mathcal{A}^*)} N_{a, b}^c \text{Tr}(\psi_c(f)) \\
        &= \sum_{c\in \text{Irr}(\mathcal{A}^*)} N_{a, b}^c \langle f, \tau_c\rangle.
    \end{align*}
    However, we can also write
    \begin{align*}
        \langle f, \tau_a \tau_b\rangle = \langle \Delta_*(f), \tau_a \otimes \tau_b\rangle = \text{Tr}((\psi_a \boxtimes \psi_b)(f)),
    \end{align*}
    and thus we obtain that
    \begin{align}\label{eq:product_rules_cocentral}
        \tau_a \tau_b = \sum_{c \in \text{Irr}(\mathcal{A}^*)} N_{a, b}^c \tau_c.
    \end{align}
    \item[(2)] Suppose $\psi \in \mathcal{A}^* \to \mathcal{M}(\mathbb{C}^D)$ is a representation of $\mathcal{A}^*$, then consider $\bar{\psi}:\mathcal{A}^* \to \mathcal{M}(\mathbb{C}^D)$ defined via $\bar{\psi}(f) = \psi(S_*(f))^\text{T}$ then $\bar{\psi}$ is also a representation of $\mathcal{A}^*$. Furthermore, since $S$ is a bijection, $\forall a \in \text{Irr}(\mathcal{A}^*)$, $\bar{\psi}_a$ is a irreducible representation and consequently, $\bar{\psi}_a \cong \psi_{\bar{a}}$ for some $\bar{a} \in \text{Irr}(\mathcal{A}^*)$. We note that the algebra element $\tau_{\bar{a}} \in \mathcal{A}$, $\forall f\in \mathcal{A}^*$, satisfies
    \begin{align*}
    \langle f, \tau_{\bar{a}}\rangle &= \text{Tr}[\psi(S_*(f))^\text{T}]
    = \text{Tr}[\psi(S_*(f))] \\
    &= \langle S_*(f), \tau_a\rangle =\langle f, S(\tau_a)\rangle,
    \end{align*}
    and therefore $\tau_{\bar{a}} = S(\tau_a)$. Finally, using $S(x^*)^* = S^{-1}(x)$ [Ref.~\cite{bohm1999weak}], we obtain that
    \begin{align*}
    \langle f, \tau_{\bar{a}}^*\rangle &= \langle (f^*)^*, \tau_{\bar{a}}^*\rangle = \overline{\langle f^*, S(\tau_{\bar{a}}^*)^*\rangle}, \\
    &= \overline{\langle f^*, S^{-1}(\tau_{\bar{a}})\rangle}  
    = \overline{\langle f^*, \tau_a\rangle}, \\
    &= \overline{\text{Tr}(\psi_a(f^*))} = \overline{\text{Tr}(\psi_a(f)^\dagger)} = \langle f, \tau_a\rangle. 
    \end{align*}
    Therefore, we conclude that
    \begin{align}\label{eq:star_rule_concentral}
    \tau_{\bar{a}} = \tau_a^*.
    \end{align}
\end{enumerate}
Finally, we remark that Eqs.~\eqref{eq:product_rules_cocentral} and \eqref{eq:star_rule_concentral} imply that the set $\{\tau_a\}_{a\in \text{Irr}(\mathcal{A}^*)}$ is closed under product and $*$ operations. Consequently, $\mathcal{T} = \text{span}(\{\tau_a\}_{a\in \text{Irr}(\mathcal{A}^*)}) \subseteq \mathcal{A}$ is also a $C^*$ algebra. Furthermore, since $\mathcal{A}$ is, by assumption, finite-dimensional, so is $\mathcal{T}$. Therefore, $\mathcal{T}$ is a direct sum of matrix algebras and consequently it has an identity element $\tau_I = \sum_{a\in \text{Irr}(\mathcal{A}^*)} e_a \tau_a \in \mathcal{T}$ ($e_a \in \mathbb{C}$) which will play an important rule in the next section. We remark that this identity element could, in general, be different from the algebra identity $1$.

\subsection{Generating MPUs from $C^*$-Weak Hopf algebras}\label{sec:MPU_weak_hopf}
\noindent We begin by reviewing the procedure provided in Ref.~\cite{molnar2022matrix} for generating an MPO from a $C^*$-weak Hopf algebra.
\begin{lemma}[Ref.~\cite{molnar2022matrix}]\label{lemma:MPO_construction}
    Given a $C^*$-weak Hopf algebra $\mathcal{A}$ together with a representation of $\mathcal{A}$, $\phi:\mathcal{A} \to \mathcal{M}(\mathbb{C}^d)$, then for $x \in \mathcal{A}$, the operator $\textnormal{MPO}_\phi(x)$ defined by
    \[
    \textnormal{MPO}_\phi(x) = \phi^{\otimes n}\circ \Delta^{n - 1}(x),
    \]
    is a matrix product operator which can be expressed as
    \[
    \textnormal{MPO}_\phi(x) = 
    \begin{array}{c}
        \begin{tikzpicture}[scale=.5,baseline={([yshift=-0.75ex] current bounding box.center)}]
		      \foreach \x in {0,...,0}{
                \SingleTrRight{(0,0)}
        }
		      \foreach \x in {1,...,1}{
                \GTensor{(-\singledx*\x,0)}{1}{.5}{\small $A_\phi$}{0}
                \myarrow{-\singledx*\x,-.65}{1};
                \myarrow{-\singledx*\x,.85}{1};
        }
        		      \foreach \x in {3,...,3}{
                \GTensor{(-\singledx*\x,0)}{1}{.5}{\small $A_\phi$}{0}
                \myarrow{-\singledx*\x,-.65}{1};
                \myarrow{-\singledx*\x,.85}{1};
        }
		      \foreach \x in {2,...,2}{
                \SingleDots{-\doubledx*\x-.5,0}{\doubledx*.8}
        }
		      \foreach \x in {4,...,4}{
                \GTensor{(-\singledx*\x,0)}{1}{.5}{\small $A_\phi$}{0}
                \myarrow{-\singledx*\x,-.65}{1};
                \myarrow{-\singledx*\x,.85}{1};
        }
		      \foreach \x in {5,...,5}{
                \bTensor{-\singledx*\x,0}{\small $b_x$}
        }
		      \foreach \x in {6,...,6}{
                \SingleTrLeft{(-\singledx*\x+.3,0)}
        }
        \end{tikzpicture}
        \end{array}
    \]
    where 
    \[
    \begin{array}{c}
        \begin{tikzpicture}[scale=.5,baseline={([yshift=-0.75ex] current bounding box.center)}]
		      \foreach \x in {1,...,1}{
                \GTensor{(-\singledx*\x,0)}{1}{.5}{\small $A_\phi$}{0}
                \myarrow{-\singledx*\x,-.65}{1};
                \myarrow{-\singledx*\x,.85}{1};
        }
        \end{tikzpicture}
        \end{array} = \sum_\alpha \phi(e_\alpha) \otimes \psi(e^\alpha)
    \]
    with $\psi : \mathcal{A}^* \to \mathcal{M}(\mathbb{C}^D)$ being a representation of $\mathcal{A}^*$ such that there exists a matrix $b_x \in \mathcal{M}(\mathbb{C}^d)$ satisfying $\textnormal{Tr}[\psi(f) b_x] = f(x) \ \forall \ f \in \mathcal{A}^*$.
\end{lemma}
\begin{proof}
    We will denote by $\{e_\alpha\}$ a basis for $\mathcal{A}$ and $\{e^\alpha\}$ a basis for $\mathcal{A}^*$, where the two sets are chosen such that $e^\alpha(e_\beta) = \delta_{\alpha, \beta}$. Explicitly evaluating the tensor contraction, we obtain that
    \begin{align*}
    &\begin{array}{c}
        \begin{tikzpicture}[scale=.5,baseline={([yshift=-0.75ex] current bounding box.center)}]
		      \foreach \x in {0,...,0}{
                \SingleTrRight{(0,0)}
        }
		      \foreach \x in {1,...,1}{
                \GTensor{(-\singledx*\x,0)}{1}{.5}{\small $A_\phi$}{0}
                \myarrow{-\singledx*\x,-.65}{1};
                \myarrow{-\singledx*\x,.85}{1};
        }
        		      \foreach \x in {3,...,3}{
                \GTensor{(-\singledx*\x,0)}{1}{.5}{\small $A_\phi$}{0}
                \myarrow{-\singledx*\x,-.65}{1};
                \myarrow{-\singledx*\x,.85}{1};
        }
		      \foreach \x in {2,...,2}{
                \SingleDots{-\doubledx*\x-.5,0}{\doubledx*.8}
        }
		      \foreach \x in {4,...,4}{
                \GTensor{(-\singledx*\x,0)}{1}{.5}{\small $A_\phi$}{0}
                \myarrow{-\singledx*\x,-.65}{1};
                \myarrow{-\singledx*\x,.85}{1};
        }
		      \foreach \x in {5,...,5}{
                \bTensor{-\singledx*\x,0}{\small $b_x$}
        }
		      \foreach \x in {6,...,6}{
                \SingleTrLeft{(-\singledx*\x+.3,0)}
        }
        \end{tikzpicture}
        \end{array} 
        \\ &= \phi^{\otimes n}  \sum_{\alpha_1,\alpha_2  \dots, \alpha_n} \big(e_{\alpha_1} \otimes e_{\alpha_2}  \otimes \dots e_{\alpha_n})\times \nonumber\\
        &\qquad \qquad \qquad \qquad \text{Tr}[b_x\psi(e^{\alpha_1}e^{\alpha_2}\dots e^{\alpha_n})].
    \end{align*}
    Furthermore, by choice of $b_x$, 
    \begin{align*}
    \text{Tr}[b_x \psi(e^{\alpha_1}e^{\alpha_2} \dots e^{\alpha_n}) ] &= (e^{\alpha_1}e^{\alpha_2} \dots e^{\alpha_n})(x) \\
    &= (e^{\alpha_1} \otimes e^{\alpha_2} \dots e^{\alpha_n})\circ \Delta^{n - 1}(x),
    \end{align*}
    and consequently
    \begin{align*}
    &\begin{array}{c}
        \begin{tikzpicture}[scale=.5,baseline={([yshift=-0.75ex] current bounding box.center)}]
		      \foreach \x in {0,...,0}{
                \SingleTrRight{(0,0)}
        }
		      \foreach \x in {1,...,1}{
                \GTensor{(-\singledx*\x,0)}{1}{.5}{\small $A_\phi$}{0}
                \myarrow{-\singledx*\x,-.65}{1};
                \myarrow{-\singledx*\x,.85}{1};
        }
        		      \foreach \x in {3,...,3}{
                \GTensor{(-\singledx*\x,0)}{1}{.5}{\small $A_\phi$}{0}
                \myarrow{-\singledx*\x,-.65}{1};
                \myarrow{-\singledx*\x,.85}{1};
        }
		      \foreach \x in {2,...,2}{
                \SingleDots{-\doubledx*\x-.5,0}{\doubledx*.8}
        }
		      \foreach \x in {4,...,4}{
                \GTensor{(-\singledx*\x,0)}{1}{.5}{\small $A_\phi$}{0}
                \myarrow{-\singledx*\x,-.65}{1};
                \myarrow{-\singledx*\x,.85}{1};
        }
		      \foreach \x in {5,...,5}{
                \bTensor{-\singledx*\x,0}{\small $b_x$}
        }
		      \foreach \x in {6,...,6}{
                \SingleTrLeft{(-\singledx*\x+.3,0)}
        }
        \end{tikzpicture}
        \end{array} \\
        &= \phi^{\otimes n} \bigg(\sum_{\alpha_1, \dots, \alpha_n} e_{\alpha_1} e^{\alpha_1}\otimes e_{\alpha_2} e^{\alpha_2} \dots e_{\alpha_n} e^{\alpha_n} \bigg) \circ \Delta^{n - 1}(x).
    \end{align*}
    Using this together with fact that $\sum_\alpha e_\alpha e^\alpha = 1$ establishes the lemma statement.
\end{proof}
\noindent Of special interest are the MPOs generated by the cocentral elements $\tau_a \in \mathcal{A}$ corresponding to the Irrep classes of $\mathcal{A}^*$.
\begin{lemma}\label{lemma:mpo_irrep}
    If $\phi :\mathcal{A} \in \textnormal{End}(\mathbb{C}^d)$ is an injective representation, the matrix product operators $O_a = \textnormal{MPO}_\phi(\tau_a)$ have the following properties:
    \begin{enumerate}
        \item[(a)] The MPO $O_a$ can be represented as
        \[
         O_a = \begin{array}{c}
        \begin{tikzpicture}[scale=.5,baseline={([yshift=-0.75ex] current bounding box.center)}]
		      \foreach \x in {0,...,0}{
                \SingleTrRight{(0,0)}
        }
		      \foreach \x in {1,...,1}{
                \GTensor{(-\singledx*\x,0)}{1}{.5}{\small $A_a$}{0}
                \myarrow{-\singledx*\x,-.65}{1};
                \myarrow{-\singledx*\x,.85}{1};
        }
        		      \foreach \x in {3,...,3}{
                \GTensor{(-\singledx*\x,0)}{1}{.5}{\small $A_a$}{0}
                \myarrow{-\singledx*\x,-.65}{1};
                \myarrow{-\singledx*\x,.85}{1};
        }
		      \foreach \x in {2,...,2}{
                \SingleDots{-\doubledx*\x-.5,0}{\doubledx*.8}
        }
		      \foreach \x in {4,...,4}{
                \GTensor{(-\singledx*\x,0)}{1}{.5}{\small $A_a$}{0}
                \myarrow{-\singledx*\x,-.65}{1};
                \myarrow{-\singledx*\x,.85}{1};
        }
		      \foreach \x in {5,...,5}{
                \SingleTrLeft{(-\singledx*\x+.3,0)}
        }
        \end{tikzpicture}
        \end{array},
        \]
        where $A_a$ are injective tensors and $\{O_a\}_{a \in \textnormal{Irr}(\mathcal{A}^*)}$ generate linearly independent MPOs for sufficiently large $N$.
        \item[(b)] The MPOs $O_a$ satisfy
        \[
        O_a O_b = \sum_c N_{a, b}^c O_c \text{ and } O_a^\dagger = O_{\bar{a}}.
        \]
    \end{enumerate}
\end{lemma}
\begin{proof}
    (a) To construct $O_a = \text{MPO}_\phi(\tau_a)$, we choose $\psi \cong \psi_a$ in lemma \ref{lemma:MPO_construction} and since $\forall f \in \mathcal{A}^*: f(\tau_a) = \text{Tr}(\psi(f))$, this allows us to choose $b(\tau_a) = \1$. Thus $O_a$ is an MPO with periodic boundary conditions generated by the tensor $A_a$ given by
    \[
    \begin{array}{c}
        \begin{tikzpicture}[scale=.5,baseline={([yshift=-0.75ex] current bounding box.center)}]
              \foreach \x in {1,...,1}{
                \GTensor{(-\singledx*\x,0)}{1}{.5}{\small $A_a$}{0}
                \myarrow{-\singledx*\x,-.65}{1};
                \myarrow{-\singledx*\x,.85}{1};
        }
        \end{tikzpicture}
    \end{array} = \sum_\alpha \phi(e_\alpha) \otimes \psi_a(e^\alpha)
    \]
    Next, to show that $A_a$, seen as a map from $\mathbb{C}^{D_a} \otimes \mathbb{C}^{D_a} \to \mathbb{C}^d \otimes \mathbb{C}^d$, is an injective tensor, we note that for $X \in \mathbb{C}^{D_a \times D_a}$, suppose $Y$ is given by
    \[
    \begin{array}{c}
        Y := \begin{tikzpicture}[scale=.5,baseline={([yshift=-0.75ex] current bounding box.center)}]
		      \foreach \x in {0,...,0}{
                \SingleTrRight{(0,0)}
        }
		      \foreach \x in {1,...,1}{
                \GTensor{(-\singledx*\x,0)}{1}{.5}{\small $A_a$}{0}
                \myarrow{-\singledx*\x,-.65}{1};
                \myarrow{-\singledx*\x,.85}{1};
        }
		      \foreach \x in {2,...,2}{
                \bTensor{-\singledx*\x,0}{\small $X$}
        }
		      \foreach \x in {3,...,3}{
                \SingleTrLeft{(-\singledx*\x+.3,0)}
        }
        \end{tikzpicture}
        \end{array} = \sum_\alpha \text{Tr}[X \psi_a(e^{\alpha})] \phi(e_\alpha).
    \]
    Note that since $\phi$ is injective, the operators $\{\phi(e_\alpha)\}_\alpha$ are linearly independent and thus given $Y$, we can uniquely obtain $\text{Tr}[X\psi_a(e^\alpha)]$. Furthermore, since $\psi_a$ is an irreducible representation, $\{\psi_a(e^\alpha)\}_\alpha$ are a complete, although not necessarily linearly independent, basis. Hence $\text{Tr}[X\psi_a(e^\alpha)]$ unqiuely determines $X$.

    Next, we show that $\{O_a\}_{a \in \text{Irr}(\mathcal{A}^*)}$ are linearly independent for sufficiently large $N$ --- for this, it is sufficient to show that the tensors $\{A_a\}_{a \in \text{Irr}(\mathcal{A}^*)}$ are not related to each other by a similarity transformation on the bond-space. To show this, let us assume that this is not the case, i.e., $\exists X$ such that for $a \neq b$,
    \[
    \begin{array}{c}
     \begin{tikzpicture}[scale=.5,baseline={([yshift=-0.75ex] current bounding box.center)}]
		      \foreach \x in {0,...,0}{
                \GTensor{(-\singledx*\x,0)}{1}{.5}{\small $A_a$}{0}
                \myarrow{-\singledx*\x,-.65}{1};
                \myarrow{-\singledx*\x,.85}{1};
        }
		      \foreach \x in {1,...,1}{
                \bTensor{-\singledx*\x,0}{\small $X$}
        }
                \foreach \x in {-1,...,-1}{
                \bTensor{-\singledx*\x,0}{\scriptsize $X^{\!-\!1}$}
        }
        \end{tikzpicture} =  \begin{tikzpicture}[scale=.5,baseline={([yshift=-0.75ex] current bounding box.center)}]
		      \foreach \x in {0,...,0}{
                \GTensor{(-\singledx*\x,0)}{1}{.5}{\small $A_b$}{0}
                \myarrow{-\singledx*\x,-.65}{1};
                \myarrow{-\singledx*\x,.85}{1};
        }
        \end{tikzpicture}
        \end{array}
        \;.
    \]
    Again, using the fact that $\phi$ is injective and thus $\{\phi(e_\alpha)\}_{\alpha}$ are linearly independent, we obtain that this implies $X \psi_a(e^\alpha) X^{-1} = \psi_b(e^\alpha) \implies \forall f \in \mathcal{A}^*: X \psi_a(f) X^{-1} = \psi_b(f)$ which contradicts the fact that $\psi_a$ and $\psi_b$ belong to different Irrep classes.

    (b) follows directly from the fact that $O_a = \phi^{\otimes n} \circ \Delta^{n - 1}(\tau_a)$ and thus
    \begin{align*}
    O_a O_b &= (\phi^{\otimes n} \circ \Delta^{n - 1}(\tau_a)) (\phi^{\otimes n} \circ \Delta^{n - 1}(\tau_a)) \\
    &= \phi^{\otimes n} \circ \Delta^{n - 1}(\tau_a \tau_b) = \sum_c N_{a, b}^c \phi^{\otimes n} \circ \Delta^{n - 1}(\tau_c) \\
    &= \sum_c N_{a, b}^c O_c.
    \end{align*}
    Similarly, 
    \begin{align*}
    O_a^\dagger &= (\phi^{\otimes n} \circ \Delta^{n- 1}(\tau_a))^\dagger \\
    &= \phi^{\otimes n}\circ \Delta^{n- 1}(\tau_a^*) = \phi^{\otimes n}\circ \Delta^{n- 1}(\tau_{\bar{a}}) = O_{\bar{a}}.
    \end{align*}
\end{proof}
\noindent Finally, we construct MPUs from $C^*$-weak Hopf algebras. These correspond to unitary elements in the $*$-algebra $\text{span}(\{\tau_a\}_{a\in \text{Irr}(\mathcal{A}^*)})$.
\begin{proposition}\label{prop:mpu_final_hopf}
    Consider the $C^*$-algebra $\mathcal{T} =\textnormal{span}(\{\tau_a\}_{a\in \mathcal{A}^*})$. Suppose $\tau_I = \sum_a e_a \tau_a$ ($e_a \in \mathbb{C}$) is the identity element of $\mathcal{T}$ and $u = \sum_a u_a \tau_a \in \mathcal{T}$ ($u_a \in \mathbb{C}$) is an element satisfying $u^* u =  uu^* =  \tau_I$. Then for any injective representation $\phi:\mathcal{A} \to \mathcal{M}(\mathbb{C}^d)$, 
    \[
    U = \1^{\otimes N}+ \sum_{a \in \textnormal{Irr}(\mathcal{A}^*)} (u_a - e_a) \textnormal{MPO}_\phi(\tau_a),
    \]
    is an MPU with physical dimension $d$, with a translationally invariant bulk tensor $A$ and boundary tensor $b$ such that $A = \bigoplus_{k} A_k$ and $b = \bigoplus_k b_k$ where $A_k$ are injective and independent tensors and $b_k$ are full rank matrices.
\end{proposition}
\begin{proof}
We define 
\begin{align*}
P &= \sum_{a \in \text{Irr}(\mathcal{A}^*)} e_a O_a = \phi^{\otimes n} \circ \Delta^{n - 1}(\tau_I)  \\
V &= \sum_{a \in \text{Irr}(\mathcal{A}^*)} u_a O_a = \phi^{\otimes n} \circ \Delta^{n - 1}(u).
\end{align*}
We then note that since $\tau_I^2 = \tau_I$ and $\tau_I^\dagger = \tau_I$, $P^2 = P$ and $P = P^\dagger$, so $P$ is an orthogonal projector. Furthermore, since $u \tau_I = \tau_I u = u$ and $u^*u = uu^* = \tau_I$, $V P = P V = V$ and $V^\dagger V = VV^\dagger = P$, i.e., $V$ is a unitary when restricted to the image of $P$. Suppose $Q = \1^{\otimes N} - P$,  it then immediately follows that $U = Q +  V$ is a unitary.

Finally, $U$ can be explicitly represented with
\begin{align*}
A^{ij} &= \delta_{ij} \oplus \bigg(\bigoplus_{a\in \text{Irr}(\mathcal{A}^*)} A_a^{ij}\bigg) \\
b &= 1 \oplus \bigg(\bigoplus_{a\in \text{Irr}(\mathcal{A}^*)} (e^{i\alpha} u_a - e_a) I_{D_a}\bigg).
\end{align*}
We remark that it is possible that, depending on the choice of $u$, $ u_a - e_a = 0$ for some $a \in \text{Irr}(\mathcal{A}^*)$ --- in that case, we can neglect this block in the representation of $U$.
\end{proof}

An explicit method for constructing unitary elements $u$ would be to construct elements that correspond to projectors in the algebra $\mathcal{T} = \text{span}(\{\tau_a\}_{a\in \text{Irr}(\mathcal{A}^*)})$. Abstractly, since $\mathcal{T}$ is a finite-dimensional $C^*$ algebra, it is isomorphic to a direct sum of matrix algebras, i.e., $\mathcal{T} \cong \bigoplus_{i = 1}^m \mathcal{M}(\mathbb{C}^{d_i})$. Consider the elements $p_1, p_2\dots p_i \in \mathcal{T}$ defined by the correspondence 
\[
p_i \cong 0 \oplus  \dots \oplus \1_{d_i} \oplus\dots \oplus  0.
\]
Then, $\{p_i\}_{i \in \{1, 2\dots m\}}$ are a complete set of orthogonal projectors in $\mathcal{T}$ i.e.~$p_i^* = p_i$, $p_i p_j = p_j p_i = \delta_{i, j} p_i$ and $\sum_{i = 1}^m p_i =\tau_I$. Practically, we can explicitly determine $p_i$ by the following process: assume the ansatz $p_i = \sum_{a \in \text{Irr}(\mathcal{A}^*)} P_{i, a}\tau_a$ for some $P_{i, a} \in \mathbb{C}$ and imposing the conditions $p_i = p_i^*$ and $p_{i}^2 = p_i$ which, together with the rules in Eqs.~\eqref{eq:product_rules_cocentral} and \eqref{eq:star_rule_concentral} yield
\begin{align}
    P_{i, a}^* = P_{i, \bar{a}} \text{ and } P_{i, c}\delta_{i, j} = \sum_{a, b \in \text{Irr}(\mathcal{A}^*)}N_{a, b}^c P_{i, a} P_{j, b}.
\end{align}
Solving these equations allows us to explicitly construct the elements $p_1, p_2\dots p_m$. Once the elements $p_1, p_2 \dots p_m$ are constructed, we can then construct certain unitary elements $u$ via
\[
u = \sum_{i = 1}^{m}e^{i\phi_i} p_i,
\]
which can be interpreted as a unitary element that applies different phases on the orthogonal subspaces corresponding to the projectors $p_i$.

\subsection{Example: MPUs from Lee-Yang Weak Hopf Algebra}
As a concrete example, we consider the Lee-Yang algebra $\mathcal{A} = \mathcal{M}(\mathbb{C}^2) \oplus \mathcal{M}(\mathbb{C}^3)$. We will denote the standard basis for this algebra by $e_\alpha$ where the index $\alpha =(k; a, b)$ with $k \in \{1, 2\}$ and
\begin{align*}
&e_{(1; a, b)} = \ket{a}\!\bra{b} \oplus 0 \text{ for }a, b \in \{1, 2\} \text{, and }\\
&e_{(2; a, b)} = 0\oplus \ket{a}\!\bra{b} \text{ for }a, b \in \{1, 2, 3\}.
\end{align*}
The product and $*$ operation on $e_{k; a, b}$ is defined in the usual way for matrices:
\begin{align}\label{eq:primal_product_star}
e_{(k;a, b)} e_{(k'; a', b')} = \delta_{k, k'}\delta_{a', b}e_{(k; a, b')} \text{ and }e_{(k; a, b)}^* = e_{(k; b, a)}.
\end{align}
The coproduct $\Delta$ for the Lee-Yang algebra, defined on the basis elements $e_{(k; a, b)}$, is given explicitly by \cite{ruiz2024matrix}
\begin{align}\label{eq:coprodct_lee_yang}
    \Delta(e_{(1; 1, 1)}) &:= e_{(1; 1, 1)} \otimes e_{(1; 1, 1)} + e_{(2; 1, 1)} \otimes e_{(2; 2, 2)},\nonumber \\
\Delta(e_{(1; 1, 2)}) &:= e_{(1; 1, 2)} \otimes e_{(1; 1, 2)} + \zeta^2 e_{(2; 1, 2)} \otimes e_{(2; 2, 1)} +\nonumber\\
&\quad \zeta e_{(2; 1, 3)} \otimes e_{(2; 2, 3)},\nonumber \\
\Delta(e_{(1; 2, 2)}) &:= e_{(1; 2, 2)} \otimes e_{(1; 2, 2)} + \zeta^4 e_{(2; 2, 2)} \otimes e_{(2; 1, 1)} + \nonumber \\
&\quad \zeta^3 e_{(2; 2, 3)} \otimes e_{(2; 1, 3)} + \zeta^3 e_{(2; 3, 2)} \otimes e_{(2; 3, 1)} +\nonumber\\
&\quad \zeta^2 e_{(2; 3, 3)} \otimes e_{(2; 3, 3)}, \nonumber \\
\Delta(e_{(2; 1, 1)}) &:= e_{(1; 1, 1)} \otimes e_{(2; 1, 1)} + e_{(2; 1, 1)} \otimes e_{(1; 2, 2)} + \nonumber\\
&\quad e_{(2; 1, 1)} \otimes e_{(2; 3, 3)}, \nonumber \\
\Delta(e_{(2; 1, 2)}) &:= e_{(1; 1, 2)} \otimes e_{(2; 1, 2)} + e_{(2; 1, 2)} \otimes e_{(1; 2, 1)} + \nonumber\\
&\quad e_{(2; 1, 3)} \otimes e_{(2; 3, 2)},\nonumber \\
\Delta(e_{(2; 1, 3)}) &:= e_{(1; 1, 2)} \otimes e_{(2; 1, 3)} + e_{(2; 1, 1)} \otimes e_{(1; 2, 2)} +\nonumber\\
&\quad \zeta e_{(2; 1, 2)} \otimes e_{(2; 3, 1)} - \zeta^2 e_{(2; 1, 3)} \otimes e_{(2;3, 3)}, \nonumber \\
\Delta(e_{(2; 2, 2)}) &:= e_{(1; 2, 2)} \otimes e_{(2; 2, 2)} + e_{(2; 2, 2)} \otimes e_{(1; 1, 1)} \nonumber\\
&\quad + e_{(2; 3, 3)} \otimes e_{(2; 2, 2)}, \nonumber \\
\Delta(e_{(2; 2, 3)}) &:= e_{(1; 2, 2)} \otimes e_{(2; 2, 3)} + e_{(2; 2, 3)} \otimes e_{(1; 2, 1)} + \nonumber\\
&\quad \zeta e_{(2; 3, 2)} \otimes e_{(2; 2, 1)} - \zeta^2 e_{(2; 3, 3)} \otimes e_{(2; 2, 3)}, \nonumber \\
\Delta(e_{(2; 3, 3)}) &:= e_{(1; 2, 2)} \otimes e_{(2; 3, 3)} + e_{(2; 3, 3)} \otimes e_{(1; 2, 2)} +\nonumber\\
&\quad \zeta^2 e_{(2; 2, 2)} \otimes e_{(2; 1, 1)} - \zeta^3 e_{(2; 2, 3)} \otimes e_{(2; 1, 3)} \nonumber\\
&\quad - \zeta^3 e_{(2; 3, 2)} \otimes e_{(2; 3, 1)} + \zeta^4 e_{(2; 3, 3)} \otimes e_{(2; 3, 3)},
\end{align}
where $\zeta = [(\sqrt{5} - 1)/2]^{1/2}$. The antipode $S$ is defined by \cite{ruiz2024matrix}
\begin{align}\label{eq:antipode_lee_yang}
    &S(e_{(1; 1, 1)}) = e_{(1; 1, 1)}, S(e_{(1; 2, 2)}) = e_{(1; 2, 2)},\nonumber\\
    &S(e_{(1; 1, 2)}) = e_{(1; 2, 1)}, S(e_{(1; 2, 1)}) = e_{(1; 1, 2)},\nonumber \\
    &S(e_{(2; 1, 1)}) = e_{(2; 2, 2)}, S(e_{(2; 2, 2)}) = e_{(2; 1, 1)},\nonumber \\
    &S(e_{(2; 3, 3)}) = e_{(2; 3, 3)}, \nonumber\\
    &S(e_{(2; 1, 2)}) = \zeta^{-1} e_{(2; 2, 1)}, S(e_{(2; 2, 1)}) = \zeta^{1} e_{(2; 1, 2)}, \nonumber \\
    &S(e_{(2; 1, 3)}) = \zeta^{-2} e_{(2; 2, 3)}, S(e_{(2; 3, 1)}) = \zeta^2 e_{(2;3, 2)}, \nonumber \\
    &S(e_{(2; 2, 3)}) = \zeta^{-1}e_{(2; 1, 3)}, S(e_{(2; 3, 2)}) = \zeta e_{(2; 3, 1)}.
\end{align}

Next, we construct the dual algebra $\mathcal{A}^*$ --- we begin with noting that, as a linear vector space, $\mathcal{A}$ is self-dual, i.e., $\mathcal{A}^* = \mathcal{M}(\mathbb{C}^2) \oplus \mathcal{M}(\mathbb{C}^3)$. The action of an element in the dual, $f \in \mathcal{M}(\mathbb{C}^2) \oplus \mathcal{M}(\mathbb{C}^3)$, on an element in the algebra $x \in \mathcal{M}(\mathbb{C}^2) \oplus \mathcal{M}(\mathbb{C}^3)$ is defined by
\begin{align}\label{eq:dual_action}
f(x) = \langle f, x\rangle = \sum_{k \in \{1, 2\}} \text{Tr}(f_k^\dagger x_k),
\end{align}
where $f = f_1 \oplus f_2$ and $x = x_1 \oplus x_2$. This definition of $f(x)$ also allows us to construct the dual basis elements $e^\alpha \cong e^{(1; a, b)} \in \mathcal{M}(\mathbb{C}^2) \oplus   \mathcal{M}(\mathbb{C}^3)$ where $k \in \{1, 2\}$ and
\begin{align}\label{eq:dual_basis_def}
&e^{(1; a, b)} = \ket{a}\!\bra{b} \oplus 0 \text{ for }a, b \in \{1, 2\} \text{, and }\nonumber \\
&e^{(2; a, b)} = 0\oplus \ket{a}\!\bra{b} \text{ for }a, b \in \{1, 2, 3\}.
\end{align}
From Eq.~\eqref{eq:dual_action}, it then follows that
\[
\langle e^{(k; a, b)}, e_{(k'; a', b')}\rangle = \delta_{k, k'}\delta_{a, a'} \delta_{b, b'}.
\]

At this point, we emphasize a subtlety with the dual basis --- from the expressions for the dual basis elements in Eq.~\eqref{eq:dual_basis_def}, it could be tempting to conclude that the product and $*$ operations on the dual basis elements $e^{(k; a, b)}$ are defined similar to that of a matrix algebra, i.e., Eq.~\eqref{eq:primal_product_star}. However, we remark that the dual $\mathcal{A}^*$ is $\mathcal{M}(\mathbb{C}^2) \oplus \mathcal{M}(\mathbb{C}^3)$ only as a \emph{linear space} and not as an algebra. Therefore, as described in \cref{lemma:dual_C*}, the product on $\mathcal{A}^*$ is in fact defined by the co-product $\Delta$ and the $*$ on $\mathcal{A}^*$  by the antipode $S$. More specifically, if
\[
\Delta(e_{\alpha}) = \sum_{\gamma, \delta} \text{D}^\alpha_{\gamma, \delta} e_{\gamma}\otimes e_{\delta},
\]
where $\text{D}_{\gamma, \delta}^\alpha \in \mathbb{C}$ can be read off from Eq.~\eqref{eq:coprodct_lee_yang}, then from \cref{lemma:dual_C*}(c) it follows that
\begin{align}\label{eq:lee_yang_product_dual}
e^{\gamma} e^{\delta} = \sum_{\alpha}\text{D}_{\gamma, \delta}^\alpha e^{\alpha}. 
\end{align}
Similarly, the $*$ operation on the dual is defined by the antipode $S$ via \cref{lemma:dual_C*}(d). More specifically, if
\[
S(e_\alpha) = \sum_{\beta} \text{S}_\beta^\alpha e_\beta,
\]
where $\text{S}^\alpha_\beta \in \mathbb{C}$ can be read off from Eq.~\eqref{eq:antipode_lee_yang}, then
\begin{align}\label{eq:lee_yang_product_star}
e^{\beta*} = \sum_\alpha (\text{S}_\beta^\alpha)^* e^\alpha.
\end{align}
While Eqs.~\eqref{eq:lee_yang_product_dual} and \eqref{eq:lee_yang_product_star}, in principle, contains all the information needed to multiply and $*$ elements in $\mathcal{A}^*$, it is convenient to use a different basis for the dual, $\tilde{e}^{(k; a, b)}$ defined explicitly via \cite{liu2025parent}
\begin{align}\label{eq:new_basis_dual}
&e^{(1; 1, 1)} = \tilde{e}^{(1; 1, 1)}, e^{(1; 1, 2)} = \tilde{e}^{(2; 1, 1)}, e^{(1; 2, 1)} = \tilde{e}^{(2; 2, 2)}, \nonumber\\
&e^{(1; 2, 2)} = \tilde{e}^{(1; 2, 2)} + \tilde{e}^{(2; 3, 3)}, \nonumber \\
&e^{(2; 1, 1)} = \tilde{e}^{(1; 1, 2)}, e^{(2; 1, 2)} = \tilde{e}^{(2; 1, 2)}, e^{(2; 2, 1)} = \zeta^{-2}\tilde{e}^{(2; 2, 1)}, \nonumber \\
&e^{(2; 1, 3)} = \tilde{e}^{(2; 1, 3)}, e^{(2; 3, 1)} = \zeta \tilde{e}^{(2; 2, 3)}, e^{(2; 2, 2)} = \tilde{e}^{(1; 2, 1)}, \nonumber \\
&e^{(2; 2, 3)} = \zeta \tilde{e}^{(2; 3, 1)}, e^{(2; 3, 2)} = \tilde{e}^{(2; 3, 2)}, \nonumber \\
&e^{(2; 3, 3)} = \tilde{e}^{(1; 2, 2)} - \zeta^2 \tilde{e}^{(2;3, 3)}.
\end{align}
It can be checked that the product and $*$ rules for $\tilde{e}^{(k; a, b)}$ are similar to Eq.~\eqref{eq:primal_product_star} i.e.,
\begin{align}\label{eq:dual_product_star_new}
    \tilde{e}^{(k; a, b)} \tilde{e}^{(k'; a', b')} = \delta_{k, k'} \delta_{a', b} e^{(k; a, b')} \text{ and }\tilde{e}^{(k; a, b)*} = \tilde{e}^{(k; b, a)}.
\end{align}
Furthermore, we can also construct a basis for the algebra $\mathcal{A}$, $\tilde{e}_{(k;a, b)}$ which are dual to the basis $\tilde{e}^{(k; a, b)}$, i.e., they satisfy
\begin{align}\label{eq:new_basis}
\langle \tilde{e}^{(k; a, b)}, \tilde{e}_{(k'; a', b')}\rangle = \delta_{k, k'} \delta_{a, a'}\delta_{b, b'}.
\end{align}
From Eq.~\eqref{eq:new_basis}, it then follows that Eqs.~\eqref{eq:new_basis_dual} holds with the substitutions $e^{(k; a, b)} \to \tilde{e}_{(k; a, b)}$ and $\tilde{e}^{(k; a, b)} \to e_{(k; a, b)}$, thus yielding a conversion between the basis $\tilde{e}_{(k;a, b)}$ and $e_{(k; a, b)}$:
\begin{align}\label{eq:new_basis_primal}
&\tilde{e}_{(1; 1, 1)} = {e}_{(1; 1, 1)}, \tilde{e}_{(1; 1, 2)} = {e}_{(2; 1, 1)}, \tilde{e}_{(1; 2, 1)} = {e}_{(2; 2, 2)}, \nonumber\\
&\tilde{e}_{(1; 2, 2)} = {e}_{(1; 2, 2)} + {e}_{(2; 3, 3)}, \nonumber \\
&\tilde{e}_{(2; 1, 1)} = {e}_{(1; 1, 2)}, \tilde{e}_{(2; 1, 2)} = {e}_{(2; 1, 2)}, \tilde{e}_{(2; 2, 1)} = \zeta^{-2}{e}_{(2; 2, 1)}, \nonumber \\
&\tilde{e}_{(2; 1, 3)} = {e}_{(2; 1, 3)}, \tilde{e}_{(2; 3, 1)} = \zeta {e}_{(2; 2, 3)}, \tilde{e}_{(2; 2, 2)} = {e}_{(1; 2, 1)}, \nonumber \\
&\tilde{e}_{(2; 2, 3)} = \zeta {e}_{(2; 3, 1)}, \tilde{e}_{(2; 3, 2)} = {e}_{(2; 3, 2)}, \nonumber \\
&\tilde{e}_{(2; 3, 3)} = {e}_{(1; 2, 2)} - \zeta^2 {e}_{(2;3, 3)}.
\end{align}
We again remark that the product and $*$ rules for this new basis for the algebra $\mathcal{A}$, $\tilde{e}_{(k; a, b)}$, is \emph{not} given by Eq.~\eqref{eq:primal_product_star} --- to multiply or $*$ $\tilde{e}_{(k; a, b)}$, the most convenient is to use Eq.~\eqref{eq:new_basis_primal} to express them in terms of ${e}_{(k; a, b)}$ and then use the usual product and star rules in Eq.~\eqref{eq:primal_product_star}.

\prlsection{Irrep classes of $\mathcal{A}^*$}The product and $*$ rules of the modified basis $\tilde{e}^{(k; a, b)}$ in Eq.~\eqref{eq:dual_product_star_new} mimic that of a matrix algebra --- this allows us to immediately identify the two irreducible irrepresentations of $\mathcal{A}^*$, $\psi_e:\mathcal{M}(\mathbb{C}^2) \oplus \mathcal{M}(\mathbb{C}^3) \to \mathcal{M}(\mathbb{C}^2) $ and $\psi_\sigma: \mathcal{M}(\mathbb{C}^2) \oplus \mathcal{M}(\mathbb{C}^3) \to \mathcal{M}(\mathbb{C}^3) $, given by
\begin{align}\label{eq:irrep_classes_lee_yang}
    &\psi_e(\tilde{e}^{(1; a, b)}) = \ket{a}\!\bra{b}, \psi_e(\tilde{e}^{(2; a, b)}) = 0 \text{ and}, \nonumber\\
    &\psi_\sigma(\tilde{e}^{(1; a, b)} = 0, \psi_e(\tilde{e}^{(2; a, b)}) = \ket{a}\!\bra{b}.
\end{align}
We will label the irrep class corresponding to $\psi_e$ via $e$ and the irrep class corresponding to $\psi_\sigma$ by $\sigma$. We can now compute the cocentral elements $\tau_e$ and $\tau_\sigma$ corresponding to the two irrep classes --- recall that for $a \in \text{Irr}(\mathcal{A}^*)$ , $\tau_a$ is an element of the algebra $\mathcal{A}$ satisfying $\langle f, \tau_a\rangle = \text{Tr}[\psi_a(f)] \ \forall f \in \mathcal{A}^*$. Corresponding to $\psi_e$ defined in Eq.~\eqref{eq:irrep_classes_lee_yang}, 
\[\tau_e = \tilde{e}_{(1; 1, 1)} + \tilde{e}_{(1; 2, 2)},
\]
since for $ f = \sum_{k, a, b} f_{(k; a, b)}\tilde{e}^{(k; a, b)}$, 
\begin{align*}
    \text{Tr}(\psi_e(f)) &= \sum_{k, a, b} f_{(k; a, b)} \tilde{e}^{(k;a, b)}, \nonumber\\
    &= \sum_{k, a, b}f_{(k; a, b)}  \delta_{a, b} \delta_{k, 1} =\langle f, \tilde{e}_{(1; 1, 1)} + \tilde{e}_{(1; 2, 2)}\rangle.
\end{align*}
Similarly, corresponding to $\psi_\sigma$ defined in Eq.~\eqref{eq:irrep_classes_lee_yang}, 
\[
\tau_\sigma = \tilde{e}_{(2; 1, 1)} + \tilde{e}_{(2; 2, 2)} + \tilde{e}_{(2; 3, 3)}. 
\]
It is convenient to express $\tau_e, \tau_\sigma$ in terms of the basis $e_{(k; a, b)}$ using Eq.~\eqref{eq:new_basis_primal}:
\begin{align*}
    &\tau_e = e_{(1; 1, 1)} + e_{(1; 2, 2)} + e_{(2; 3, 3)} \text{ and},\\
    &\tau_\sigma = e_{(1; 1, 2)} + e_{(1; 2, 1)} + e_{(1; 2, 2)} - \zeta^2 e_{(2; 3, 3)}.
\end{align*}
From these explicit expressions, we can compute the product and $*$ rules for $\tau_e, \tau_\sigma$ which are often called the ``fusion rules" for Fibonacci anyons in physics literature:
\begin{align}\label{eq:taus_rules_lee_yang}
    &\tau_e^2 = \tau_e, \tau_e \tau_\sigma = \tau_\sigma \tau_e = \tau_\sigma,  \tau_\sigma^2 = \tau_e + \tau_\sigma \text{, and } \tau_e^* = \tau_e, \tau_\sigma^* = \tau_\sigma. 
\end{align}
Thus, as anticipated from the general discussion in \cref{sec:weak_hopf_review}, $\mathcal{T} = \text{span}(\{\tau_e, \tau_\sigma\})$ is itself a $C^*$ algebra --- furthermore, from Eq.~\eqref{eq:taus_rules_lee_yang}, it follows that $\tau_e$ is the identity element of $\mathcal{T}$. We can also explicitly construct unitary elements in $\mathcal{T}$. For this, we first note from Eq.~\eqref{eq:taus_rules_lee_yang} that the elements $p, q \in \mathcal{T}$ given by
\begin{align}
    p = \frac{1}{\sqrt{5}}\big(\zeta^2 \tau_e + \tau_\sigma\big),\quad q = \tau_e - p,
\end{align}
satisfy $p^2 = p, q^2 = q, pq = 0, p^* = p$ and $q^* = q$, i.e., $p, q$ can be identified as orthonormal complete set of projectors within $\mathcal{T}$. Furthermore, $p, q$ are also an alternative basis for $\mathcal{T}$, i.e., $\mathcal{T} = \text{span}(\{p, q\}$). Thus, a general unitary element $u \in \mathcal{T}$ will be specified by $\alpha, \beta \in \mathbb{R}$ and will be of the form
\begin{subequations}\label{eq:unitary_element_lee_yang}
\begin{align}
    u = e^{i\alpha} p + e^{i \beta} q = u_e \tau_e + u_\sigma \tau_\sigma,
\end{align}
where 
\begin{align}
    &u_e = \frac{1}{\sqrt{5}}\big({\zeta^2} e^{i\alpha} + \big(\sqrt{5} - \zeta^2\big)e^{i\beta} \big), \nonumber\\
    &u_\sigma = \frac{1}{\sqrt{5}}\big( e^{i\alpha} - e^{i\beta}\big).
\end{align}
\end{subequations}

\prlsection{Matrix Product Unitaries}Finally, we have all the ingredients needed to explicitly construct MPUs from the Lee-Yang Weak Hopf Algebra. We fix an injective representation $\phi$ for $\mathcal{A} = \mathcal{M}(\mathbb{C}^2) \oplus \mathcal{M}(\mathbb{C}^3)$ as the map $\phi: \mathcal{M}(\mathbb{C}^2) \oplus \mathcal{M}(\mathbb{C}^3) \to \mathcal{M}(\mathbb{C}^5)$ defined via
\begin{align}\label{eq:choice_phi}
\phi(x_1 \oplus x_2) = \begin{bmatrix}
    x_1 & 0 \\
    0 & x_2
\end{bmatrix}
\end{align}
As described in \cref{lemma:mpo_irrep}, $\text{MPO}_\phi(\tau_e)$,  $\text{MPO}_\phi(\tau_\sigma)$ are respectively generated by tensors $A_e, A_\sigma$ given by
\begin{align}
A_e = \sum_{k, a, b} \phi(\tilde{e}_{(k; a, b)}) \otimes \psi_e(\tilde{e}^{(k; a, b)}) = \sum_{a, b}  \phi(\tilde{e}_{(1; a, b)}) \otimes \ket{a}\!\bra{b}, \nonumber\\
A_\sigma = \sum_{k, a, b} \phi(\tilde{e}_{(k; a, b)}) \otimes \psi_\sigma(\tilde{e}^{(k; a, b)}) = \sum_{a, b}  \phi(\tilde{e}_{(2; a, b)}) \otimes \ket{a}\!\bra{b},
\end{align}
 and a periodic boundary condition. The first operator in the tensor product in the above equation corresponds to the physical space and has physical dimension $5$, and the second operator corresponds to the bond space and has bond dimension $2$ for $A_e$ and $3$ for $A_\sigma$. We can now use the choice of $\phi$ [Eq.~\eqref{eq:choice_phi}], $\psi_e$,  $\psi_\sigma$ [Eq.~\eqref{eq:irrep_classes_lee_yang}] as well as Eq.~\eqref{eq:new_basis_primal} to obtain explicitly the tensors $A_e, A_\sigma$. The tensor $A_e$, with bond-dimension $=2$ and physical dimension $=5$, is given by
\begin{align*}
    &\begin{array}{c}
        \begin{tikzpicture}[scale=.5,baseline={([yshift=-0.75ex] current bounding box.center)}]
              \foreach \x in {1,...,1}{
                \GTensor{(-\singledx*\x,0)}{1}{.5}{\small $A_e$}{0}
                \myarrow{-\singledx*\x,-.65}{1};
                \myarrow{-\singledx*\x,.85}{1};
                \draw (-\singledx*\x - 0.75*\singledx,0) node {\small $1$};
                \draw (-\singledx*\x + 0.7*\singledx,0) node {\small $1$};
        }
        \end{tikzpicture}
    \end{array} = \ket{1}\!\bra{1},
    \begin{array}{c}
        \begin{tikzpicture}[scale=.5,baseline={([yshift=-0.75ex] current bounding box.center)}]
              \foreach \x in {1,...,1}{
                \GTensor{(-\singledx*\x,0)}{1}{.5}{\small $A_e$}{0}
                \myarrow{-\singledx*\x,-.65}{1};
                \myarrow{-\singledx*\x,.85}{1};
                \draw (-\singledx*\x - 0.75*\singledx,0) node {\small $2$};
                \draw (-\singledx*\x + 0.7*\singledx,0) node {\small $2$};
        }
        \end{tikzpicture}
    \end{array} = \ket{2}\!\bra{2} + \ket{5}\!\bra{5}, \nonumber \\
    &\begin{array}{c}
        \begin{tikzpicture}[scale=.5,baseline={([yshift=-0.75ex] current bounding box.center)}]
              \foreach \x in {1,...,1}{
                \GTensor{(-\singledx*\x,0)}{1}{.5}{\small $A_e$}{0}
                \myarrow{-\singledx*\x,-.65}{1};
                \myarrow{-\singledx*\x,.85}{1};
                \draw (-\singledx*\x - 0.75*\singledx,0) node {\small $1$};
                \draw (-\singledx*\x + 0.7*\singledx,0) node {\small $2$};
        }
        \end{tikzpicture}
    \end{array} = \ket{3}\!\bra{3},
    \begin{array}{c}
        \begin{tikzpicture}[scale=.5,baseline={([yshift=-0.75ex] current bounding box.center)}]
              \foreach \x in {1,...,1}{
                \GTensor{(-\singledx*\x,0)}{1}{.5}{\small $A_e$}{0}
                \myarrow{-\singledx*\x,-.65}{1};
                \myarrow{-\singledx*\x,.85}{1};
                \draw (-\singledx*\x - 0.75*\singledx,0) node {\small $2$};
                \draw (-\singledx*\x + 0.7*\singledx,0) node {\small $1$};
        }
        \end{tikzpicture}
    \end{array} = \ket{4}\!\bra{4}.
\end{align*}
The tensor $A_\sigma$ with bond-dimension $=3$ and physical dimension $=5$ is given by
\begin{align*}
    &\begin{array}{c}
        \begin{tikzpicture}[scale=.5,baseline={([yshift=-0.75ex] current bounding box.center)}]
              \foreach \x in {1,...,1}{
                \GTensor{(-\singledx*\x,0)}{1}{.5}{\small $A_\sigma$}{0}
                \myarrow{-\singledx*\x,-.65}{1};
                \myarrow{-\singledx*\x,.85}{1};
                \draw (-\singledx*\x - 0.75*\singledx,0) node {\small $1$};
                \draw (-\singledx*\x + 0.7*\singledx,0) node {\small $1$};
        }
        \end{tikzpicture}
    \end{array} = \ket{1}\!\bra{2},
    \begin{array}{c}
        \begin{tikzpicture}[scale=.5,baseline={([yshift=-0.75ex] current bounding box.center)}]
              \foreach \x in {1,...,1}{
                \GTensor{(-\singledx*\x,0)}{1}{.5}{\small $A_\sigma$}{0}
                \myarrow{-\singledx*\x,-.65}{1};
                \myarrow{-\singledx*\x,.85}{1};
                \draw (-\singledx*\x - 0.75*\singledx,0) node {\small $2$};
                \draw (-\singledx*\x + 0.7*\singledx,0) node {\small $2$};
        }
        \end{tikzpicture}
    \end{array} = \ket{2}\!\bra{1}, \nonumber \\
    &\begin{array}{c}
        \begin{tikzpicture}[scale=.5,baseline={([yshift=-0.75ex] current bounding box.center)}]
              \foreach \x in {1,...,1}{
                \GTensor{(-\singledx*\x,0)}{1}{.5}{\small $A_\sigma$}{0}
                \myarrow{-\singledx*\x,-.65}{1};
                \myarrow{-\singledx*\x,.85}{1};
                \draw (-\singledx*\x - 0.75*\singledx,0) node {\small $1$};
                \draw (-\singledx*\x + 0.7*\singledx,0) node {\small $2$};
        }
        \end{tikzpicture}
    \end{array} = \ket{3}\!\bra{4},
    \begin{array}{c}
        \begin{tikzpicture}[scale=.5,baseline={([yshift=-0.75ex] current bounding box.center)}]
              \foreach \x in {1,...,1}{
                \GTensor{(-\singledx*\x,0)}{1}{.5}{\small $A_\sigma$}{0}
                \myarrow{-\singledx*\x,-.65}{1};
                \myarrow{-\singledx*\x,.85}{1};
                \draw (-\singledx*\x - 0.75*\singledx,0) node {\small $2$};
                \draw (-\singledx*\x + 0.7*\singledx,0) node {\small $1$};
        }
        \end{tikzpicture}
    \end{array} = \zeta^{-2}\ket{4}\!\bra{3}, \nonumber \\
    &\begin{array}{c}
        \begin{tikzpicture}[scale=.5,baseline={([yshift=-0.75ex] current bounding box.center)}]
              \foreach \x in {1,...,1}{
                \GTensor{(-\singledx*\x,0)}{1}{.5}{\small $A_\sigma$}{0}
                \myarrow{-\singledx*\x,-.65}{1};
                \myarrow{-\singledx*\x,.85}{1};
                \draw (-\singledx*\x - 0.75*\singledx,0) node {\small $1$};
                \draw (-\singledx*\x + 0.7*\singledx,0) node {\small $3$};
        }
        \end{tikzpicture}
    \end{array} = \ket{3}\!\bra{5},
    \begin{array}{c}
        \begin{tikzpicture}[scale=.5,baseline={([yshift=-0.75ex] current bounding box.center)}]
              \foreach \x in {1,...,1}{
                \GTensor{(-\singledx*\x,0)}{1}{.5}{\small $A_\sigma$}{0}
                \myarrow{-\singledx*\x,-.65}{1};
                \myarrow{-\singledx*\x,.85}{1};
                \draw (-\singledx*\x - 0.75*\singledx,0) node {\small $3$};
                \draw (-\singledx*\x + 0.7*\singledx,0) node {\small $1$};
        }
        \end{tikzpicture}
    \end{array} = \zeta \ket{4}\!\bra{5}, \nonumber \\
    &\begin{array}{c}
        \begin{tikzpicture}[scale=.5,baseline={([yshift=-0.75ex] current bounding box.center)}]
              \foreach \x in {1,...,1}{
                \GTensor{(-\singledx*\x,0)}{1}{.5}{\small $A_\sigma$}{0}
                \myarrow{-\singledx*\x,-.65}{1};
                \myarrow{-\singledx*\x,.85}{1};
                \draw (-\singledx*\x - 0.75*\singledx,0) node {\small $2$};
                \draw (-\singledx*\x + 0.7*\singledx,0) node {\small $3$};
        }
        \end{tikzpicture}
    \end{array} = \zeta\ket{5}\!\bra{3},
    \begin{array}{c}
        \begin{tikzpicture}[scale=.5,baseline={([yshift=-0.75ex] current bounding box.center)}]
              \foreach \x in {1,...,1}{
                \GTensor{(-\singledx*\x,0)}{1}{.5}{\small $A_\sigma$}{0}
                \myarrow{-\singledx*\x,-.65}{1};
                \myarrow{-\singledx*\x,.85}{1};
                \draw (-\singledx*\x - 0.75*\singledx,0) node {\small $3$};
                \draw (-\singledx*\x + 0.7*\singledx,0) node {\small $2$};
        }
        \end{tikzpicture}
    \end{array} = \ket{5}\!\bra{4}, \nonumber \\
        &\begin{array}{c}
        \begin{tikzpicture}[scale=.5,baseline={([yshift=-0.75ex] current bounding box.center)}]
              \foreach \x in {1,...,1}{
                \GTensor{(-\singledx*\x,0)}{1}{.5}{\small $A_\sigma$}{0}
                \myarrow{-\singledx*\x,-.65}{1};
                \myarrow{-\singledx*\x,.85}{1};
                \draw (-\singledx*\x - 0.75*\singledx,0) node {\small $3$};
                \draw (-\singledx*\x + 0.7*\singledx,0) node {\small $3$};
        }
        \end{tikzpicture}
    \end{array} = \ket{2}\!\bra{2} - \zeta^{-2}\ket{5}\!\bra{5},
\end{align*}
We can verify by inspecting the tensors $A_e, A_\sigma$ that they are injective as expected from \cref{lemma:mpo_irrep}. 

Finally, we consider the MPU constructed in \cref{prop:mpu_final_hopf}: Using the unitary element in Eq.~\eqref{eq:unitary_element_lee_yang}, we obtain an MPU $U$ with bond-dimension at-most $6$ with the tensor $A$ given by
\begin{align*}
    A^{ij} = \begin{bmatrix}
        \delta_{i, j} & 0 & 0 \\
        0 & A_e^{ij} & 0 \\
        0 & 0 & A_\sigma^{ij}
    \end{bmatrix}
\end{align*}
and boundary $b$ given by
\begin{align*}
    b = \begin{bmatrix}
        1 & 0 & 0 \\
        0 & (u_e - 1)\1_2 & 0 \\
        0 & 0 & u_\sigma \1_3
    \end{bmatrix},
\end{align*}
where $u_e, u_\sigma$ are defined in Eq.~(\ref{eq:unitary_element_lee_yang}b). We remark that this MPU falls in the category satisfying \cref{assumption} and is thus implementable in $\text{poly}(N)$ time using the algorithm described in the main text. 

\end{document}